\documentclass[conference,compsoc]{./IEEEtran}
\usepackage{xcolor}
\usepackage{xspace}
\usepackage{amsmath,amsthm,amssymb}
\usepackage{bbm}
\usepackage{hyperref}

\newdimen\proofrulebreadth \proofrulebreadth=.05em
\newdimen\proofdotseparation \proofdotseparation=1.25ex
\newdimen\proofrulebaseline \proofrulebaseline=2ex
\newcount\proofdotnumber \proofdotnumber=3
\let\then\relax
\def\hfi{\hskip0pt plus.0001fil}
\mathchardef\squigto="3A3B
\newif\ifinsideprooftree\insideprooftreefalse
\newif\ifonleftofproofrule\onleftofproofrulefalse
\newif\ifproofdots\proofdotsfalse
\newif\ifdoubleproof\doubleprooffalse
\let\wereinproofbit\relax
\newdimen\shortenproofleft
\newdimen\shortenproofright
\newdimen\proofbelowshift
\newbox\proofabove
\newbox\proofbelow
\newbox\proofrulename
\def\shiftproofbelow{\let\next\relax\afterassignment\setshiftproofbelow\dimen0 }
\def\shiftproofbelowneg{\def\next{\multiply\dimen0 by-1 }%
\afterassignment\setshiftproofbelow\dimen0 }
\def\setshiftproofbelow{\next\proofbelowshift=\dimen0 }
\def\setproofrulebreadth{\proofrulebreadth}

\def\prooftree{
\ifnum  \lastpenalty=1
\then   \unpenalty
\else   \onleftofproofrulefalse
\fi
\ifonleftofproofrule
\else   \ifinsideprooftree
        \then   \hskip.5em plus1fil
        \fi
\fi
\bgroup
\setbox\proofbelow=\hbox{}\setbox\proofrulename=\hbox{}%
\let\justifies\proofover\let\leadsto\proofoverdots\let\Justifies\proofoverdbl
\let\using\proofusing\let\[\prooftree
\ifinsideprooftree\let\]\endprooftree\fi
\proofdotsfalse\doubleprooffalse
\let\thickness\setproofrulebreadth
\let\shiftright\shiftproofbelow \let\shift\shiftproofbelow
\let\shiftleft\shiftproofbelowneg
\let\ifwasinsideprooftree\ifinsideprooftree
\insideprooftreetrue
\setbox\proofabove=\hbox\bgroup$\displaystyle 
\let\wereinproofbit\prooftree
\shortenproofleft=0pt \shortenproofright=0pt \proofbelowshift=0pt
\onleftofproofruletrue\penalty1
}

\def\eproofbit{
\ifx    \wereinproofbit\prooftree
\then   \ifcase \lastpenalty
        \then   \shortenproofright=0pt  
        \or     \unpenalty\hfil         
        \or     \unpenalty\unskip       
        \else   \shortenproofright=0pt  
        \fi
\fi
\global\dimen0=\shortenproofleft
\global\dimen1=\shortenproofright
\global\dimen2=\proofrulebreadth
\global\dimen3=\proofbelowshift
\global\dimen4=\proofdotseparation
\global\count255=\proofdotnumber
$\egroup  
\shortenproofleft=\dimen0
\shortenproofright=\dimen1
\proofrulebreadth=\dimen2
\proofbelowshift=\dimen3
\proofdotseparation=\dimen4
\proofdotnumber=\count255
}

\def\proofover{
\eproofbit 
\setbox\proofbelow=\hbox\bgroup 
\let\wereinproofbit\proofover
$\displaystyle
}%
\def\proofoverdbl{
\eproofbit 
\doubleprooftrue
\setbox\proofbelow=\hbox\bgroup 
\let\wereinproofbit\proofoverdbl
$\displaystyle
}%
\def\proofoverdots{
\eproofbit 
\proofdotstrue
\setbox\proofbelow=\hbox\bgroup 
\let\wereinproofbit\proofoverdots
$\displaystyle
}%
\def\proofusing{
\eproofbit 
\setbox\proofrulename=\hbox\bgroup 
\let\wereinproofbit\proofusing
\kern0.3em$
}

\def\endprooftree{
\eproofbit 
  \dimen5 =0pt
\dimen0=\wd\proofabove \advance\dimen0-\shortenproofleft
\advance\dimen0-\shortenproofright
\dimen1=.5\dimen0 \advance\dimen1-.5\wd\proofbelow
\dimen4=\dimen1
\advance\dimen1\proofbelowshift \advance\dimen4-\proofbelowshift
\ifdim  \dimen1<0pt
\then   \advance\shortenproofleft\dimen1
        \advance\dimen0-\dimen1
        \dimen1=0pt
        \ifdim  \shortenproofleft<0pt
        \then   \setbox\proofabove=\hbox{%
                        \kern-\shortenproofleft\unhbox\proofabove}%
                \shortenproofleft=0pt
        \fi
\fi
\ifdim  \dimen4<0pt
\then   \advance\shortenproofright\dimen4
        \advance\dimen0-\dimen4
        \dimen4=0pt
\fi
\ifdim  \shortenproofright<\wd\proofrulename
\then   \shortenproofright=\wd\proofrulename
\fi
\dimen2=\shortenproofleft \advance\dimen2 by\dimen1
\dimen3=\shortenproofright\advance\dimen3 by\dimen4
\ifproofdots
\then
        \dimen6=\shortenproofleft \advance\dimen6 .5\dimen0
        \setbox1=\vbox to\proofdotseparation{\vss\hbox{$\cdot$}\vss}%
        \setbox0=\hbox{%
                \advance\dimen6-.5\wd1
                \kern\dimen6
                $\vcenter to\proofdotnumber\proofdotseparation
                        {\leaders\box1\vfill}$%
                \unhbox\proofrulename}%
\else   \dimen6=\fontdimen22\the\textfont2 
        \dimen7=\dimen6
        \advance\dimen6by.5\proofrulebreadth
        \advance\dimen7by-.5\proofrulebreadth
        \setbox0=\hbox{%
                \kern\shortenproofleft
                \ifdoubleproof
                \then   \hbox to\dimen0{%
                        $\mathsurround0pt\mathord=\mkern-6mu%
                        \cleaders\hbox{$\mkern-2mu=\mkern-2mu$}\hfill
                        \mkern-6mu\mathord=$}%
                \else   \vrule height\dimen6 depth-\dimen7 width\dimen0
                \fi
                \unhbox\proofrulename}%
        \ht0=\dimen6 \dp0=-\dimen7
\fi
\let\doll\relax
\ifwasinsideprooftree
\then   \let\VBOX\vbox
\else   \ifmmode\else$\let\doll=$\fi
        \let\VBOX\vcenter
\fi
\VBOX   {\baselineskip\proofrulebaseline \lineskip.2ex
        \expandafter\lineskiplimit\ifproofdots0ex\else-0.6ex\fi
        \hbox   spread\dimen5   {\hfi\unhbox\proofabove\hfi}%
        \hbox{\box0}%
        \hbox   {\kern\dimen2 \box\proofbelow}}\doll%
\global\dimen2=\dimen2
\global\dimen3=\dimen3
\egroup 
\ifonleftofproofrule
\then   \shortenproofleft=\dimen2
\fi
\shortenproofright=\dimen3
\onleftofproofrulefalse
\ifinsideprooftree
\then   \hskip.5em plus 1fil \penalty2
\fi
}

\newcommand{\indrulename}[1]{\textsc{#1}}
\newcommand{\indrule}[3]{
\ensuremath{
\begin{array}{c}
  \prooftree #2
    \justifies #3
    \thickness=0.05em
    \using \indrulename{#1}
  \endprooftree
\end{array}}}

\newcommand{\indruleNPos}[4]{
\begin{array}[#1]{c@{}r}
\hspace{-.2cm}
 #3
\hspace{-.2cm}
\vspace{-.1cm}
\\
& \,#2\!\hspace{-.5cm}\vspace{-.2cm} \\
\cline{1-1}\vspace{-.3cm} \\
  #4 \hspace{.5cm}\,
\end{array}
}
\newcommand{\indruleN}[3]{{\small\indruleNPos{b}{#1}{#2}{#3}}}
\newcommand{\indruleNParen}[3]{{\small\left(\indruleNPos{t}{#1}{#2}{#3}\HS\right)}}

\newcommand{\derivdots}[1]{
\begin{array}[b]{c@{}r}
\vdots 
\\
#1
\end{array}
}

\renewcommand{\theenumi}{\arabic{enumi}}
\renewcommand{\theenumii}{\arabic{enumii}}
\renewcommand{\theenumiii}{\arabic{enumiii}}
\renewcommand{\theenumiv}{\arabic{enumiv}}
\makeatletter
\renewcommand\p@enumii{\theenumi.}
\renewcommand\p@enumiii{\theenumi.\theenumii.}
\renewcommand\p@enumiv{\theenumi.\theenumii.\theenumiii.}
\makeatother

\theoremstyle{break}
\newtheorem{dummythm}{dummythm}
\newtheorem{lemma}[dummythm]{Lemma}

\newtheorem{proposition}[dummythm]{Proposition}
\newtheorem{theorem}[dummythm]{Theorem}
\newtheorem{example}[dummythm]{Example}
\newtheorem{corollary}[dummythm]{Corollary}
\theoremstyle{definition}
\newtheorem{definition}[dummythm]{Definition}
\theoremstyle{remark}
\newtheorem{remark}[dummythm]{Remark}

\newcommand{\llem}[1]{\label{lemma:#1}}
\newcommand{\rlem}[1]{Lem.~\ref{lemma:#1}}
\newcommand{\ldef}[1]{\label{def:#1}}
\newcommand{\rdef}[1]{Def.~\ref{def:#1}}
\newcommand{\lprop}[1]{\label{prop:#1}}
\newcommand{\rprop}[1]{Prop.~\ref{prop:#1}}
\newcommand{\lthm}[1]{\label{thm:#1}}
\newcommand{\rthm}[1]{Thm.~\ref{thm:#1}}

\newcommand{\lcoro}[1]{\label{coro:#1}}
\newcommand{\rcoro}[1]{Coro.~\ref{coro:#1}}
\newcommand{\lsec}[1]{\label{section:#1}}
\newcommand{\rsec}[1]{Section~\ref{section:#1}}

\newcommand{\lexample}[1]{\label{example:#1}}
\newcommand{\rexample}[1]{Ex.~\ref{example:#1}}

\renewcommand{\emptyset}{\varnothing}

\newcommand{\Nat}{\mathbb{N}}

\newcommand{\Hs}{\hspace{.3cm}}
\newcommand{\HS}{\hspace{.5cm}}

\newcommand{\ST}{\ |\ }
\newcommand{\ie}{{\em i.e.}\xspace}
\newcommand{\eg}{{\em e.g.}\xspace}

\newcommand{\ih}{IH\xspace}
\newcommand{\set}[1]{\{#1\}}
\newcommand{\eqdef}{\,\mathrel{\overset{\mathrm{def}}{=}}\,}

\newcommand{\symdiff}{\triangle}

\newcommand{\sub}[2]{[#1\!:=\!#2]}

\newcommand{\fv}[1]{\mathsf{fv}(#1)}

\newcommand{\under}{\underline{\,\,\,}}

\newcommand{\imp}{\Rightarrow}
\newcommand{\Bot}{\bot}

\newcommand{\lam}[2]{\lambda#1.\,#2}

\newcommand{\var}{x}
\newcommand{\vartwo}{y}
\newcommand{\varthree}{z}

\newcommand{\tm}{t}
\newcommand{\tmtwo}{s}
\newcommand{\tmthree}{u}

\newcommand{\btyp}{\alpha}
\newcommand{\btyptwo}{\beta}
\newcommand{\btypthree}{\gamma}

\newcommand{\typ}{A}
\newcommand{\typtwo}{B}
\newcommand{\typthree}{C}

\newcommand{\ev}{P}
\newcommand{\evtwo}{Q}
\newcommand{\evthree}{R}

\newcommand{\PP}{{}^\oplus}
\newcommand{\NN}{{}^\ominus}
\newcommand{\pp}{{}^+}
\newcommand{\nn}{{}^-}
\newcommand{\OP}{{}^{\sim}}

\newcommand{\emptyctx}{\emptyset}
\newcommand{\tctx}{\Gamma}
\newcommand{\tctxtwo}{\Delta}

\newcommand{\rulename}[1]{\indrulename{\text{#1}}}

\newcommand{\PRK}{\textsc{prk}}
\newcommand{\lambdaC}{\lambda^{\PRK}}
\newcommand{\lambdaCeta}{\lambda^{\PRK}_\eta}
\newcommand{\classem}[1]{c(#1)}

\newcommand{\strongabssym}[1]{{\RHD\!\!\!\LHD_{#1}}}
\newcommand{\strongabs}[3]{#2 \mathrel{\strongabssym{#1}} #3}
\newcommand{\strongabstable}[3]{
  \begin{array}[t]{l@{}l}
                    & #2 \\
  \strongabssym{#1} & #3
  \end{array}
}
\newcommand{\abs}[3]{#2 \mathrel{\bowtie_{#1}} #3}

\newcommand{\trivsym}{\star}
\newcommand{\triv}{\trivsym}

\newcommand{\trivF}{\trivsym}

\newcommand{\abortsym}{\mathcal{E}}
\newcommand{\abort}[2]{\abortsym_{#1}(#2)}

\newcommand{\abortF}[2]{\abortsym_{#1}(#2)}

\newcommand{\pair}[2]{\langle#1,#2\rangle}
\newcommand{\pairp}[2]{\pair{#1}{#2}\pp}
\newcommand{\pairn}[2]{\pair{#1}{#2}\nn}
\newcommand{\pairpn}[2]{\pair{#1}{#2}^{\pm}}
\newcommand{\pairnp}[2]{\pair{#1}{#2}^{\mp}}
\newcommand{\pairc}[2]{\pair{#1}{#2}^\clasym}
\newcommand{\pairF}[2]{\pair{#1}{#2}}
\newcommand{\projisym}[1][i]{\pi_{#1}}

\newcommand{\projip}[2][i]{\projisym[#1]^+(#2)}
\newcommand{\projin}[2][i]{\projisym[#1]^-(#2)}
\newcommand{\projipn}[2][i]{\projisym[#1]^{\pm}(#2)}
\newcommand{\projic}[2][i]{\projisym[#1]^{\clasym}(#2)}
\newcommand{\projiF}[2][i]{\projisym[#1](#2)}

\newcommand{\inisym}[1][i]{\mathsf{in}_{#1}}

\newcommand{\inip}[2][i]{\inisym[#1]\!\!\pp(#2)}
\newcommand{\inin}[2][i]{\inisym[#1]\!\!\nn(#2)}
\newcommand{\inipn}[2][i]{\inisym[#1]^{\pm}(#2)}
\newcommand{\ininp}[2][i]{\inisym[#1]^{\mp}(#2)}
\newcommand{\inic}[2][i]{\inisym[#1]^{\clasym}(#2)}
\newcommand{\iniF}[2][i]{\inisym[#1](#2)}
\newcommand{\casesym}{\delta}
\newcommand{\caseto}{.}

\newcommand{\casep}[5]{\casesym\pp#1\,[_{#2} \caseto #3][_{#4} \caseto #5]}

\newcommand{\caseptablex}[5]{\begin{array}[t]{l@{\ }l}
  \casesym\pp & #1 \\
              & [_{#2} \caseto #3] \\
              & [_{#4} \caseto #5] \\
  \end{array}
}
\newcommand{\casen}[5]{\casesym\nn#1\,[_{#2} \caseto #3][_{#4} \caseto #5]}
\newcommand{\casepn}[5]{\casesym^{\pm}#1\,[_{#2} \caseto #3][_{#4} \caseto #5]}
\newcommand{\casec}[5]{\casesym^{\clasym}#1\,[_{#2} \caseto #3][_{#4} \caseto #5]}
\newcommand{\caseF}[5]{\casesym#1\,[_{#2} \caseto #3][_{#4} \caseto #5]}
\newcommand{\caseFtable}[5]{\begin{array}[t]{l}\casesym#1 \\
    \HS[_{#2} \caseto #3] \\
    \HS[_{#4} \caseto #5] \\
  \end{array}
}
\newcommand{\caseFtablex}[5]{\begin{array}[t]{l@{}l}
    \casesym#1 & [_{#2} \caseto #3] \\
               & [_{#4} \caseto #5] \\
  \end{array}
}

\newcommand{\negisym}{\nu}
\newcommand{\negip}[1]{\negisym\pp#1}
\newcommand{\negin}[1]{\negisym\nn#1}
\newcommand{\negipn}[1]{\negisym^{\pm}#1}
\newcommand{\neginp}[1]{\negisym^{\mp}#1}
\newcommand{\negesym}{\mu}
\newcommand{\negep}[1]{\negesym\pp#1}
\newcommand{\negen}[1]{\negesym\nn#1}
\newcommand{\negepn}[1]{\negesym^{\pm}#1}
\newcommand{\neglamsym}{\Lambda}
\newcommand{\negapsym}{\texttt{\textup{\#}}}
\newcommand{\neglamc}[2]{\neglamsym^{\clasym}_{#1}.\,#2}
\newcommand{\negapc}[2]{#1 \negapsym^{\clasym} #2}

\newcommand{\clasym}{\mathcal{C}}

\newcommand{\clasapsym}{\bullet}
\newcommand{\claslamp}[2]{\mathsf{IC}\pp_{#1}.\,#2}
\newcommand{\claslamn}[2]{\mathsf{IC}\nn_{#1}.\,#2}
\newcommand{\claslampn}[2]{\mathsf{IC}^{\pm}_{#1}.\,#2}

\newcommand{\claslamptable}[2]{
  \begin{array}[t]{l}
  \mathsf{IC}\pp_{#1}.\\
  \ #2
  \end{array}
}
\newcommand{\clasapp}[2]{#1 \clasapsym\!\!\pp\, #2}
\newcommand{\clasapn}[2]{#1 \clasapsym\!\!\nn\, #2}
\newcommand{\clasappn}[2]{#1 \clasapsym\!\!^{\pm}\, #2}
\newcommand{\clasapptable}[2]{
  \left(\begin{array}{l@{}l}
                      & #1 \\
  \clasapsym\!\!\pp\, & #2
  \end{array}\right)
}

\newcommand{\Ax}{\rulename{Ax}}
\newcommand{\Abs}{\rulename{Abs}}

\newcommand{\Iandp}{\rulename{I$\land\pp$}}
\newcommand{\Iorn}{\rulename{I$\lor\nn$}}

\newcommand{\Eandp}[1][i]{\rulename{E$\land\pp_{#1}$}}
\newcommand{\Eorn}[1][i]{\rulename{E$\lor\nn_{#1}$}}

\newcommand{\Iorp}[1][i]{\rulename{I$\lor\pp_{#1}$}}
\newcommand{\Iandn}[1][i]{\rulename{I$\land\nn_{#1}$}}

\newcommand{\Eorp}{\rulename{E$\lor\pp$}}
\newcommand{\Eandn}{\rulename{E$\land\nn$}}

\newcommand{\Inotp}{\rulename{I$\lnot\pp$}}
\newcommand{\Inotn}{\rulename{I$\lnot\nn$}}

\newcommand{\Enotp}{\rulename{E$\lnot\pp$}}
\newcommand{\Enotn}{\rulename{E$\lnot\nn$}}

\newcommand{\Icp}{\rulename{IC$\pp$}}
\newcommand{\Icn}{\rulename{IC$\nn$}}

\newcommand{\Ecp}{\rulename{EC$\pp$}}
\newcommand{\Ecn}{\rulename{EC$\nn$}}

\newcommand{\Iallp}{\rulename{I$\forall\pp$}}

\newcommand{\Eallp}{\rulename{E$\forall\pp$}}

\newcommand{\Weakening}{\rulename{W}}
\newcommand{\Cut}{\rulename{Cut}}
\newcommand{\Contrapose}{\rulename{Contra}}
\newcommand{\contrapose}[3]{{\updownarrow_{#1}^{#2}(#3)}}
\newcommand{\ClStr}{\rulename{CS}}
\newcommand{\PrCon}{\rulename{PC}}

\newcommand{\lemN}[1]{\pitchfork^-_{#1}}
\newcommand{\lemP}[1]{\pitchfork^+_{#1}}
\newcommand{\lemC}[1]{\pitchfork^{\clasym}_{#1}}
\newcommand{\lemNinner}[2]{\Delta^-_{#1,#2}}
\newcommand{\lemPinner}[2]{\Delta^+_{#1,#2}}

\newcommand{\rewritingRuleName}[1]{\mathsf{#1}}

\newcommand{\toa}[1]{\xrightarrow{#1}}
\newcommand{\rtoa}[1]{\mathrel{\xrightarrow{#1}\!\!{}^*\,\,}}
\newcommand{\ptoa}[1]{\mathrel{\xrightarrow{#1}\!\!{}^+\,\,}}
\newcommand{\rto}{\mathrel{\rightarrow^*}}

\newcommand{\ruleProj}{\rewritingRuleName{proj}}
\newcommand{\ruleCase}{\rewritingRuleName{case}}
\newcommand{\ruleNeg}{\rewritingRuleName{neg}}
\newcommand{\ruleBeta}{\rewritingRuleName{beta}}
\newcommand{\ruleAbsPairInj}{\rewritingRuleName{absPairInj}}
\newcommand{\ruleAbsInjPair}{\rewritingRuleName{absInjPair}}
\newcommand{\ruleAbsNeg}{\rewritingRuleName{absNeg}}
\newcommand{\ruleEta}{\rewritingRuleName{eta}}
\newcommand{\ruleAnon}{\rewritingRuleName{r}}

\newcommand{\ctxhole}{\Box}
\newcommand{\ctxof}[1]{\langle#1\rangle}
\newcommand{\gctx}{\mathtt{C}}

\newcommand{\gctxof}[1]{\gctx\ctxof{#1}}

\newcommand{\elctx}{\mathtt{E}}
\newcommand{\elctxof}[1]{\elctx\ctxof{#1}}
\newcommand{\casectx}{\mathtt{K}}
\newcommand{\casectxof}[1]{\casectx\ctxof{#1}}

\newcommand{\nf}{N}
\newcommand{\neu}{S}

\newcommand{\typeConstraints}{\mathcal{C}}
\newcommand{\typeConstraintsPosNeg}{\mathcal{C}_{\mathbf{pn}}}
\newcommand{\posvars}[1]{\textsf{p}(#1)}
\newcommand{\negvars}[1]{\textsf{n}(#1)}
\newcommand{\wposvars}[1]{\textsf{p}^{\mathtt{w}}(#1)}
\newcommand{\wnegvars}[1]{\textsf{n}^{\mathtt{w}}(#1)}

\newcommand{\Pos}[2]{\mathbf{p}_{#1,#2}}
\newcommand{\Neg}[2]{\mathbf{n}_{#1,#2}}
\newcommand{\semF}[1]{[\![#1]\!]}
\newcommand{\tunit}{\mathbf{1}}
\newcommand{\tzero}{\mathbf{0}}
\newcommand{\funabsF}[2]{\mathsf{abs}^{#1}_{#2}}

\newcommand{\compl}[1]{||#1||}

\newcommand{\derivation}{\pi}
\newcommand{\derivationtwo}{\xi}
\newcommand{\derivationthree}{\rho}
\newcommand{\typingDerivation}{\pi}

\newcommand{\all}[2]{\forall #1.\,#2}

\newcommand{\lamtp}[2]{\lambda\pp #1.\,#2}

\newcommand{\lamc}[2]{\lambda^{\clasym}_{#1}.\,#2}

\newcommand{\apptp}[2]{#1 \bullet\!\!\pp #2}

\newcommand{\appc}[2]{#1 \,\texttt{\textup{@}}^{\clasym}\, #2}

\newcommand{\kripmod}{\mathcal{M}}
\newcommand{\kripstruct}{(\wlset,\wleq,\wlpos{},\wlneg{})}
\newcommand{\wlset}{\mathcal{W}}
\newcommand{\wl}{w}
\newcommand{\wleq}{\mathrel{\leq}}
\newcommand{\wgeq}{\mathrel{\geq}}
\newcommand{\wlint}[1]{\mathcal{V}_{#1}}
\newcommand{\wlpos}[1]{\mathcal{V}^+_{#1}}
\newcommand{\wlneg}[1]{\mathcal{V}^-_{#1}}
\newcommand{\kripforce}[3][\kripmod]{#1,#2 \Vdash #3}
\newcommand{\kripnotforce}[3][\kripmod]{#1,#2 \nVdash #3}
\newcommand{\kripforcefull}[3][\kripmod]{#1,#2 \Vvdash #3}
\newcommand{\kripentails}[2]{#1 \Vvdash #2}

\usepackage{wasysym}
\newcommand{\trunc}[1]{\ocircle{#1}}

\colorlet{darkblue}{blue!60!black}
\newcommand{\SeeAppendix}[1]{\textcolor{darkblue}{[#1]}}

\newcounter{alphasect}
\def\alphainsection{0}

\let\oldsection=\section
\def\section{%
  \ifnum\alphainsection=1%
    \addtocounter{alphasect}{1}%
  \fi%
\oldsection}%

\renewcommand\thesection{%
  \ifnum\alphainsection=1%
    \Alph{alphasect}%
  \else%
    \arabic{section}%
  \fi%
}%

\newenvironment{alphasection}{%
  \ifnum\alphainsection=1%
    \errhelp={Let other blocks end at the beginning of the next block.}
    \errmessage{Nested Alpha section not allowed}
  \fi%
  \setcounter{alphasect}{0}
  \def\alphainsection{1}
}{%
  \setcounter{alphasect}{0}
  \def\alphainsection{0}
}%

\newcommand{\Case}[1]{\smallskip\noindent$\bullet$\,#1}

\newcommand{\IMP}{\Rightarrow}
\newcommand{\extwith}[1]{\langle#1\rangle}

\ifCLASSOPTIONcompsoc
  \usepackage[nocompress]{cite}
\else
  \usepackage{cite}
\fi

\ifCLASSINFOpdf
\else
\fi
\begin{document}
%
\title{
  A Constructive Logic with Classical Proofs and Refutations\\
  (Extended Version)
}

\author{\IEEEauthorblockN{Pablo Barenbaum}
\IEEEauthorblockA{
Departamento de Computación, \\
Facultad de Ciencias Exactas y Naturales, \\
Universidad de Buenos Aires, Argentina.\\
Universidad Nacional de Quilmes, Argentina.\\
Email: pbarenbaum@dc.uba.ar}
\and
\IEEEauthorblockN{Teodoro Freund}
\IEEEauthorblockA{
Departamento de Computación, \\
Facultad de Ciencias Exactas y Naturales, \\
Universidad de Buenos Aires, Argentina.\\
Email: tfreund95@gmail.com}
}


%


\maketitle

\begin{abstract}
We study a conservative extension of classical propositional logic
distinguishing between four modes of statement:
a proposition may be affirmed or denied, and it may be strong or classical.
Proofs of strong propositions must be constructive in some sense,
whereas proofs of classical propositions proceed by contradiction.
The system, in natural deduction style, is shown to be sound and complete
with respect to a Kripke semantics.
We develop the system from the perspective of the propositions-as-types correspondence
by deriving a term assignment system with confluent reduction.
The proof of strong normalization relies on a translation to System F with Mendler-style recursion.
\end{abstract}


%
\IEEEpeerreviewmaketitle

\section{Introduction}

Intuitionistic logic was born out of Brouwer's remark that
the {\em law of excluded middle} ($\typ \lor \neg\typ$)
allows one to prove propositions in a seemingly non-constructive way.
But what constitutes a {\em constructive} proof, exactly?
A possible answer to this question may be found in
the {\em realizability} interpretation,
also known as the Brouwer--Heyting--Kolmogorov interpretation,
which establishes what kinds of mathematical constructions
can be regarded as a {\em realizer} or {\em canonical proof}
of a proposition.
For example, a canonical proof of a conjunction $(\typ\land\typtwo)$
is given by a pair $\langle{p,q}\rangle$,
where $p$ and $q$ are in turn canonical proofs of $\typ$ and $\typtwo$
respectively.
From the works of Gentzen~\cite{gentzen1935untersuchungen}
and Prawitz~\cite{prawitz1965natural} we know that,
in intuitionistic natural deduction, an arbitrary proof of a proposition $\typ$
can always be {\em normalized} to a canonical proof of~$\typ$.

These ideas culminate in the {\em propositions-as-types} correspondence,
the realization that a proposition $\typ$ may be understood
as a {\em type} that expresses the specification of a program.
A proof $p$ of $\typ$ may be understood as a program fulfilling
the specification $\typ$.
Running the program corresponds to applying a computational procedure
that normalizes the proof $p$ to obtain a canonical proof $p'$ of~$\typ$.
Under this paradigm, proofs in intuitionistic natural deduction
can be identified with programs in the {\em simply typed $\lambda$-calculus}.
This correspondence has been extended to encompass many logical systems,
including first-order~\cite{debruijn1970mathematical,martinlof1971theory}
and second-order intuitionistic logic~\cite{thesisgirard,reynolds1974towards},
linear logic~\cite{girard1987linear},
classical logic~\cite{griffin1989formulae,Curien00theduality,symmetric-Barbanera-berardi,lambdamu-parigot},
and modal logic~\cite{bierman2000intuitionistic,DBLP:journals/jacm/DaviesP01}.
These developments unveil the deep connection
between logic and computer science, and they have practical applications
in the development of programming languages and proof assistants
based in type theory such as \textsc{Coq} and \textsc{Agda}.

In this paper, we define a logical system $\PRK$,
presented in natural deduction style,
that distinguishes between four ``modes'' of stating a proposition $\typ$,
which are written $\typ\pp$ (strong affirmation), $\typ\nn$ (strong denial),
$\typ\PP$ (classical affirmation), and $\typ\NN$ (classical denial).
As the name implies, strong affirmation is stronger than classical affirmation,
\ie from $\typ\pp$ one may deduce $\typ\PP$,
and likewise from $\typ\nn$ one may deduce $\typ\NN$.
Affirmation and denial are contradictory,
\ie from $\typ\pp$ and $\typ\nn$ one may derive any conclusion,
and similarly for $\typ\PP$ and $\typ\NN$.
This logic turns out to be a {\em conservative extension}
of classical propositional logic,
in the sense that a proposition $\typ$ is classically valid
if and only if $\typ\PP$ is valid in $\PRK$.

System $\PRK$ is then shown to be {\em sound} and {\em complete}
with respect to a Kripke-style semantics.
This helps to elucidate the difference between the four modes of statement.
In particular, strong affirmation and denial have a constructive
``flavor''---for example, the law of excluded middle holds
classically, \ie $(\typ\lor\neg\typ)\PP$ is valid,
but it does not hold strongly, \ie $(\typ\lor\neg\typ)\pp$ is not valid.

Furthermore, following the propositions-as-types paradigm,
we derive an associated calculus $\lambdaC$ and we show that
it enjoys the expected meta-theoretical properties:
{\em confluence}, {\em subject reduction} and {\em strong normalization}.
Besides, we characterize the set of {\em normal forms}.
This sheds a new light on the structure of classical proofs,
and it may form the basis of the type systems for future programming
languages and proof assistants.

{\bf Classical Proofs and Refutations.}
It is well-known that
intuitionistic propositional logic enjoys the {\em disjunctive property},
that is, a canonical proof of a disjunction $(\typ \lor \typtwo)$
is given by either a canonical proof of $\typ$
or a canonical proof of $\typtwo$.
In particular, the proof {\em contains one bit of information},
indicating whether it encloses a proof of $\typ$ or of $\typtwo$.
In contrast, the intuitionistic notion of {\em refutation} (proof of a negation)
is not dual to the notion of proof.
For example,
the set of refutations of a conjunction
$(\typ \land \typtwo)$ is {\em not} the disjoint union
of the set of refutations of $\typ$ and the set of refutations of $\typtwo$.
This is related to the fact that one of De~Morgan's laws, namely
$\neg(\typ \land \typtwo) \to (\neg\typ \lor \neg\typtwo)$,
which is classically valid, does not hold intuitionistically.
The reason is that the proof of a negation in intuitionistic logic
proceeds by contradiction,
\ie the equivalence $\neg\typ \equiv (\typ \to \bot)$ holds.
As a matter of fact, a refutation in intuitionistic logic
{\em contains no information}\footnote{As attested for example by
the fact that, in {\em homotopy type theory}, a type of the form $\neg\typ$
can always be shown to be a mere proposition,
\ie if it is inhabited, it is equivalent to the unit type;
see for instance~\cite[Section~3.6]{hottbook}.}.

The attempt to recover the symmetry between the notions of
proof and refutation in a constructive setting lead Nelson
to study logical systems with {\em strong negation}~\cite{nelson1949constructible}.
One way to formulate Nelson's system is to distinguish
between two modes to state a proposition $\typ$,
that we may call affirmation~($\typ\pp$) and denial~($\typ\nn$),
whose witnesses are called {\em proofs} and {\em refutations} of $\typ$
respectively\footnote{These are called ``P-realizers'' and ``N-realizers'' by Nelson.}.
The following (informal) equations suggest
a realizability interpretation
for affirmations and denials of
conjunction, disjunction and negation,
found in Nelson's system:
\[
{\small
  \begin{array}{rcl@{\hspace{.5cm}}rcl}
    (\typ\land\typtwo)\pp & \approx & \typ\pp \times \typtwo\pp
  &
    (\typ\land\typtwo)\nn & \approx & \typ\nn \uplus \typtwo\nn
  \\
    (\typ\lor\typtwo)\pp & \approx & \typ\pp \uplus \typtwo\pp
  &
    (\typ\lor\typtwo)\nn & \approx & \typ\nn \times \typtwo\nn
  \\
    (\neg\typ)\pp & \approx & \typ\nn
  &
    (\neg\typ)\nn & \approx & \typ\pp
  \end{array}
}
\]
These equations state, for example,
that the set of proofs of a conjunction $(\typ\land\typtwo)$
is the cartesian product of the set of proofs of $\typ$
and the set of proofs of $\typtwo$,
while the set of refutations of a conjunction $(\typ\land\typtwo)$ is
the disjoint union of the set of refutations of $\typ$
and the set of refutations of $\typtwo$.

This paper was conceived with the goal in mind of providing a
{\bf realizability interpretation for classical logic} based
on this strong notion of negation.
Long-established embeddings of classical logic into intuitionistic logic,
such as G\"odel's, are based on {\em double-negation translations}.
These translations rely on the equivalence $\typ \equiv \neg\neg\typ$,
which is classically, but not intuitionistically, valid.

Our starting point is a different equivalence, namely $\typ \equiv (\neg\typ \to \typ)$,
which is again classically, but not intuitionistically, valid.
To formulate the interpretation, we introduce a further distinction,
according to which a proposition~$\typ$ may be qualified
as {\em strong} or {\em classical}, resulting in four possible modes:
\[
  \begin{array}{l|ll}
    & \text{affirmation} & \text{denial} \\
  \hline
  \text{strong}    & \typ\pp & \typ\nn \\
  \text{classical} & \typ\PP & \typ\NN \\
  \end{array}
\]
As before, the witness of an affirmation (resp. denial)
is called a proof (resp. refutation).
Our first intuition is that a classical proof of a proposition $\typ$
should be given by a construction that
transforms a {\em strong} refutation of $\typ$ into
a strong proof of $\typ$.
Hence, informally speaking,
the realizability interpretation should include an equation
like $\typ\PP \approx (\typ\nn \to \typ\pp)$,
where $(X \to Y)$ is expected to denote the set of ``transformations'',
from $X$ to $Y$ in some suitable sense.

In a preliminary version of this work, we explored a
realizability interpretation based on such an equation,
and its dual equation, $\typ\NN \approx (\typ\pp \to \typ\nn)$.
But, unfortunately, we were not able to formulate
a well-behaved system from the computational point of view\footnote{The
difficulty is that it is not obvious how to normalize a proof of falsity
derived from $\typ\PP$ and $\typ\NN$, \ie a contradiction obtained from
combining a classical proof and a classical refutation of a proposition $\typ$.}.
The study of proof normalization suggests that the ``right'' equations should
instead be
$\typ\PP \approx (\typ\NN \to \typ\pp)$ and its dual,
$\typ\NN \approx (\typ\PP \to \typ\nn)$.
This means that a classical proof of a proposition $\typ$
should be given by a transformation that
takes a {\em classical} refutation of $\typ$ as an input
and produces a strong proof of $\typ$ as an output.
This is indeed the path that we follow.

The complete set of equations that suggest the realizability
interpretation that we study in this paper is:
\[
{\small
  \begin{array}{rcl@{\hspace{.5cm}}rcl}
    (\typ\land\typtwo)\pp & \approx & \typ\PP \times \typtwo\PP
  &
    (\typ\land\typtwo)\nn & \approx & \typ\NN \uplus \typtwo\NN
  \\
    (\typ\lor\typtwo)\pp & \approx & \typ\PP \uplus \typtwo\PP
  &
    (\typ\lor\typtwo)\nn & \approx & \typ\NN \times \typtwo\NN
  \\
    (\neg\typ)\pp & \approx & \typ\NN
  &
    (\neg\typ)\nn & \approx & \typ\PP
  \\
    \typ\PP & \approx & \typ\NN \to \typ\pp
  &
    \typ\NN & \approx & \typ\PP \to \typ\nn
  \\
  \end{array}
}
\]
Observe that a strong proof of a conjunction
is given by a pair of {\em classical} (and not strong) proofs.
Similarly for the other connectives,
\eg a strong refutation of $\neg\typ$ is given by a
{\em classical} (and not a strong) proof of $\typ$.
For the sake of brevity, in this paper we will only consider
three logical connectives: conjunction, disjunction, and negation.
Extending our results and techniques to incorporate other
propositional connectives,
such as implication, and truth and falsity constants, should not
present major obstacles.

One technical difficulty that we confront is the fact that
the last two equations are mutually recursive.
This means, in particular, that these equations cannot be understood
as a translation from formulae of $\PRK$ to formulae of other systems
(such as the simply typed $\lambda$-calculus), at least not in the
naive sense. However, as we shall see, these recursive equations
do fulfill Mendler's {\em positivity requirement}~\cite{mendler1991inductive},
which allows us to give a translation from $\lambdaC$
to System~F extended with (non-strictly positive) recursive type constraints.
\smallskip

{\bf Structure of This Paper.}
The remainder of this paper is organized as follows.
In \rsec{prk_natural_deduction} ({\bf Natural Deduction})
we present the proof system $\PRK$ in natural deduction style,
and we study some basic facts,
such as weakening and substitution.
In \rsec{prk_kripke_semantics} ({\bf Kripke Semantics})
we define an ad~hoc notion of Kripke model,
and we show that $\PRK$ is sound and complete
with respect to this Kripke semantics, \ie, a sequent $\tctx \vdash \typ$
is provable in $\PRK$ if and only if it holds in every Kripke model.
In \rsec{prk_lambdaC} ({\bf The $\lambdaC$-calculus})
we derive a term assignment for $\PRK$, and we endow it
with a small-step reduction semantics.
We show that the system is confluent and that it enjoys subject reduction.
To show that $\lambdaC$ is strongly normalizing, we rely on the
aforementioned translation to System~F extended with recursive type constraints.
We also provide an inductive characterization of the set of normal forms,
and we show that an extensionality rule akin to $\eta$-reduction may
be incorporated to the system.
In \rsec{prk_classical_logic} ({\bf Relation with Classical Logic})
we show that $\PRK$ is a conservative extension of
classical logic. We show how this provides a new computational
interpretation for classical logic.
In \rsec{prk_conclusion} ({\bf Conclusion}) we conclude, and we
discuss related and future work.

\section{The Natural Deduction System $\PRK$}
\lsec{prk_natural_deduction}

In this section, we define the logical system $\PRK$,
formulated in natural deduction style~(\rdef{system_prk}).
We then prove that some typical reasoning principles,
namely weakening, cut, and substitution,
as well as some principles specific to this system,
are admissible in $\PRK$~(\rlem{admissible_rules_logic}).
An important result in this section is the
projection lemma~(\rlem{projection_lemma}).
We also formulate an explicit duality principle~(\rlem{duality_principle}).
\smallskip

We suppose given a denumerable set
of {\em propositional variables} $\btyp,\btyptwo,\btypthree,\hdots$.
The set of {\em pure propositions} is given by the abstract syntax:
\[
\begin{array}{rrll}
  \typ,\typtwo,\typthree,\hdots
    & ::= & \btyp
    & \text{propositional variable}
  \\
    & \mid & \typ \land \typtwo
    & \text{conjunction}
  \\
    & \mid & \typ \lor \typtwo
    & \text{disjunction}
  \\
    & \mid & \neg\typ
    & \text{negation}
  \\
\end{array}
\]
The set of {\em moded propositions} (or just {\em propositions})
is given by the abstract syntax:
\[
\begin{array}{rrll}
  \ev,\evtwo,\evthree,\hdots ::=
    & ::=  & \typ\pp & \text{strong affirmation}
  \\
    & \mid & \typ\nn & \text{strong denial}
  \\
    & \mid & \typ\PP & \text{classical afirmation}
  \\
    & \mid & \typ\NN & \text{classical denial}
\end{array}
\]
As mentioned in the introduction, propositions are classified into
four modes, which arise from discriminating two dimensions.
The first dimension (called {\em sign}) distinguishes between
{\em affirmations} ($\typ\pp$ and $\typ\PP$)
and {\em denials} ($\typ\nn$ and $\typ\NN$),
sometimes also called {\em positive} and {\em negative} propositions.
The second dimension (called {\em strength})
distinguishes between
{\em strong propositions} ($\typ\pp$ and $\typ\nn$)
and {\em classical propositions} ($\typ\PP$ and $\typ\NN$).
Note that modes cannot be nested,
\eg $(\typ\pp \land \typtwo\pp)\nn$
is not a well-formed proposition.

The {\em opposite proposition} $\ev\OP$ of a given proposition $\ev$
is defined by flipping the sign, but preserving the strength:
\[
  \begin{array}{rcl@{\hspace{1cm}}rcl}
    (\typ\pp)\OP & \eqdef & \typ\nn
  &
    (\typ\nn)\OP & \eqdef & \typ\pp
  \\
    (\typ\PP)\OP & \eqdef & \typ\NN
  &
    (\typ\NN)\OP & \eqdef & \typ\PP
  \end{array}
\]
The {\em classical projection} of a given proposition $\ev$
is written $\trunc{\ev}$ and defined by preserving the sign
and making the strength classical:
\[
  \begin{array}{rcl@{\HS}rcl}
    \trunc{(\typ\pp)} & \eqdef & \typ\PP &
    \trunc{(\typ\nn)} & \eqdef & \typ\NN \\
    \trunc{(\typ\PP)} & \eqdef & \typ\PP &
    \trunc{(\typ\NN)} & \eqdef & \typ\NN \\
  \end{array}
\]
Note that $\ev\OP\OP = \ev$, $\trunc{\trunc{\ev}} = \trunc{\ev}$,
and $\trunc{(\ev\OP)} = (\trunc{\ev})\OP$.

\begin{definition}[System $\PRK$]
\ldef{system_prk}
Judgments in $\PRK$ are of the form $\tctx \vdash \ev$,
where $\tctx$ is a finite {\em set} of moded propositions,
\ie we work implicitly up to structural rules of
contraction and exchange.
Derivability of judgments is defined inductively by the 
following inference schemes.

Except for the first two rules,
the system is defined following the realizability interpretation
of propositions discussed in the introduction. For instance,
rules $\Iandp$ and $\Eandp$ embody the equation for the
strong affirmation of a conjunction,
$(\typ\land\typtwo)\pp \approx (\typ\PP \times \typtwo\PP)$.

\[
{\small
\indrule{\Ax}{
}{
  \tctx,\ev \vdash \ev
}
\indrule{\Abs}{
  \tctx \vdash \ev
  \HS
  \tctx \vdash \ev\OP
  \HS
  \text{$\ev$ strong}
}{
  \tctx \vdash\evtwo
}
}
\]
\[
{\small
\indrule{\Iandp}{
  \tctx \vdash \typ\PP
  \HS
  \tctx \vdash \typtwo\PP
}{
  \tctx \vdash (\typ \land \typtwo)\pp
}
\indrule{\Iorn}{
  \tctx \vdash \typ\NN
  \HS
  \tctx \vdash \typtwo\NN
}{
  \tctx \vdash (\typ \lor \typtwo)\nn
}
}
\]
\[
{\small
\indrule{\Eandp}{
  \tctx \vdash (\typ_1 \land \typ_2)\pp
  \HS i \in \set{1, 2}
}{
  \tctx \vdash \typ_i\PP
}
}
\]
\[
{\small
\indrule{\Eorn}{
  \tctx \vdash (\typ_1 \lor \typ_2)\nn
  \HS i \in \set{1, 2}
}{
  \tctx \vdash \typ_i\NN
}
}
\]
\[
{\small
\indrule{\Iorp}{
  \tctx \vdash \typ_i\PP
  \HS i \in \set{1, 2}
}{
  \tctx \vdash (\typ_1 \lor \typ_2)\pp
}
\indrule{\Iandn}{
  \tctx \vdash \typ_i\NN
  \HS i \in \set{1, 2}
}{
  \tctx \vdash (\typ_1 \land \typ_2)\nn
}
}
\]
\[
{\small
\indrule{\Eorp}{
  \tctx \vdash (\typ \lor \typtwo)\pp
  \HS
  \tctx, \typ\PP \vdash \ev
  \HS
  \tctx, \typtwo\PP \vdash \ev
}{
  \tctx \vdash \ev
}
}
  \] %
  \[ %
{\small
\indrule{\Eandn}{
  \tctx \vdash (\typ \land \typtwo)\nn
  \HS
  \tctx, \typ\NN \vdash \ev
  \HS
  \tctx, \typtwo\NN \vdash \ev
}{
  \tctx \vdash \ev
}
}
\]
\[
{\small
\indrule{\Inotp}{
  \tctx \vdash \typ\NN
}{
  \tctx \vdash (\neg\typ)\pp
}
\indrule{\Inotn}{
  \tctx \vdash \typ\PP
}{
  \tctx \vdash (\neg\typ)\nn
}
}
\]
\[
{\small
\indrule{\Enotp}{
  \tctx \vdash (\neg\typ)\pp
}{
  \tctx \vdash \typ\NN
}
\indrule{\Enotn}{
  \tctx \vdash (\neg\typ)\nn
}{
  \tctx \vdash \typ\PP
}
}
\]
\[
{\small
\indrule{\Icp}{
  \tctx, \typ\NN \vdash \typ\pp
}{
  \tctx \vdash \typ\PP
}
\indrule{\Icn}{
  \tctx, \typ\PP \vdash \typ\nn
}{
  \tctx \vdash \typ\NN
}
}
\]
\[
{\small
\indrule{\Ecp}{
  \tctx \vdash \typ\PP
  \HS
  \tctx \vdash \typ\NN
}{
  \tctx \vdash \typ\pp
}
\indrule{\Ecn}{
  \tctx \vdash \typ\NN
  \HS
  \tctx \vdash \typ\PP
}{
  \tctx \vdash \typ\nn
}
}
\]
\end{definition}
\medskip

Rule $\Ax$ is the standard axiom rule.
Rule $\Abs$ is the {\em absurdity} rule, which allows one to derive
any proposition $\evtwo$ from a strong proposition $\ev$ and its opposite $\ev\OP$.
Rules $\Iandp$ and $\Eandp$ are introduction and elimination rules for
the strong affirmation of a conjunction.
Rules $\Iorp$ and $\Eorp$ are introduction and elimination rules for
the strong affirmation of a disjunction;
note that $\Eorp$ allows one to conclude {\em any} proposition.
Rules $\Inotp$ and $\Enotp$ are introduction and elimination rules for
the strong affirmation of a negation.
Rules $\Icp$ and $\Ecp$ are introduction and elimination rules for
classical affirmation (resembling introduction and elimination rules
for implication: $\Icp$ resembles the deduction theorem, while
$\Ecp$ resembles {\em modus ponens}).
The negative rules are dual to the positive ones,
\ie the rules for an affirmation of a given connective
have the same structure as the rules for the denial of the dual connective.
Note that conjunction is dual to disjunction and
negation and classical proposition are dual to themselves.


In the rest of this paper, we frequently use the following lemma
without explicit mention. It establishes a number of basic reasoning
principles that are valid in $\PRK$.

\begin{lemma}
\llem{admissible_rules_logic}
\llem{projection_of_conclusions}
\llem{classical_strengthening}
The following inference schemes are admissible in $\PRK$:
\begin{enumerate}
\item {\bf Weakening} ($\Weakening$):
  if $\tctx \vdash \ev$ then $\tctx,\evtwo \vdash \ev$.
\item {\bf Cut} ($\Cut$):
  if $\tctx,\ev \vdash \evtwo$ and $\tctx \vdash \ev$
  then $\tctx \vdash \evtwo$.
\item {\bf Substitution}:
  if $\tctx \vdash \evtwo$
  then $\tctx\sub{\btyp}{\typ} \vdash \evtwo\sub{\btyp}{\typ}$,
  where $-\sub{\btyp}{\typ}$ denotes the substitution of the
  propositional variable $\btyp$ for the pure proposition~$\typ$.
\item {\bf Generalized absurdity} ($\Abs'$):
  if $\tctx \vdash \ev$ and $\tctx \vdash \ev\OP$,
  where $\ev$ is not necessarily strong,
  then $\tctx \vdash \evtwo$.
\item {\bf Projection of conclusions} ($\PrCon$):
  if $\tctx \vdash \ev$ then $\tctx \vdash \trunc{\ev}$.
\item {\bf Contraposition} ($\Contrapose$):
  if $\ev$ is classical and
  $\tctx,\ev \vdash \evtwo$
  then $\tctx,\evtwo\OP \vdash \ev\OP$.
\item {\bf Classical strengthening} ($\ClStr$):
  if $\ev$ is classical and 
  $\tctx,\ev\OP \vdash \ev$ then $\tctx \vdash \ev$.
\end{enumerate}
\end{lemma}
\begin{proof}
{\bf Weakening}, {\bf cut}, and {\bf substitution}
are routine proofs
by induction on the derivation of the first judgment.

For {\bf generalized absurdity},
suppose that $\tctx \vdash \ev$ and $\tctx \vdash \ev\OP$. 
If $\ev$ is strong, applying the $\Abs$ rule we may conclude
$\tctx \vdash \evtwo$.
If $\ev$ is classical, there are two cases, depending on whether
$\ev$ is positive or negative. If $\ev$ is positive, \ie $\ev = \typ\PP$
then:
\[
  \indruleN{\Abs}{
    \indruleN{\Ecp}{
      \tctx \vdash \typ\PP
      \HS
      \tctx \vdash \typ\NN
    }{
      \tctx \vdash \typ\pp
    }
    \HS
    \indruleN{\Ecn}{
      \tctx \vdash \typ\NN
      \HS
      \tctx \vdash \typ\PP
    }{
      \tctx \vdash \typ\nn
    }
  }{
    \tctx \vdash \evtwo
  }
\]
If $\ev$ is negative, \ie $\ev = \typ\NN$, the proof is symmetric.

For {\bf projection of conclusions},
if $\ev$ is classical, \ie of the form $\typ\PP$ or $\typ\NN$, we are done.
If $\ev$ is strong, \ie of the form $\typ\pp$ or $\typ\nn$,
we conclude by applying the $\Icp$ or the $\Icn$ rule respectively.
For example, if $\ev = \typ\pp$:
\[
  \indruleN{\Icp}{
    \indruleN{\Weakening}{
      \tctx \vdash \typ\pp
    }{
      \tctx,\typ\NN \vdash \typ\pp
    }
  }{
    \tctx \vdash \typ\PP
  }
\]

For {\bf contraposition} we only study the case when $\ev$ is positive,
\ie $\ev = \typ\PP$; the negative case is symmetric.
So let $\tctx,\typ\PP \vdash \evtwo$. Then:
\[
  \indruleN{\Icn}{
    \indruleN{\Abs'}{
      \indruleN{\Weakening}{
        \tctx,\typ\PP \vdash \evtwo
      }{
        \tctx,\evtwo\OP,\typ\PP \vdash \evtwo
      }
      \HS
      \indruleN{\Ax}{
      }{
        \tctx,\evtwo\OP,\typ\PP \vdash \evtwo\OP
      }
    }{
      \tctx,\evtwo\OP,\typ\PP \vdash \typ\nn
    }
  }{
    \tctx,\evtwo\OP \vdash \typ\NN
  }
\]

For {\bf classical strengthening} we only study the case
when $\ev$ is positive, \ie $\ev = \typ\PP$; the negative case is symmetric.
So let $\tctx,\typ\NN \vdash \typ\PP$. Then:
\[
  \indruleN{\Icp}{
    \indruleN{\Ecp}{
      \tctx,\typ\NN \vdash \typ\PP
      \HS
      \indruleN{\Ax}{}{\tctx,\typ\NN \vdash \typ\NN}
    }{
      \tctx,\typ\NN \vdash \typ\pp
    }
  }{
    \tctx \vdash \typ\PP
  }\qedhere
\]
\end{proof}

\begin{example}[Law of excluded middle]
\lexample{lem_and_noncontr}
The law of excluded middle holds classically in $\PRK$,
that is, $\vdash (\typ\lor\neg\typ)\PP$.
Indeed, let $\tctx = \set{(\typ\lor\neg\typ)\NN, (\neg\typ)\NN}$,
and let $\derivation$ be the following derivation:
\[
  \begin{array}{c}
    \indruleN{\Icp}{
      \indruleN{\Iorp[1]}{
        \indruleN{\Icp}{
          \indruleN{\Abs'}{
            \indruleN{\Ax}{
            }{
              \tctx, \typ\NN \vdash (\neg\typ)\NN
            }
            \HS
            \indruleN{\Icp}{
              \indruleN{\Inotp}{
                \indruleN{\Ax}{
                }{
                  \tctx, \typ\NN, (\neg\typ)\NN \vdash \typ\NN
                }
              }{
                \tctx, \typ\NN, (\neg\typ)\NN \vdash (\neg\typ)\pp
              }
            }{
              \tctx, \typ\NN \vdash (\neg\typ)\PP
            }
          }{
            \tctx, \typ\NN \vdash \typ\pp
          }
        }{
          \tctx \vdash \typ\PP
        }
      }{
        \tctx \vdash (\typ\lor\neg\typ)\pp
      }
    }{
      \tctx \vdash (\typ\lor\neg\typ)\PP
    }
  \end{array}
\]
Then we have that:
\[
  \indruleN{\Icp}{
    \indruleN{\Iorp[2]}{
      \indruleN{\Icp}{
        \indruleN{\Inotp}{
          \indruleN{\Eorn}{
            \indruleN{\Ecn}{
              \indruleN{\Ax}{
              }{
                (\typ\lor\neg\typ)\NN, (\neg\typ)\NN \vdash (\typ\lor\neg\typ)\NN
              }
              \HS\HS
              \derivdots{\derivation}
              \HS
            }{
              (\typ\lor\neg\typ)\NN, (\neg\typ)\NN \vdash (\typ\lor\neg\typ)\nn
            }
          }{
            (\typ\lor\neg\typ)\NN, (\neg\typ)\NN \vdash \typ\NN
          }
        }{
          (\typ\lor\neg\typ)\NN, (\neg\typ)\NN \vdash (\neg\typ)\pp
        }
      }{
        (\typ\lor\neg\typ)\NN \vdash (\neg\typ)\PP
      }
    }{
      (\typ\lor\neg\typ)\NN \vdash (\typ\lor\neg\typ)\pp
    }
  }{
    \vdash (\typ\lor\neg\typ)\PP
  }
\]
Dually, the law of non-contradiction holds classically in $\PRK$,
that is, $\vdash (\typ\land\neg\typ)\NN$ holds.
Results from the following section will entail that
the strong law of excluded middle, $\vdash (\typ\lor\neg\typ)\pp$,
does not hold in $\PRK$ (see~\rexample{counter_model_lem}).
The reader may attempt to derive this judgment to convince
herself that it does not hold.
\end{example}

{\bf Projection Lemma.}
The proof of the following lemma is subtle.
It will be a key tool in order to prove completeness of $\PRK$
with respect to the Kripke semantics:
\begin{lemma}
\llem{projection_lemma}
If $\tctx,\ev \vdash \evtwo$
then $\tctx,\trunc{\ev} \vdash \trunc{\evtwo}$.
\end{lemma}
\begin{proof}
By induction on the derivation of $\tctx,\ev \vdash \evtwo$.
The difficult cases are conjunction and disjunction elimination.
\SeeAppendix{See~\rsec{appendix:projection_lemma} in the appendix for the proof.}
\end{proof}
A corollary obtained from iterating the projection lemma is that
if $\ev_1,\hdots,\ev_n \vdash \evtwo$
then $\trunc{\ev_1},\hdots,\trunc{\ev_n} \vdash \trunc{\evtwo}$.
\medskip
 
{\bf Duality Principle.}
The {\em dual} of a pure proposition $\typ$ is written $\typ^\bot$
and defined as:
\[
  \begin{array}{rcl@{\ }rcl}
  \btyp^\bot & \eqdef & \btyp
  &
  (\typ\land\typtwo)^\bot & \eqdef & \typ^\bot \lor \typtwo^\bot 
  \\
  (\typ\lor\typtwo)^\bot & \eqdef & \typ^\bot \land \typtwo^\bot 
  &
  (\neg\typ)^\bot & \eqdef & \neg(\typ^\bot)
  \end{array}
\]
The dual of a proposition $\ev$ is written $\ev^\bot$ and defined as:
\[
  \begin{array}{rcl@{\HS}rcl}
  (\typ\pp)^\bot & \eqdef & (\typ^\bot)\nn
  &
  (\typ\nn)^\bot & \eqdef & (\typ^\bot)\pp
  \\
  (\typ\PP)^\bot & \eqdef & (\typ^\bot)\NN
  &
  (\typ\NN)^\bot & \eqdef & (\typ^\bot)\PP
  \end{array}
\]
The following duality principle is then straightforward to prove
by induction on the derivation of the judgment:

\begin{lemma}
\llem{duality_principle}
If $\ev_1,\hdots,\ev_n \vdash \evtwo$
then $\ev_1^\bot,\hdots,\ev_n^\bot \vdash \evtwo^\bot$.
\end{lemma}

\section{Kripke Semantics for $\PRK$}
\lsec{prk_kripke_semantics}

In this section, we define a Kripke semantics~(\rdef{kripke_model},
\rdef{kripke_forcing}),
for which system $\PRK$ turns out to be sound~(\rprop{kripke_soundness})
and complete~(\rthm{kripke_completeness}).
Recall that a Kripke model $\kripmod$ in intuitionistic
logic\footnote{See for instance~\cite[Section~5.3]{DBLP:books/daglib/0080654}.}
is given by a set $\wlset$ of elements called {\em worlds},
a partial order $\leq$ on $\wlset$ called the {\em accessibility relation},
and for each world $\wl \in \wlset$ a set $\wlint{\wl}$
of propositional variables
verifying a {\em monotonicity} property,
namely, $\wl \leq \wl'$ implies $\wlint{\wl} \subseteq \wlint{\wl'}$.
A relation of forcing $\kripforce{\wl}{\typ}$ is defined
for each proposition $\typ$ by structural recursion on $\typ$.
In the base case, $\kripforce{\wl}{\btyp}$
is declared to hold
for a propositional variable $\btyp$ whenever $\btyp \in \wlint{\wl}$.

This standard notion of Kripke model is adapted for $\PRK$
by replacing the set $\wlint{\wl}$
with two sets $\wlpos{\wl}$ and $\wlneg{\wl}$
and by imposing an additional condition we call {\em stabilization},
stating that a propositional variable must eventually belong to the
union $\wlpos{\wl} \cup \wlneg{\wl}$,
but never to the intersection $\wlpos{\wl} \cap \wlneg{\wl}$.
The relation of forcing $\kripforce{\wl}{\ev}$ is then defined in such a way
that $\kripforce{\wl}{\btyp\pp}$ is declared to hold if
$\btyp \in \wlpos{\wl}$.
Similarly, $\btyp\nn$ is declared to hold if $\btyp \in \wlneg{\wl}$.
One difficulty that we find
is how to define the forcing relation for a classical proposition
like $\typ\PP$.
The forcing relation for $\typ\PP$ should behave, informally speaking,
like an intuitionistic implication ``$\typ\NN \to \typ\pp$''.
However this does not provide a
{\em bona fide} definition, because
the interpretation of $\typ\PP$ depends on $\typ\NN$,
and the interpretation of $\typ\NN$ depends in turn on $\typ\PP$.
What we do is define the interpretations of $\typ\PP$ and $\typ\NN$
without referring to each other.
A key lemma~(\rlem{rule_of_classical_forcing}) then ensures
that $\typ\PP$ is given the same semantics as an intuitionistic
implication of the form ``$\typ\NN \to \typ\pp$''.

\begin{definition}
\ldef{kripke_model}
A {\em Kripke model} (for $\PRK$)
is a structure $\kripmod = \kripstruct$
where $\wlset = \set{\wl,\wl',\hdots}$ is a set of worlds,
$\wleq$ is a partial order on $\wlset$,
and for each world $\wl \in \wlset$
there are sets $\wlpos{\wl}$ and $\wlneg{\wl}$ of propositional variables,
such that the following conditions hold:
\begin{enumerate}
\item {\bf Monotonicity.}
      If $\wl \wleq \wl'$
      then $\wlpos{\wl} \subseteq \wlpos{\wl'}$
      and $\wlneg{\wl} \subseteq \wlneg{\wl'}$.
\item {\bf Stabilization.}
      For all $\wl \in \wlset$ and all $\btyp$,
      there exists $\wl' \wgeq \wl$
      such that $\btyp \in \wlpos{\wl'} \symdiff \wlneg{\wl'}$.
\end{enumerate}
Note that we write $\wl' \wgeq \wl$ for $\wl \wleq \wl'$,
and $\symdiff$ denotes the symmetric difference on sets,
that is,
$X \symdiff Y =
 (X \setminus Y)
 \cup (Y \setminus X)$.
\end{definition}

The definition of the forcing relation is given by induction on
the following notion of {\em measure} $\#(-)$ of a proposition $\ev$:
\[
  \begin{array}{rcl@{\HS}rcl}
  \#(\typ\pp) & \eqdef & 2 |\typ| &
  \#(\typ\nn) & \eqdef & 2 |\typ| \\
  \#(\typ\PP) & \eqdef & 2 |\typ| + 1 &
  \#(\typ\NN) & \eqdef & 2 |\typ| + 1 
  \end{array}
\]
where $|\typ|$ denotes the {\em size}, \ie the number of symbols,
in the formula $\typ$.
Note in particular that
$\#(\typ\PP) = \#(\typ\NN) > \#(\typ\pp) = \#(\typ\nn)$,
that
$\#((\typ\star\typtwo)\pp) = \#((\typ\star\typtwo)\nn) > \#(\typ\PP) = \#(\typ\NN)$
for $\star \in \set{\land,\lor}$,
and that
$\#((\neg\typ)\pp) = \#((\neg\typ)\nn) > \#(\typ\PP) = \#(\typ\NN)$.

\begin{definition}[Forcing]
\ldef{kripke_forcing}
Given a Kripke model, we define the {\em forcing} relation,
written $\kripforce{\wl}{\ev}$
for each world $\wl \in \wlset$
and each proposition $\ev$,
as follows, by induction on the {\em measure} $\#(\ev)$:
\begin{center}
{\small
\begin{tabular}{l@{\ }l@{\ }l@{\ }l@{\ }l}
  $\kripforce{\wl}{\btyp\pp}$
  & $\iff$ &
  $\btyp \in \wlpos{\wl}$
\\
  $\kripforce{\wl}{\btyp\nn}$
  & $\iff$ &
  $\btyp \in \wlneg{\wl}$
\\
  $\kripforce{\wl}{(\typ\land\typtwo)\pp}$
  & $\iff$ &
  $\kripforce{\wl}{\typ\PP}$
  & \text{and} &
  $\kripforce{\wl}{\typtwo\PP}$
\\
  $\kripforce{\wl}{(\typ\land\typtwo)\nn}$
  & $\iff$ &
  $\kripforce{\wl}{\typ\NN}$
  & \text{or} &
  $\kripforce{\wl}{\typtwo\NN}$
\\
  $\kripforce{\wl}{(\typ\lor\typtwo)\pp}$
  & $\iff$ &
  $\kripforce{\wl}{\typ\PP}$
  & \text{or} &
  $\kripforce{\wl}{\typtwo\PP}$
\\
  $\kripforce{\wl}{(\typ\lor\typtwo)\nn}$
  & $\iff$ &
  $\kripforce{\wl}{\typ\NN}$
  & \text{and} &
  $\kripforce{\wl}{\typtwo\NN}$
\\
  $\kripforce{\wl}{(\neg\typ)\pp}$
  & $\iff$ &
  $\kripforce{\wl}{\typ\NN}$
\\
  $\kripforce{\wl}{(\neg\typ)\nn}$
  & $\iff$ &
  $\kripforce{\wl}{\typ\PP}$
\\
  $\kripforce{\wl}{\typ\PP}$
  & $\iff$ &
  $\kripnotforce{\wl'}{\typ\nn}$
  & \text{for all} & $\wl' \wgeq \wl$
\\
  $\kripforce{\wl}{\typ\NN}$
  & $\iff$ &
  $\kripnotforce{\wl'}{\typ\pp}$
  & \text{for all} & $\wl' \wgeq \wl$
\end{tabular}
}
\end{center}
Furthermore, if $\tctx$ is a (possibly infinite) set of propositions, we write:
\[
{\small
\begin{array}{l@{\ }l}
  \kripforcefull{\wl}{\tctx}
  &
  \iff \text{$\kripforce{\wl}{\ev}$ for every $\ev \in \tctx$}
\\
  \kripforcefull{\tctx}{\ev}
  &
  \iff
  \text{$\kripforcefull{\wl}{\tctx}$ implies $\kripforce{\wl}{\ev}$
        for every $\wl$}
\\
  \kripentails{\tctx}{\ev}
  &
  \iff
  \text{$\kripforcefull{\tctx}{\ev}$ for every Kripke model $\kripmod$}
\end{array}}
\]
\end{definition}

Note that most cases in the definition of forcing do not mention
the accessibility relation, other than for classical propositions.

\begin{example}[Counter-model for the strong excluded middle]
\lexample{counter_model_lem}
There is a Kripke model $\kripmod$ with a world $\wl_0$
such that $\kripnotforce{\wl_0}{(\btyp\lor\neg\btyp)\pp}$.
Indeed, let $\mathcal{P}$ be the set of all propositional variables,
and let $\kripmod$ be the Kripke model such that
$\wlset = \set{\wl_0,\wl_1,\wl_2}$
with $\wl_0 \leq \wl_1$ and $\wl_0 \leq \wl_2$,
where $\wlpos{}$ and $\wlneg{}$ are defined as follows:
\[
  \begin{array}{r|l|l}
  & \wlpos{} & \wlneg{} \\
  \hline
    \wl_0 & \emptyset     & \emptyset \\
    \wl_1 & \mathcal{P}   & \emptyset \\
    \wl_2 & \emptyset     & \mathcal{P} \\
  \end{array}
\]
It is easy to verify that $\kripmod$ is a Kripke model
and that $\kripnotforce{\wl_0}{(\btyp\lor\neg\btyp)\pp}$.
Note, on the other hand, that the classical excluded middle holds,
\ie $\kripforce{\wl_0}{(\btyp\lor\neg\btyp)\PP}$.
\end{example}

Before going on, we introduce typical nomenclature.
If $\tctx$ is a possibly infinite set of propositions,
we say that $\tctx \vdash \evtwo$ holds
whenever the judgment $\tctxtwo \vdash \evtwo$ is derivable in
$\PRK$ for some finite subset $\tctxtwo \subseteq \tctx$.
A set $\tctx$ of propositions is {\em consistent} if
there is a proposition $\ev$ such that $\tctx \nvdash \ev$.
Otherwise, $\tctx$ is {\em inconsistent}.

In the remainder of this section we shall prove that $\PRK$ is sound
and complete with respect to this notion of Kripke model.
\ie that $\tctx \vdash \ev$ holds if and only if
$\kripentails{\tctx}{\ev}$ holds.
We begin by establishing some basic properties of the forcing
relation.
\SeeAppendix{See~\rsec{appendix:properties_forcing}
in the appendix for the proofs.}
\begin{lemma}[Properties of Forcing]
\llem{properties_forcing}
\llem{monotonicity_forcing}
\llem{stabilization_forcing}
\llem{non_contradiction_forcing}
\quad
\begin{enumerate}
\item {\bf Monotonicity.}
  If $\kripforce{\wl}{\ev}$ and $\wl \leq \wl'$
  then $\kripforce{\wl'}{\ev}$.
\item {\bf Stabilization.}
  For every world $\wl$ and every proposition $\ev$,
  there is a world $\wl' \wgeq \wl$
  such that
  either $\kripforce{\wl'}{\ev}$
  or $\kripforce{\wl'}{\ev\OP}$ hold,
  but not both.
\item {\bf Non-contradiction.}
  If $\kripforce{\wl}{\ev}$ then $\kripnotforce{\wl}{\ev\OP}$.
\end{enumerate}
\end{lemma}

To prove {\bf soundness},
we first need an auxiliary lemma that gives
necessary and sufficient conditions for a classical proposition
to hold.
\SeeAppendix{See~\rsec{appendix:kripke_soundness}
in the appendix for the full proof of soundness.}
\begin{lemma}[Rule of Classical Forcing]
\llem{rule_of_classical_forcing}
\quad
\begin{enumerate}
\item $(\kripforce{\wl}{\typ\PP})$
      if and only if \\
      $(\forall \wl' \wgeq \wl)
       ((\kripforce{\wl'}{\typ\NN}) \implies (\kripforce{\wl'}{\typ\pp}))$.
\item $(\kripforce{\wl}{\typ\NN})$
      if and only if \\
      $(\forall \wl' \wgeq \wl)
       ((\kripforce{\wl'}{\typ\PP}) \implies (\kripforce{\wl'}{\typ\nn}))$.
\end{enumerate}
\end{lemma}

\noindent With these tools at our disposal, it is immediate to prove soundness:
\begin{proposition}[Soundness]
\lprop{kripke_soundness}
If $\tctx \vdash \ev$ then $\kripentails{\tctx}{\ev}$.
\end{proposition}
\begin{proof}
By induction on the derivation of $\tctx \vdash \ev$.
The interesting cases are the $\Icp$ and $\Icn$ rules,
which follow from the rule of classical forcing~(\rlem{rule_of_classical_forcing}),
and the $\Abs$, $\Ecp$, and $\Ecn$ rules, which follow from
the property of non-contradiction~(\rlem{non_contradiction_forcing}).
\end{proof}

To prove {\bf completeness}, the methodology that we follow is the
standard one, which proceeds by contraposition assuming
that $\tctx \nvdash \ev$ and building a counter-model.
The counter-model is given by a Kripke model
$\kripmod_0$ and a world $\wl$ such that
$\kripforcefull[\kripmod_0]{\wl}{\tctx}$
but $\kripnotforce[\kripmod_0]{\wl}{\ev}$.
In fact, the choice of the Kripke model $\kripmod_0$ does not depend
on $\tctx$ nor $\ev$.
Rather, $\kripmod_0$ is always chosen to be the {\em canonical}
Kripke model whose worlds
are ``saturated'' sets of propositions
({\em prime theories}, sometimes called {\em disjunctive theories}).
Completeness is obtained by taking $\tctx$ and {\em saturating} it
it to a prime theory $\tctx'$
which then verifies $\kripforcefull[\kripmod_0]{\tctx'}{\tctx}$
but $\kripnotforce[\kripmod_0]{\tctx'}{\ev}$.
\SeeAppendix{In the remainder of this section, the proofs of the
technical lemmas are only sketched;
see~\rsec{appendix:kripke_completeness} in the appendix for the
full proofs.}

\begin{definition}[Prime theory]
A {\em prime theory} is a set of propositions $\tctx$
such that the following hold:
\begin{enumerate}
\item {\bf Closure by deduction.}
  If $\tctx \vdash \ev$ then $\ev \in \tctx$.
\item {\bf Consistency.}
  $\tctx$ is consistent. 
\item {\bf Disjunctive property.} \\
  $\bullet$ If $(\typ \lor \typtwo)\pp \in \tctx$ then either $\typ\PP \in \tctx$ or $\typtwo\PP \in \tctx$. \\
  $\bullet$ If $(\typ \land \typtwo)\nn \in \tctx$ then either $\typ\NN \in \tctx$ or $\typtwo\NN \in \tctx$.
\end{enumerate}
\end{definition}

\begin{lemma}[Saturation]
\llem{kripke_saturation}
Let $\tctx$ be a consistent set of propositions,
and let $\evtwo$ be a proposition such that $\tctx \nvdash \evtwo$.
Then there exists a prime theory $\tctx' \supseteq \tctx$ such that
$\tctx' \nvdash \evtwo$.
\end{lemma}
\begin{proof}
\SeeAppendix{\rlem{appendix:kripke_saturation} in the appendix.}
\end{proof}

\begin{definition}[Canonical model]
The {\em canonical model} is the structure
$\kripmod_0 = (\wlset_0,\subseteq,\wlpos{},\wlneg{})$, where:
\begin{enumerate}
\item $\wlset_0$ is the set of all prime theories.
\item $\subseteq$ denotes the set-theoretic inclusion between prime theories.
\item $\wlpos{\tctx} = \set{\btyp \ST \btyp\pp \in \tctx}$
      and
      $\wlneg{\tctx} = \set{\btyp \ST \btyp\nn \in \tctx}$.
\end{enumerate}
\end{definition}

\begin{lemma}
The canonical model is a Kripke model.
\end{lemma}
\begin{proof}
\SeeAppendix{\rlem{appendix:canonical_model_is_kripke} in the appendix.}
The difficult part is proving the stabilization property,
which relies on the fact that if $\tctx$ is a consistent
set and $\ev$ is a proposition,
then $\tctx \cup \set{\ev}$ and $\tctx \cup \set{\ev\OP}$
are not both inconsistent.
\end{proof}

\begin{lemma}[Main Semantic Lemma]
\llem{kripke_main_semantic_lemma}
Let $\tctx$ be a prime theory.
Then
$\kripforce[\kripmod_0]{\tctx}{\ev}$ holds in the canonical model
if and only if $\ev \in \tctx$.
\end{lemma}
\begin{proof}
\SeeAppendix{\rlem{appendix:kripke_main_semantic_lemma} in the appendix.}
By induction on the measure $\#(\ev)$.
The difficult case is when $\ev$ is a classical proposition,
which requires resorting to the Saturation lemma~(\rlem{kripke_saturation}).
\end{proof}

\begin{theorem}[Completeness]
\lthm{kripke_completeness}
If $\kripentails{\tctx}{\ev}$ then $\tctx \vdash \ev$.
\end{theorem}
\begin{proof}
The proof is by contraposition,
\ie let $\tctx \nvdash \ev$
and let us show that
there is a Kripke model $\kripmod$
and a world $\wl$ such that $\kripforcefull{\wl}{\tctx}$
but $\kripnotforce{\wl}{\ev}$.
Note that $\tctx$ is consistent, so
by Saturation~(\rlem{kripke_saturation})
there exists a prime theory $\tctx' \supseteq \tctx$
such that $\tctx' \nvdash \ev$.
Note that $\kripforcefull[\kripmod_0]{\tctx'}{\tctx}$
because, by the Main~Semantic~Lemma~(\rlem{kripke_main_semantic_lemma}),
we have that $\kripforce[\kripmod_0]{\tctx'}{\evtwo}$
for every $\evtwo \in \tctx \subseteq \tctx'$.
Moreover, also by the Main~Semantic~Lemma~(\rlem{kripke_main_semantic_lemma}),
we have that $\kripnotforce[\kripmod_0]{\tctx'}{\ev}$ because $\ev \notin \tctx'$.
\end{proof}

\section{Propositions as types: the $\lambdaC$-Calculus}
\lsec{prk_lambdaC}

In this section, we formulate a typed $\lambda$-calculus, called $\lambdaC$,
by deriving a system of term assignment for derivations in $\PRK$,
and furnishing it with reduction rules.
Besides the basic results of confluence~(\rprop{lambdaC_confluent})
and subject reduction~(\rprop{subject_reduction})
the central result in this section is a translation~(\rdef{semF_translation_types},
\rdef{semF_translation})
from $\PRK$ to System~F extended with recursive type constraints,
following Mendler~\cite{mendler1991inductive}.
The translation maps each type $\ev$ of $\PRK$
to a type $\semF{\ev}$,
and each term $\tm$ of type $\ev$ to a term $\semF{\tm}$ of type $\semF{\ev}$.
Recursion is needed in order to be able to translate classical propositions,
which are characterized by the recursive equations discussed
in the introduction, $\typ\PP \approx (\typ\NN \to \typ\pp)$
and its dual $\typ\NN \approx (\typ\PP \to \typ\nn)$.
Moreover the translation is such that each reduction step $\tm \to \tmtwo$
in $\lambdaC$
is simulated in one or more steps $\semF{\tm} \to^+ \semF{\tmtwo}$
in the extended System~F.
This translation provides a {\em syntactical model} for $\lambdaC$ in the
sense of~\cite{DBLP:conf/cpp/BoulierPT17},
and one of its consequences is that $\lambdaC$ is strongly normalizing~(\rthm{lambdaC_canonical}).
This allows us to prove {\em canonicity}~(\rthm{canonicity}),
for which we study an inductive characterization of the set of normal forms.
Finally, we consider an extension $\lambdaCeta$ of the system that
incorporates an extensionality rule for classical proofs~(\rdef{lambdaCeta_calculus},
\rthm{lambdaCeta_canonical}).

Propositions $\ev,\evtwo,\hdots$
are sometimes also called {\em types} in this section.
We assume given a denumerable set of variables
$\var,\vartwo,\varthree,\hdots$.
The set of {\em typing contexts} is defined by the grammar
$\tctx  ::= \emptyctx \mid \tctx,\var:\ev$,
where each variable is assumed to occur at most once in a typing context.
Typing contexts are considered implicitly up to reordering\footnote{Remark
that the type system $\lambdaC$ is a {\em refinement}
of the logical system $\PRK$
because contexts are {\em multisets}, rather than {\em sets}, of
assumptions:
there is no structural rule of contraction in $\lambdaC$.
This means, for example, that in $\lambdaC$ there are
two different derivations of the sequent $\ev,\ev \vdash \ev$ using the $\Ax$ rule,
depending on which one of the two assumptions is used,
whereas in $\PRK$ there is only one such proof.
This is a typical situation in a propositions-as-types setting.}

The set of terms is given by the following abstract syntax.
The letter $i$ ranges over the set $\set{1,2}$.
Some terms are decorated with a plus or a minus sign.
In the grammar we write ``${}^{\pm}$'' to stand
for either ``$\pp$'' or ``$\nn$''.
\[
{\small
\begin{array}{rrll}
  \tm,\tmtwo,\tmthree,\hdots
    & ::= & \var
    & \text{variable}
  \\
    & \mid & \strongabs{\ev}{\tm}{\tmtwo}
    & \text{absurdity}
  \\
    & \mid & \pairpn{\tm}{\tmtwo}
    & \text{$\land\pp$ / $\lor\nn$ introduction}
  \\
    & \mid & \projipn{\tm}
    & \text{$\land\pp$ / $\lor\nn$ elimination}
  \\
    & \mid & \inipn{\tm}
    & \text{$\lor\pp$ / $\land\nn$ introduction}
  \\
    & \mid & \casepn{\tm}{\var:\ev}{\tmtwo}{\vartwo:\evtwo}{\tmthree}
    & \text{$\lor\pp$ / $\land\nn$ elimination}
  \\
    & \mid & \negipn{\tm}
    & \text{$\neg\pp$ / $\neg\nn$ introduction}
  \\
    & \mid & \negepn{\tm}
    & \text{$\neg\pp$ / $\neg\nn$ elimination}
  \\
    & \mid & \claslampn{(\var:\ev)}{\tm}
    & \text{classical introduction}
  \\
    & \mid & \clasappn{\tm}{\tmtwo}
    & \text{classical elimination}
\end{array}
}
\]
The notions of free and bound occurrences of variables are defined as
expected considering that
$\casepn{\tm}{\var:\ev}{\tmtwo}{\vartwo:\evtwo}{\tmthree}$
binds occurrences of $\var$ in $\tmtwo$
and occurrences of $\vartwo$ in $\tmthree$,
whereas $\claslampn{\var:\ev}{\tm}$ binds occurrences of $\var$ in $\tm$.
We implicitly work modulo $\alpha$-renaming of bound variables.
We write $\fv{\tm}$ for the set of free variables of $\tm$,
and $\tm\sub{\var}{\tmtwo}$ for the capture-avoiding substitution
of $\var$ by $\tmtwo$ in $\tm$.
Sometimes we omit type decorations if they are irrelevant or clear
from the context,
for example, we may write $\claslamp{\var}{\tm}$
rather than $\claslamp{(\var:\typ\NN)}{\tm}$,
and $\strongabs{}{\tm}{\tmtwo}$
rather than $\strongabs{\ev}{\tm}{\tmtwo}$.
Sometimes we also omit the name
of unused bound variables, writing ``$\under$'' instead;
\eg if $\var \not\in \fv{\tm}$
we may write $\claslamp{\under}{\tm}$ rather than
$\claslamp{(\var:\typ\NN)}{\tm}$.

\begin{definition}[The $\lambdaC$ type system]
\ldef{lambdaC_type_system}
Typing judgments are of the form $\tctx \vdash \tm : \ev$.
Derivability of judgments is defined inductively
by the following typing rules:
\[
{\small
\indrule{\Ax}{
}{
  \tctx,\var:\ev \vdash \var:\ev
}
\indrule{\Abs}{
  \tctx \vdash \tm : \ev
  \Hs
  \tctx \vdash \tmtwo : \ev\OP
  \Hs
  \text{$\ev$ strong}
}{
  \tctx \vdash \strongabs{\evtwo}{\tm}{\tmtwo} : \evtwo
}
}
\]
\[
{\small
\indrule{\Iandp}{
  \tctx \vdash \tm : \typ\PP
  \Hs
  \tctx \vdash \tmtwo : \typtwo\PP
}{
  \tctx \vdash \pairp{\tm}{\tmtwo} : (\typ \land \typtwo)\pp
}
\indrule{\Iorn}{
  \tctx \vdash \tm : \typ\NN
  \Hs
  \tctx \vdash \tmtwo : \typtwo\NN
}{
  \tctx \vdash \pairn{\tm}{\tmtwo} : (\typ \lor \typtwo)\nn
}
}
\]
\[
{\small
\indrule{\Eandp}{
  \tctx \vdash \tm : (\typ_1 \land \typ_2)\pp
  \Hs i \in \set{1, 2}
}{
  \tctx \vdash \projip{\tm} : \typ_i\PP
}
}
  \] %
  \[ %
{\small
\indrule{\Eorn}{
  \tctx \vdash \tm : (\typ_1 \lor \typ_2)\nn
  \Hs i \in \set{1, 2}
}{
  \tctx \vdash \projin{\tm} : \typ_i\NN
}
}
\]
\[
{\small
\indrule{\Iorp}{
  \tctx \vdash \tm : \typ_i\PP
  \Hs i \in \set{1, 2}
}{
  \tctx \vdash \inip{\tm} : (\typ_1 \lor \typ_2)\pp
}
\indrule{\Iandn}{
  \tctx \vdash \tm : \typ_i\NN
  \Hs i \in \set{1, 2}
}{
  \tctx \vdash \inin{\tm} : (\typ_1 \land \typ_2)\nn
}
}
\]
\[
{\small
\indrule{\Eorp}{
  \tctx \vdash \tm : (\typ \lor \typtwo)\pp
  \Hs
  \tctx, \var:\typ\PP \vdash \tmtwo : \ev
  \Hs
  \tctx, \vartwo:\typtwo\PP \vdash \tmthree : \ev
}{
  \tctx \vdash \casep{\tm}{\var:\typ\PP}{\tmtwo}{\vartwo:\typtwo\PP}{\tmthree}
             : \ev
}
}
  \] %
  \[ %
{\small
\indrule{\Eandn}{
  \tctx \vdash \tm : (\typ \land \typtwo)\nn
  \Hs
  \tctx, \var:\typ\NN \vdash \tmtwo : \ev
  \Hs
  \tctx, \vartwo:\typtwo\NN \vdash \tmthree : \ev
}{
  \tctx \vdash \casen{\tm}{\var:\typ\NN}{\tmtwo}{\vartwo:\typtwo\NN}{\tmthree}
             : \ev
}
}
\]
\[
{\small
\indrule{\Inotp}{
  \tctx \vdash \tm : \typ\NN
}{
  \tctx \vdash \negip{\tm} : (\neg\typ)\pp
}
\indrule{\Inotn}{
  \tctx \vdash \tm : \typ\PP
}{
  \tctx \vdash \negin{\tm} : (\neg\typ)\nn
}
}
\]
\[
{\small
\indrule{\Enotp}{
  \tctx \vdash \tm : (\neg\typ)\pp
}{
  \tctx \vdash \negep{\tm} : \typ\NN
}
\indrule{\Enotn}{
  \tctx \vdash \tm : (\neg\typ)\nn
}{
  \tctx \vdash \negen{\tm} : \typ\PP
}
}
\]
\[
{\small
\indrule{\Icp}{
  \tctx, \var : \typ\NN \vdash \tm : \typ\pp
}{
  \tctx \vdash \claslamp{(\var:\typ\NN)}{\tm} : \typ\PP
}
\indrule{\Icn}{
  \tctx, \var : \typ\PP \vdash \tm : \typ\nn
}{
  \tctx \vdash \claslamp{(\var:\typ\PP)}{\tm} : \typ\NN
}
}
\]
\[
{\small
\indrule{\Ecp}{
  \tctx \vdash \tm : \typ\PP
  \Hs
  \tctx \vdash \tmtwo : \typ\NN
}{
  \tctx \vdash \clasapp{\tm}{\tmtwo} : \typ\pp
}
\indrule{\Ecn}{
  \tctx \vdash \tm : \typ\NN
  \Hs
  \tctx \vdash \tmtwo : \typ\PP
}{
  \tctx \vdash \clasapn{\tm}{\tmtwo} : \typ\nn
}
}
\]
\end{definition}

\bigskip

\begin{remark}
Each typing rule in $\lambdaC$~(\rdef{lambdaC_type_system})
corresponds exactly
to the rule of the same name in $\PRK$~(\rdef{system_prk}).
It is immediate to show that
$\ev_1,\hdots,\ev_n \vdash \evtwo$ is derivable in $\PRK$
if and only if
$\var_1:\ev_1,\hdots,\var_n:\ev_n \vdash \tm : \evtwo$
is derivable in $\lambdaC$ for some term $\tm$.
\end{remark}

We begin by studying properties of $\lambdaC$ from the {\em logical} point
of view, as a type system. In particular, the following lemma adapts
some of the results in~\rlem{admissible_rules_logic}
and~\rexample{lem_and_noncontr} to $\lambdaC$, providing explicit proof
terms for derivations.

\begin{lemma}
\llem{admissible_rules}
\llem{lem_and_noncontr}
\label{lemma:admissible_rules}
The following rules are admissible in $\lambdaC$:
\begin{enumerate}
\item {\bf Weakening} ($\Weakening$):
  If $\tctx \vdash \tm : \ev$
  and $\var \not\in \fv{\tm}$
  then $\tctx, \var:\evtwo \vdash \tm : \ev$.
\item {\bf Cut} ($\Cut$):
  \label{lemma:admissible_rules:subst}
  if $\tctx,\var:\ev \vdash \tm : \evtwo$
  and $\tctx \vdash \tmtwo : \ev$
  then $\tctx \vdash \tm\sub{\var}{\tmtwo} : \evtwo$.
\item {\bf Generalized absurdity} ($\Abs'$):
  if $\tctx \vdash \tm : \ev$
  and $\tctx \vdash \tmtwo : \ev\OP$,
  where $\ev$ is not necessarily strong,
  there is a term $\abs{\evtwo}{\tm}{\tmtwo}$
  such that $\tctx \vdash \abs{\evtwo}{\tm}{\tmtwo} : \evtwo$.
\item {\bf Contraposition} ($\Contrapose$):
  if $\ev$ is classical and
  $\tctx, \var : \ev \vdash \tm : \evtwo$,
  there is a term $\contrapose{\var}{\vartwo}{\tm}$
  such that
  $\tctx, \vartwo : \evtwo\OP \vdash \contrapose{\var}{\vartwo}{\tm} : \ev\OP$.
\item {\bf Excluded middle}:
  there is a term $\lemP{\typ}$
  such that $\tctx \vdash \lemP{\typ} : (\typ \lor \neg\typ)\PP$.
\item {\bf Non-contradiction}:
  there is a term $\lemN{\typ}$
  such that $\tctx \vdash \lemN{\typ} : (\typ \land \neg\typ)\NN$.
\end{enumerate}
\end{lemma}
\begin{proof}
{\bf Weakening} and {\bf cut} are routine by induction on the derivation
of the first premise of the rule.
For {\bf generalized absurdity}, it suffices to take:
\[
  \abs{\evtwo}{\tm}{\tmtwo} \eqdef
    \begin{cases}
      \strongabs{\evtwo}{\tm}{\tmtwo}
      & \text{if $\ev$ is strong} \\
      \strongabs{\evtwo}{(\clasapp{\tm}{\tmtwo})}{(\clasapn{\tmtwo}{\tm})}
      & \text{if $\ev = \typ\PP$} \\
      \strongabs{\evtwo}{(\clasapn{\tm}{\tmtwo})}{(\clasapp{\tmtwo}{\tm})}
      & \text{if $\ev = \typ\NN$} \\
    \end{cases}
\]
For {\bf contraposition}, it suffices to take:
\[
  \contrapose{\var}{\vartwo}{\tm} \eqdef
    \begin{cases}
      \claslamn{(\var:\typ\PP)}{
        (\abs{\typ\nn}{
          \tm
        }{ 
          \vartwo
        })
      }
      & \text{if $\ev = \typ\PP$}
    \\
      \claslamn{(\var:\typ\NN)}{
        (\abs{\typ\pp}{
          \tm
        }{ 
          \vartwo
        })
      }
      & \text{if $\ev = \typ\NN$}
    \end{cases}
\]
For {\bf excluded middle}, it suffices to take:
  \[
  \!\!
  {\small
    \begin{array}{r@{\,}c@{\,}l}
    \lemP{\typ} & \eqdef &
      \claslamp{(\var:(\typ\lor\neg\typ)\NN)}{
        \inip[2]{
          \claslamp{(\vartwo:\neg\typ\NN)}{
            \negip{
              \projin[1]{
                \clasapn{
                  \var
                }{
                  \lemPinner{\vartwo}{\typ}
                }
              }
            }
          }
        }
      }
    \\
    \lemPinner{\vartwo}{\typ} & \eqdef &
      \claslamp{(\under:(\typ\lor\neg\typ)\NN)}{
        \inip[1]{
          \claslamp{(\varthree:\typ\NN)}{
            (\abs{
              \typ\pp
            }{
              \vartwo
            }{
              \claslamp{(\under:\neg\typ\NN)}{
                \negip{
                  \varthree
                }
              }
            })
          }
        }
      }
    \end{array}
  }
  \]
Dually, for {\bf non-contradiction}:
  \[
  \!\!
  {\small
    \begin{array}{r@{\,}c@{\,}l}
    \lemN{\typ} & \eqdef &
      \claslamn{(\var:(\typ \land \neg\typ)\PP)}{
        \inin[2]{
          \claslamn{(\vartwo:\neg\typ\PP)}{
            \negin{
              \projip[1]{
                \clasapp{
                  \var
                }{
                  \lemNinner{\vartwo}{\typ}
                }
              }
            }
          }
        }
      }
    \\
    \lemNinner{\vartwo}{\typ} & \eqdef &
      \claslamn{(\under:(\typ \land \neg\typ)\PP)}{
        \inin[1]{
          \claslamn{(\varthree:\typ\PP)}{
            (\abs{
              \typ\nn
            }{
              \vartwo
            }{
              \claslamn{(\under:\neg\typ\PP)}{
                \negin{
                  \varthree
                }
              }
            })
          }
        }
      }
    \end{array}
  }
  \]
\end{proof}

We now turn to studying the {\em computational} properties of $\lambdaC$,
provided with the following notion of reduction:

\begin{definition}[The $\lambdaC$-calculus]
\ldef{the_lambdaC_calculus}
Typable terms of $\lambdaC$ are endowed with the following
rewriting rules, closed under arbitrary contexts.
\[
  {\small
  \begin{array}{rcl@{\HS}l}
    \projipn{\pairpn{\tm_1}{\tm_2}}
    & \toa{\ruleProj} &
    \tm_i
    & \text{if $i \in \set{1,2}$}
  \\
    \casepn{(\inipn{\tm})}{\var}{\tmtwo_1}{\var}{\tmtwo_2}
    & \toa{\ruleCase} &
    \tmtwo_i\sub{\var}{\tm}
    & \text{if $i \in \set{1,2}$}
  \\
    \negepn{(\negipn{\tm})}
    & \toa{\ruleNeg} &
    \tm
  \\
    \clasappn{(\claslampn{\var}{\tm})}{\tmtwo}
    & \toa{\ruleBeta} &
    \tm\sub{\var}{\tmtwo}
  \\
    \strongabs{}{\pairpn{\tm_1}{\tm_2}}{\ininp{\tmtwo}}
    & \toa{\ruleAbsPairInj} &
    \abs{}{\tm_i}{\tmtwo}
    & \text{if $i \in \set{1,2}$}
  \\
    \strongabs{}{\inipn{\tm}}{\pairnp{\tmtwo_1}{\tmtwo_2}}
    & \toa{\ruleAbsInjPair} &
    \abs{}{\tm}{\tmtwo_i}
    & \text{if $i \in \set{1,2}$}
  \\
    \strongabs{}{(\negipn{\tm})}{(\neginp{\tmtwo})}
    & \toa{\ruleAbsNeg} &
    \abs{}{\tm}{\tmtwo}
  \end{array}
  }
\]
If many occurrences of ``$\pm$'' appear in the same expression,
they are all supposed to stand for the same sign (either $+$ or $-$),
and $\mp$ is supposed to stand for the opposite sign.
\end{definition}

\begin{example}
If $\var:\typ\NN \vdash \tm : \typ\pp$
and $\vartwo:\typ\PP \vdash \tmtwo : \typ\nn$
then:
\[
{\small
  \begin{array}{rrl}
  &&
  \strongabs{}{(\negin{(\claslamp{\var}{\tm})})}{(\negip{(\claslamn{\vartwo}{\tmtwo})})}
  \\
  & \longrightarrow &
  \abs{}{(\claslamp{\var}{\tm})}{(\claslamn{\vartwo}{\tmtwo})}
  \\
  & = &
  \strongabs{}{
    (\clasapp{(\claslamp{\var}{\tm})}{(\claslamn{\vartwo}{\tmtwo})})
  }{
    (\clasapn{(\claslamn{\vartwo}{\tmtwo})}{(\claslamp{\var}{\tm})})
  }
  \\
  & \longrightarrow &
  \strongabs{}{
    \tm\sub{\var}{(\claslamn{\vartwo}{\tmtwo})}
  }{
    (\clasapn{(\claslamn{\vartwo}{\tmtwo})}{(\claslamp{\var}{\tm})})
  }
  \\
  & \longrightarrow &
  \strongabs{}{
    \tm\sub{\var}{(\claslamn{\vartwo}{\tmtwo})}
  }{
    \tmtwo\sub{\vartwo}{(\claslamp{\var}{\tm})}
  }
  \end{array}
}
\]
\end{example}
\medskip

A first observation is that $\PRK$'s duality principle~(\rlem{duality_principle})
can be strengthened to obtain a {\bf computational duality principle} for $\lambdaC$.
The proof is immediate given that all typing and reduction rules are symmetric:

\begin{lemma}
If $\tm^\bot$ is the term that results
from flipping all the signs in $\tm$, then
$\tctx \vdash \tm : \ev$ if and only if $\tctx^\bot \vdash \tm^\bot : \ev^\bot$,
and $\tm \toa{} \tmtwo$ if and only if $\tm^\bot \toa{} \tmtwo^\bot$.
\end{lemma}

The second computational property that we study is {\bf subject reduction},
also known as {\em type preservation}.
This fundamental property ensures that reduction is well-defined over
the set of typable terms. More precisely:

\begin{proposition}
\lprop{subject_reduction}
If $\tctx \vdash \tm : \ev$
and $\tm \toa{} \tmtwo$, then $\tctx \vdash \tmtwo : \ev$.
\end{proposition}
\begin{proof}
The core of the proof consists in checking that each rewriting rule
preserves the type of the term.
\SeeAppendix{See~\rsec{appendix:subject_reduction} in the appendix
for the proof.}
\end{proof}

Third, the $\lambdaC$-calculus enjoys {\bf confluence},
the basic property of a rewriting system stating
that given reduction sequences $\tm_0 \to^* \tm_1$
and $\tm_0 \to^* \tm_2$
there must exist a term $\tm_3$ such that
$\tm_1 \to^* \tm_3$ and $\tm_2 \to^* \tm_3$.

\begin{proposition}
\lprop{lambdaC_confluent}
The $\lambdaC$-calculus is confluent.
\end{proposition}
\begin{proof}
The rewriting system $\lambdaC$ can be modeled
as a higher-order rewriting system (HRS) in the sense of Nipkow\footnote{It
suffices to model it with a single sort $\iota$, with constants such as
$\pi\pp_i : \iota \to \iota$,
$\mathsf{IC}\nn : (\iota \to \iota) \to \iota$,
etc., and rules such as
$\delta\pp(\mathsf{in}^+_1 x)\,f\,g \to f\,x$.
Strictly speaking, two constants for $\strongabs{}{}{}$ are needed,
depending on the signs.}.
This HRS is {\em orthogonal}, \ie left-linear without critical pairs,
which entails that it is confluent~\cite{nipkow1991higher}.
\end{proof}

Our next goal is to prove that $\lambdaC$ enjoys {\bf strong normalization},
that is, that there are no infinite reduction sequences
$\tm_1 \to \tm_2 \to \tm_3 \to \hdots$.
To do so, we give a translation to System~F extended with
{\em recursive type constraints}.

Type constraints are a way to define types
as solutions to recursive equations.
For instance, the type $\mathsf{T}$ of binary trees
is given by $\mathsf{T} \equiv 1 + (\mathsf{T} \times \mathsf{T})$.
In our case, the idea is to define $\typ\PP$ and $\typ\NN$ as
solutions to the mutually recursive equations
$\typ\PP \equiv (\typ\NN \to \typ\pp)$
and $\typ\NN \equiv (\typ\PP \to \typ\nn)$.

We begin by recalling the extended System~F and its relevant properties.
\medskip

{\bf System~F Extended with Recursive Type Constraints.}
In this subsection we recall the definition of
System~F$\extwith{\typeConstraints}$,
\ie System~F parameterized by an {\em arbitrary}
set of recursive type constraints $\typeConstraints$, as formulated
by Mendler~\cite{mendler1991inductive}.

The set of {\em types} in System~F$\extwith{\typeConstraints}$ is given by
$\typ ::= \btyp \mid \typ \to \typ \mid \forall\btyp.\typ$
where $\btyp,\btyptwo,\hdots$ are called {\em base types}.
The set of {\em terms} is given by
$
  \tm ::= \var^\typ
          \mid \lam{\var^\typ}{\tm}  \mid \tm\,\tm
          \mid \lam{\btyp}{\tm} \mid \tm\,\typ
$,
where $\lam{\btyp}{\tm}$ is type abstraction and
$\tm\,\typ$ is type application.
A {\em type constraint} is an equation of the form $\btyp \equiv \typ$.
System~F$\extwith{\typeConstraints}$
is parameterized by a set $\typeConstraints$
of type constraints. 
Each set $\typeConstraints$ of type constraints
induces a notion of equivalence between types,
written $\typ \equiv \typtwo$ and defined as the
congruence generated by $\typeConstraints$.
Typing rules are those of
the usual System~F~\cite[Section~11.3]{girard1989proofs}
extended with a {\em conversion} rule:
\[
{\small
  \indrule{Conv}{
    \tctx \vdash \tm : \typ
    \HS
    \typ \equiv \typtwo
  }{
    \tctx \vdash \tm : \typtwo
  }
}
\]
Variables occurring {\em positively} (resp. {\em negatively}) 
in a type $\typ$ are written $\posvars{\typ}$ (resp. $\negvars{\typ}$)
and defined as usual:
\[
  {\small
  \begin{array}{r@{\ }c@{\ }l@{\hspace{.5cm}}r@{\ }c@{\ }l}
    \posvars{\btyp} & \eqdef & \set{\btyp}
  &
    \negvars{\btyp} & \eqdef & \emptyset
  \\
    \posvars{\typ \to \typtwo} & \eqdef & \negvars{\typ} \cup \posvars{\typtwo}
  &
    \negvars{\typ \to \typtwo} & \eqdef & \posvars{\typ} \cup \negvars{\typtwo}
  \\
    \posvars{\forall\btyp.\typ} & \eqdef & \posvars{\typ} \setminus \set{\btyp}
  &
    \negvars{\forall\btyp.\typ} & \eqdef & \negvars{\typ} \setminus \set{\btyp}
  \end{array}
  }
\]
A set of type constraints $\typeConstraints$
verifies the {\em positivity condition} if
for every type constraint $(\btyp \equiv \typ) \in \typeConstraints$
and every type $\typtwo$ such that $\btyp \equiv \typtwo$
one has that $\btyp \not\in \negvars{\typtwo}$.
Mendler's main result~\cite[Theorem~13]{mendler1991inductive} is:
\begin{theorem}[Mendler, 1991]
\lthm{systemF_SN_Mendler}
If $\typeConstraints$ verifies the positivity condition,
then System~F$\extwith{\typeConstraints}$
is strongly normalizing.
\end{theorem}
\noindent We define the empty ($\tzero$), unit ($\tunit$),
product ($\typ \times \typtwo$), and sum types ($\typ + \typtwo$)
via their usual encodings in System~F
(see for instance~\cite[Section~11.3]{girard1989proofs}).
For example, the product type is defined as
$(\typ \times \typtwo) \eqdef
 \forall\btyp.((\typ \to \typtwo \to \btyp) \to \btyp)$.
with a constructor $\pairF{\tm}{\tmtwo}$
and an eliminator $\projiF{\tm}$.
\SeeAppendix{See~\rsec{appendix:system_f} for a more detailed
description of the Extended System~F.}
\medskip

{\bf System~F Extended with $\typeConstraintsPosNeg$ Constraints.}
In this subsection,
we describe System~F$\extwith{\typeConstraintsPosNeg}$,
an extension of System~F
with a {\em specific} set of recursive type constraints
called $\typeConstraintsPosNeg$.
Given that the set of base types is countably infinite,
we may assume without loss of generality
that, for any two types $\typ,\typtwo$ in System~F
there are two type variables,
called $\Pos{\typ}{\typtwo}$
and $\Neg{\typ}{\typtwo}$.
More precisely, the set of type variables can be partitioned as
$\mathbf{V} \uplus \mathbf{P} \uplus \mathbf{N}$
in such a way that
the propositional variables of $\lambdaC$
are identified with type variables of $\mathbf{V}$,
and there are bijective
mappings $(A,B) \mapsto \Pos{\typ}{\typtwo} \in \mathbf{P}$
and $(A,B) \mapsto \Neg{\typ}{\typtwo} \in \mathbf{N}$.
Note that we do not forbid $\typ$ and $\typtwo$ to have
occurrences of type variables in $\mathbf{P}$ and $\mathbf{N}$.\footnote{An
alternative, perhaps cleaner, presentation would be to define types
inductively as
$\typ,\typtwo,\hdots ::= \btyp \mid \Pos{\typ}{\typtwo} \mid \Neg{\typ}{\typtwo} \mid \typ \to \typtwo \mid \forall\btyp.\typ$.}

System~F$\extwith{\typeConstraintsPosNeg}$
is given by extending System~F
with the set of recursive type constraints $\typeConstraintsPosNeg$,
including the following equations for all types $\typ,\typtwo$:
\[
  \Pos{\typ}{\typtwo}
  \equiv
  (\Neg{\typ}{\typtwo} \to \typ)
  \HS
  \Neg{\typ}{\typtwo}
  \equiv
  (\Pos{\typ}{\typtwo} \to \typtwo)
\]
This extension is in fact {\bf strongly normalizing}:
\begin{corollary}
\lcoro{systemF_SN_posneg}
System~F$\extwith{\typeConstraintsPosNeg}$
is strongly normalizing.
\end{corollary}
\begin{proof}
A corollary of the previous theorem.
It suffices to show that the
recursive type constraints $\typeConstraintsPosNeg$
verify Mendler's positivity condition.
The proof of this fact is slightly technical.
\SeeAppendix{See~\rsec{appendix:positivity_condition} in
the appendix for the proof.}
\end{proof}
\medskip

{\bf Translating $\lambdaC$ to System~F$\extwith{\typeConstraintsPosNeg}$.}
We are now in a position to define the translation from $\lambdaC$
to System~F$\extwith{\typeConstraintsPosNeg}$.
\begin{definition}[Translation of Propositions]
\ldef{semF_translation_types}
A proposition $\ev$ of $\lambdaC$ is translated into a type
$\semF{\ev}$ of System~F$\extwith{\typeConstraintsPosNeg}$,
according to the following
definition, given by induction on the {\em measure} $\#(\ev)$
(defined in \rsec{prk_kripke_semantics}):
\[
{\small
\!\!
  \begin{array}{r@{\ }c@{\ }l}
    \semF{\btyp\pp}
    & \eqdef &
    \btyp
  \\
    \semF{(\typ \land \typtwo)\pp}
    & \eqdef &
    \semF{\typ\PP} \times \semF{\typtwo\PP}
  \\
    \semF{(\typ \lor \typtwo)\pp}
    & \eqdef &
    \semF{\typ\PP} + \semF{\typtwo\PP}
  \\
    \semF{(\neg\typ)\pp}
    & \eqdef &
    \tunit \to \semF{\typ\NN}
  \\
    \semF{\typ\PP}
    & \eqdef &
    \Pos{\semF{\typ\pp}}{\semF{\typ\nn}}
  \end{array}
  \begin{array}{r@{\ }c@{\ }l}
    \semF{\btyp\nn}
    & \eqdef &
    \btyp \to \tzero
  \\
    \semF{(\typ \land \typtwo)\nn}
    & \eqdef &
    \semF{\typ\NN} + \semF{\typtwo\NN}
  \\
    \semF{(\typ \lor \typtwo)\nn}
    & \eqdef &
    \semF{\typ\NN} \times \semF{\typtwo\NN}
  \\
    \semF{(\neg\typ)\nn}
    & \eqdef &
    \tunit \to \semF{\typ\PP}
  \\
    \semF{\typ\NN}
    & \eqdef &
    \Neg{\semF{\typ\pp}}{\semF{\typ\nn}}
  \end{array}
}
\]
Moreover, a typing context
$\tctx = (\var_1:\ev_1,\hdots,\var_n:\ev_n)$
is translated as
$\semF{\tctx} \eqdef (\var_1:\semF{\ev_1},\hdots,\var_n:\semF{\ev_n})$.
\end{definition}
Note that the translation of propositions mimicks
the equations for the realizability interpretation
discussed in the introduction.
In fact, the translation of $\semF{\typ\PP}$ is
$\Pos{\semF{\typ\pp}}{\semF{\typ\nn}}$,
which is equivalent to
$\Neg{\semF{\typ\pp}}{\semF{\typ\nn}} \to \semF{\typ\pp}$
according to the recursive type constraints in $\typeConstraintsPosNeg$,
and this in turn equals
$\semF{\typ\NN} \to \semF{\typ\pp}$,
just as required.
Similarly for the translation of $\typ\NN$.
The translation of $(\neg\typ)\pp$ is $(\tunit \to \semF{\typ\NN})$
rather than just $\semF{\typ\NN}$ for a technical reason, in
order to ensure that each reduction step in $\lambdaC$
is simulated by
{\em at least one} step in System~F$\extwith{\typeConstraintsPosNeg}$.

\begin{definition}[Translation of Terms]
\ldef{semF_translation}
First, we define a family of terms
$\vdash \funabsF{\ev}{\evtwo} : \semF{\ev} \to \semF{\ev\OP} \to \semF{\evtwo}$
in System~F$\extwith{\typeConstraintsPosNeg}$
as follows, by induction on the measure $\#(\ev)$:
\[
{\small
  \begin{array}{r@{\ }c@{\ }l}
    \funabsF{\btyp\pp}{\evtwo}
    & \eqdef &
    \lam{\var\,\vartwo}{
      \abortF{\semF{\evtwo}}{\vartwo\,\var}
    }
  \\
    \funabsF{\btyp\nn}{\evtwo}
    & \eqdef &
    \lam{\var\,\vartwo}{
      \abortF{\semF{\evtwo}}{\var\,\vartwo}
    }
  \\
    \funabsF{(\typ\land\typtwo)\pp}{\evtwo}
    & \eqdef &
    \lam{\var\,\vartwo}{
      \caseF{\vartwo}{
        \varthree
      }{
        \funabsF{\typ\PP}{\evtwo}\,\projiF[1]{\var}\,\varthree
      }{
        \varthree
      }{
        \funabsF{\typtwo\PP}{\evtwo}\,\projiF[2]{\var}\,\varthree
      }
    }
  \\
    \funabsF{(\typ\land\typtwo)\nn}{\evtwo}
    & \eqdef &
    \lam{\var\,\vartwo}{
      \caseF{\var}{
        \varthree
      }{
        \funabsF{\typ\NN}{\evtwo}\,\varthree\,\projiF[1]{\vartwo}
      }{
        \varthree
      }{
        \funabsF{\typtwo\NN}{\evtwo}\,\varthree\,\projiF[2]{\var}
      }
    }
  \\
    \funabsF{(\typ\lor\typtwo)\pp}{\evtwo}
    & \eqdef &
    \lam{\var\,\vartwo}{
      \caseF{\var}{
        \varthree
      }{
        \funabsF{\typ\PP}{\evtwo}\,\var\,\projiF[1]{\vartwo}
      }{
        \varthree
      }{
        \funabsF{\typtwo\PP}{\evtwo}\,\var\,\projiF[2]{\vartwo}
      }
    }
  \\
    \funabsF{(\typ\lor\typtwo)\nn}{\evtwo}
    & \eqdef &
    \lam{\var\,\vartwo}{
      \caseF{\vartwo}{
        \varthree
      }{
        \funabsF{\typ\NN}{\evtwo}\,\projiF[1]{\var}\,\varthree
      }{
        \varthree
      }{
        \funabsF{\typtwo\PP}{\evtwo}\,\projiF[2]{\var}\,\varthree
      }
    }
  \\
    \funabsF{(\neg\typ)\pp}{\evtwo}
    & \eqdef &
    \lam{\var\,\vartwo}{
      \funabsF{\typ\NN}{\evtwo}\,(\var\,\trivF)\,(\vartwo\,\trivF)
    }
  \\
    \funabsF{(\neg\typ)\nn}{\evtwo}
    & \eqdef &
    \lam{\var\,\vartwo}{
      \funabsF{\typ\PP}{\evtwo}\,(\var\,\trivF)\,(\vartwo\,\trivF)
    }
  \\
    \funabsF{\typ\PP}{\evtwo}
    & \eqdef &
    \lam{\var\,\vartwo}{
      \funabsF{\typ\pp}{\evtwo}
        \,(\var\,\vartwo)
        \,(\vartwo\,\var)
    }
  \\
    \funabsF{\typ\NN}{\evtwo}
    & \eqdef &
    \lam{\var\,\vartwo}{
      \funabsF{\typ\nn}{\evtwo}
        \,(\var\,\vartwo)
        \,(\vartwo\,\var)
    }
  \end{array}
}
\]
where: $\abort{\typ}{\tm}$ denotes an inhabitant of $\typ$
whenever $\tm$ is an inhabitant of the empty type;
$\caseF{\tm}{\var}{\tmtwo}{\var}{\tmthree}$ is the eliminator of the sum type;
$\projiF{\tm}$ is the eliminator of the product type;
and $\triv$ denotes the trivial inhabitant of the unit type.
Now each typable term $\tctx \vdash \tm : \ev$ in $\lambdaC$
can be translated into a term $\semF{\tctx} \vdash \semF{\tm} : \semF{\ev}$
of System~F$\extwith{\typeConstraintsPosNeg}$ as follows:
\[
{\small
  \begin{array}{rcll}
    \semF{\var}
    & \eqdef &
    \var
  \\
    \semF{\strongabs{\evtwo}{\tm}{\tmtwo}}
    & \eqdef &
    \funabsF{\ev}{\evtwo}\,\semF{\tm}\,\semF{\tmtwo}
  \\
    && \HS\text{if $\tctx \vdash \tm : \ev$ and $\tctx \vdash \tmtwo : \ev\OP$}
  \\
    \semF{\pairpn{\tm}{\tmtwo}}
    & \eqdef &
    \pairF{\semF{\tm}}{\semF{\tmtwo}}
  \\
    \semF{\projipn{\tm}}
    & \eqdef &
    \projiF{\semF{\tm}}
  \\
    \semF{\inipn{\tm}}
    & \eqdef &
    \iniF{\semF{\tm}}
  \\
    \semF{\casepn{\tm}{(\var:\ev)}{\tmtwo}{(\vartwo:\evtwo)}{\tmthree}}
    & \eqdef &
    \caseF{\semF{\tm}}{
      (\var:\semF{\ev})
    }{
      \semF{\tmtwo}
    }{
      (\vartwo:\semF{\evtwo})
    }{
      \semF{\tmthree}
    }
  \\
    \semF{\negipn{\tm}}
    & \eqdef &
    \lam{\var^{\tunit}}{\semF{\tm}}
    \HS\text{where $\var \not\in \fv{\tm}$}
  \\
    \semF{\negepn{\tm}}
    & \eqdef &
    \semF{\tm}\,\trivF
  \\
    \semF{\claslampn{(\var:\ev)}{\tm}}
    & \eqdef &
    \lam{\var^{\semF{\ev}}}{\semF{\tm}}
  \\
    \semF{\clasappn{\tm}{\tmtwo}}
    & \eqdef &
    \semF{\tm}\,\semF{\tmtwo}
  \\
  \end{array}
}
\]
\end{definition}
It is easy
to check that $\semF{\tctx} \vdash \semF{\tm} : \semF{\ev}$ holds in
System~F$\extwith{\typeConstraintsPosNeg}$ 
by induction on the derivation of the judgment $\tctx \vdash \tm : \ev$
in $\lambdaC$.
Two straightforward properties of the translation are:

\begin{lemma}
\llem{semF_properties}
\textnormal{1.} $\fv{\semF{\tm}} = \fv{\tm}$;
\textnormal{2.} $\semF{\tm\sub{\var}{\tmtwo}} = \semF{\tm}\sub{\var}{\semF{\tmtwo}}$.
\end{lemma}

The key result is the following {\bf simulation} lemma
from which strong normalization follows:
\begin{lemma}
\llem{semF_simulation}
If $\tm \toa{} \tmtwo$ in $\lambdaC$
then $\semF{\tm} \toa{}^+ \semF{\tmtwo}$
in System~F$\extwith{\typeConstraintsPosNeg}$.
\end{lemma}
\begin{proof}
\SeeAppendix{\rlem{appendix:semF_simulation} in the appendix.}
By case analysis on the rewriting rule used to derive
the step $\tm \toa{} \tmtwo$. The interesting case is
when the rule is applied at the root of the term.
As an illustrative example, consider an instance
of the $\ruleAbsPairInj$ rule,
with $\tctx \vdash \tm_1 : \typ_1\PP$,
and $\tctx \vdash \tm_2 : \typ_2\PP$,
and $\tctx \vdash \tmtwo : \typ_i\NN$ for some $i \in \set{1,2}$.
Then:
\[
{\small
  \begin{array}[b]{ll}
  &
    \semF{\strongabs{\ev}{\pairp{\tm_1}{\tm_2}}{\inin{\tmtwo}}}
  \\
  = &
    \funabsF{(\typ_1 \land \typ_2)\pp}{\ev}
      \,\pairF{\semF{\tm_1}}{\semF{\tm_2}}
      \,\iniF{\semF{\tmtwo}}
  \\
  \toa{}^+ &
    \!\!\!
      \caseFtablex{\iniF{\semF{\tmtwo}}}{
        (\varthree:\semF{\typ_1\NN})
      }{
        \funabsF{\typ_1\PP}{\ev}\,\projiF[1]{ \pairF{\semF{\tm_1}}{\semF{\tm_2}} }\,\varthree
      }{
        (\varthree:\semF{\typ_2\NN})
      }{
        \funabsF{\typ_2\PP}{\ev}\,\projiF[2]{ \pairF{\semF{\tm_1}}{\semF{\tm_2}} }\,\varthree
      }
  \\
  \toa{} &
    \funabsF{\typ_i\PP}{\ev}\,\projiF{ \pairF{\semF{\tm_1}}{\semF{\tm_2}} }\,\semF{\tmtwo}
  \\
  \toa{} &
    \funabsF{\typ_i\PP}{\ev}\,\semF{\tm_i}\,\semF{\tmtwo}
  \\
  \toa{}^+ &
    \funabsF{\typ_i\pp}{\ev}\,(\semF{\tm_i}\,\semF{\tmtwo})
                              (\semF{\tmtwo}\,\semF{\tm_i})
  \\
  = &
    \semF{\strongabs{\ev}{(\clasapp{\tm_i}{\tmtwo})}{(\clasapn{\tm_i}{\tmtwo})}}
  \\
  = &
    \semF{\abs{\ev}{\tm_i}{\tmtwo}}
  \end{array} \qedhere
}
\]
\end{proof}

\begin{theorem}
\lthm{lambdaC_canonical}
The $\lambdaC$-calculus is strongly normalizing.
\end{theorem}
\begin{proof}
An easy consequence of~\rlem{semF_simulation}
given that the extended System~F is strongly normalizing
(\rcoro{systemF_SN_posneg}).
\end{proof}
\medskip

{\bf Canonicity.}
In the previous subsections we have shown that the $\lambdaC$-calculus
enjoys subject reduction and strong normalization.
This implies that each typable term $\tm$ reduces to a normal form
$\tm'$ of the same type.
In this subsection, these results are refined
to prove a {\em canonicity} theorem,
stating that each closed, typable term $\tm$
reduces to a {\em canonical} term $\tm'$ of the same type.
For example, canonical terms of type $(\typ\lor\typtwo)\pp$ 
are of the form $\inip{\tm}$. From the logical point of view,
this means that given a strong proof of $(\typ\lor\typtwo)$,
in a context without assumptions, one can always recover
either a classical proof of $\typ$ or a classical proof of $\typtwo$.
This shows that $\PRK$ has a form of disjunctive property.

First we provide an inductive characterization
of the set of {\bf normal forms} of $\lambdaC$.

\begin{definition}[Normal terms]
\ldef{normal_terms}
The sets of {\em normal terms} ($\nf,\hdots$)
and {\em neutral terms} ($\neu,\hdots$)
are defined mutually inductively by:
\[
\begin{array}{rrllllllllll}
  \nf
    & ::=  & \neu
    & \mid & \pairpn{\nf}{\nf}
    & \mid & \inipn{\nf}
    \\
    & \mid & \negipn{\nf}
    & \mid & \claslampn{\var : \ev}{\nf}
\\
\\
  \neu
    & ::=  & \var
    & \mid & \projipn{\neu}
    & \mid & \casepn{\neu}{\var}{\nf}{\var}{\nf}
    \\
    & \mid & \negepn{\neu}
    & \mid & \clasappn{\neu}{\nf}
    \\
    & \mid & \strongabs{\ev}{\neu}{\nf}
    & \mid & \strongabs{\ev}{\nf}{\neu}
\end{array}
\]
\end{definition}

\begin{proposition}
\lprop{characterization_of_normal_terms}
A term is in the grammar of normal terms
if and only if it is a normal form,
\ie it does not reduce in $\lambdaC$.
\end{proposition}
\begin{proof}
Straightforward by induction.
\SeeAppendix{See~\rsec{appendix:characterization_of_normal_forms} in the
appendix for a detailed proof.}
\end{proof}

In order to state a canonicity theorem succintly,
we introduce some nomenclature.
A term is {\em canonical}
if it has any of the following shapes:
\[
  \pairpn{\tm_1}{\tm_2}
  \HS
  \inipn{\tm}
  \HS
  \negipn{\tm}
  \HS
  \claslampn{\var}{\tm}
\]
A typing context is {\em classical} if all the assumptions are
classical, \ie of the form $\typ\PP$ or $\typ\NN$.
Recall that a context is a term $\gctx$ with a single free occurrence of a hole
$\ctxhole$, and that $\gctxof{\tm}$ denotes the capturing substitution
of the term $\tm$ into the hole of $\gctx$.
A {\em case-context} is a context of the form
$
  \casectx ::= \ctxhole
          \mid \casepn{\casectx}{\var}{\tm}{\vartwo}{\tmtwo}
$.
An {\em eliminative context} is a context of the form
$
  \elctx ::= \ctxhole
        \mid \projipn{\elctx}
        \mid \casepn{\elctx}{\var}{\tm}{\vartwo}{\tmtwo}
        \mid \negepn{\elctx}
$.
Note that $\clasappn{\ctxhole}{\tm}$ is not eliminative
and that all case-contexts are eliminative.
An {\em explosion} is a term
of the form $\strongabs{\ev}{\tm}{\tmtwo}$
or of the form $\clasappn{\tm}{\tmtwo}$.
A term is {\em closed} if it has no free variables.
A term is {\em open} if it not closed,
\ie it has at least one free variable.

The following theorem has three parts;
the first one provides guarantees for {\em closed} terms,
whereas the two other ones provide weaker guarantees
for terms typable under an arbitrary classical context.
\begin{theorem}[Canonicity]
\lthm{canonicity}
\quad
\begin{enumerate}
\item
  Let $\vdash \tm : \ev$.
  Then $\tm$ reduces to a canonical term.
\item
  Let $\tctx \vdash \tm : \ev$
  where $\tctx$ is classical and $\ev$ is strong.
  Then either $\tm \toa{}^* \tm'$ where $\tm'$ is canonical
  or $\tm \toa{}^* \casectxof{\tm'}$
  where $\casectx$ is a case-context
  and $\tm'$ is an open explosion.
\item
  Let $\tctx \vdash \tm : \ev$
  where $\tctx$ and $\ev$ are classical.
  Then either $\tm \toa{}^* \claslampn{\var}{\tm'}$
  or $\tm \toa{}^* \elctxof{\tm'}$,
  where $\elctx$ is an eliminative context
  and $\tm'$ is a variable or an open explosion.
\end{enumerate}
\end{theorem}
\begin{proof}
By subject reduction~(\rprop{subject_reduction})
and strong normalization~(\rthm{lambdaC_canonical})
the term $\tm$ reduces to a normal form $\tm'$.
Moreover, by~\rprop{characterization_of_normal_terms},
we have that $\tm'$ is generated by the grammar of~\rdef{normal_terms}.
The proof then proceeds by induction on the derivation of $\tm'$
in the grammar of normal terms.
\SeeAppendix{See~\rsec{appendix:canonicity} in the appendix for the proof.}
\end{proof}
\medskip

{\bf Extensionality for Classical Proofs.}
To conclude the syntactic study of $\lambdaC$, we discuss that an
extensionality rule, akin to $\eta$-reduction in the $\lambda$-calculus,
may be incorporated to $\lambdaC$, obtaining a calculus $\lambdaCeta$.

\begin{definition}
\ldef{lambdaCeta_calculus}
The $\lambdaCeta$-calculus is defined by extending the $\lambdaC$ calculus
with the following reduction rule:
\[
  \claslampn{\var}{(\clasappn{\tm}{\var})}
  \ \toa{\ruleEta}\ %
  \tm
  \HS\text{if $\var \notin \fv{\tm}$}
\]
\end{definition}

\begin{theorem}
\lthm{lambdaCeta_canonical}
The $\lambdaCeta$-calculus enjoys subject
reduction,
and it is strongly normalizing and confluent.
\end{theorem}
\begin{proof}
Subject reduction is straightforward,
extending~\rprop{subject_reduction}
with an easy case for the $\ruleEta$ rule.
Local confluence is also straightforward by examining the
critical pairs.
The key lemma to prove strong normalization is that
$\ruleEta$ reduction steps can be postponed after steps of other kinds.
\noindent\SeeAppendix{See~\rsec{appendix:extensionality} for the details.}
\end{proof}

\section{Embedding Classical Logic into $\PRK$}
\lsec{prk_classical_logic}

Intuitionistic logic {\em refines}
classical logic: each intuitionistically valid formula
$\typ$ is also classically valid,
but there may be many classically equivalent
``readings'' of a classical formula
which are not intuitionistically equivalent,
such as $\neg(\typ \land \neg\typtwo)$ and $\neg\typ \lor \typtwo$.
System $\PRK$ refines classical logic in a similar sense.
For example, the classical sequent $\btyp \vdash \btyp$
may be ``read'' in $\PRK$ in various different ways,
such as $\btyp\pp \vdash \btyp\PP$ and $\btyp\PP \vdash \btyp\pp$,
of which the former holds but the latter does not.
In this section we show
that $\PRK$ is {\em conservative}~(\rprop{prk_conservative})
with respect to classical logic,
and that classical logic may be {\em embedded}~(\rthm{prk_embedding})
in $\PRK$.
We also describe the computational behavior of the terms resulting from
this embedding~(\rlem{classical_embedding_computation}).
\medskip

First, we claim that $\PRK$ is a {\bf conservative extension}
of classical logic, \ie if
$\typ\PP_1,\hdots,\typ\PP_n \vdash \typtwo\PP$
holds in $\PRK$
then the sequent $\typ_1,\hdots,\typ_n \vdash \typtwo$
holds in classical logic. In general:
\begin{proposition}
\lprop{prk_conservative}
Define $\classem{\ev}$ as follows:
\[
{\small
  \begin{array}{r@{\ }c@{\ }l@{\HS}r@{\ }c@{\ }l}
  \classem{\typ\PP} & \eqdef & \typ
  &
  \classem{\typ\NN} & \eqdef & \neg\typ
  \\
  \classem{\typ\pp} & \eqdef & \typ
  &
  \classem{\typ\nn} & \eqdef & \neg\typ \\
  \end{array}
}
\]
If the sequent $\ev_1,\hdots,\ev_n \vdash \evtwo$ holds in $\PRK$
then the sequent $\classem{\ev_1},\hdots,\classem{\ev_n} \vdash \classem{\evtwo}$
holds in classical propositional logic.
\end{proposition}
\begin{proof}
By induction on the derivation of the judgment,
observing that all the inference rules in $\PRK$
are mapped to classically valid inferences.
For example, for the $\Eandn$ rule,
note that
if $\tctx \vdash \neg(\typ \land \typtwo)$
and $\tctx, \neg\typ \vdash \typthree$
and $\tctx, \neg\typtwo \vdash \typthree$
hold in classical propositional logic then
$\tctx \vdash \typthree$.
\end{proof}

Second, we claim that classical logic may be {\bf embedded} in
$\PRK$, that is:
\begin{theorem}
\lthm{prk_embedding}
If $\typ_1,\hdots,\typ_n \vdash \typtwo$ holds in classical logic
then $\typ\PP_1,\hdots,\typ\PP_n \vdash \typtwo\PP$ holds in $\PRK$.
\end{theorem}
\begin{proof}
The proof is by induction on the proof of the sequent
$\typ_1,\hdots,\typ_n \vdash \typtwo$ in
Gentzen's system of natural deduction for classical logic \textsf{NK},
including
introduction and elimination rules for conjunction,
disjunction, and negation (encoding falsity as the pure
proposition $\Bot \eqdef (\btyp_0 \land \neg\btyp_0)$
for some fixed propositional variable $\btyp_0$),
the explosion principle, and the law of excluded middle.
We build the corresponding proof terms in $\lambdaC$:
\medskip

\noindent{1.} {\bf Conjunction introduction.}
  Let $\tctx \vdash \tm : \typ\PP$ and $\tctx \vdash \tmtwo : \typtwo\PP$.
  Then $\tctx \vdash \pairc{\tm}{\tmtwo} : (\typ \land \typtwo)\PP$
  where:
  \[
    \pairc{\tm}{\tmtwo} \eqdef
    \claslamp{(\under:(\typ\land\typtwo)\NN)}{
      \pairp{\tm}{\tmtwo}
    }
  \]

\noindent{2.} {\bf Conjunction elimination.}
  Let $\tctx \vdash \tm : (\typ_1 \land \typ_2)\PP$.
  Then $\tctx \vdash \projic{\tm} : \typ_i\PP$ where:
  \[
    \projic{\tm} \eqdef
    \claslamp{(\var:\typ_i\NN)}{
      \clasapp{
        \projip{
          \clasapp{
            \tm
          }{
            \claslamn{(\under:(\typ_1 \land \typ_2)\PP)}{\inin{\var}}
          }
        }
      }{
        \var
      }
    }
  \]

\noindent{3.} {\bf Disjunction introduction.}
  Let $\tctx \vdash \tm : \typ_i\PP$.
  Then $\tctx \vdash \inic{\tm} : (\typ_1 \lor \typ_2)\PP$ where:
  \[
    \inic{\tm} \eqdef
    \claslamp{(\under:(\typ_1\lor\typ_2)\NN)}{
      \inip{\tm}
    }
  \]

\noindent{4.} {\bf Disjunction elimination.}
  Let
  $\tctx \vdash \tm : (\typ \lor \typtwo)\PP$
  and $\tctx, \var : \typ\PP \vdash \tmtwo : \typthree\PP$
  and $\tctx, \var : \typtwo\PP \vdash \tmthree : \typthree\PP$.
  Then $\tctx \vdash
        \casec{\tm}{(\var:\typ\PP)}{\tmtwo}{(\var:\typtwo\PP)}{\tmthree}
        : \typthree\PP$,
  where:
  \[
    \claslamp{(\vartwo:\typthree\NN)}{
      \caseptablex{
        (\clasapp{
          \tm
        }{
          \claslamn{(\under:(\typ\lor\typtwo)\PP)}{
            \pairn{
              \contrapose{\var}{\vartwo}{
                \tmtwo
              }
            }{
              \contrapose{\var}{\vartwo}{
                \tmthree
              }
            }
          }
        })
      }{
        (\var : \typ\PP)
      }{
        \clasapp{
          \tmtwo
        }{
          \vartwo
        }
      }{
        (\var : \typtwo\PP)
      }{
        \clasapp{
          \tmthree
        }{
          \vartwo
        }
      }
    }
  \]
  Recall that $\contrapose{\var}{\vartwo}{\tm}$
  stands for the witness of contraposition~(\rlem{admissible_rules}).

\noindent{5.} {\bf Negation introduction.}
  By \rlem{lem_and_noncontr} we have that
  $\tctx \vdash \lemN{\btyp_0} : (\btyp_0 \land \neg\btyp_0)\NN$,
  that is $\tctx \vdash \lemN{\btyp_0} : \Bot\NN$.
  Moreover, suppose that $\tctx, \var:\typ\PP \vdash \tm : \Bot\PP$.
  Then $\tctx \vdash \neglamc{(\var:\typ\PP)}{\tm} : (\neg\typ)\PP$,
  where:
  \[
    \neglamc{(\var:\typ\PP)}{\tm} \eqdef
    \claslamp{(\under:(\neg\typ)\NN)}{
      \negip{
        \claslamn{(\var:\typ\PP)}{
          (\abs{
            \typ\nn
          }{
            \tm
          }{
            \lemN{\btyp_0}
          })
        }
      }
    }
  \]

\noindent{6.} {\bf Negation elimination.}
  Let $\tctx \vdash \tm : (\neg\typ)\PP$
  and $\tctx \vdash \tmtwo : \typ\PP$.
  Then $\tctx \vdash \negapc{\tm}{\tmtwo} : \Bot\PP$,
  where:
  \[
    \negapc{\tm}{\tmtwo} \eqdef
    \abs{
      \Bot\PP
    }{
      \tm
    }{
      \claslamn{(\under:(\neg\typ)\PP)}{
        \negin{
          \tmtwo
        }
      }
    }
  \]

\noindent{7.} {\bf Explosion.}
  Let $\tctx \vdash \tm : \Bot\PP$.
  Then $\tctx \vdash (\abs{\evtwo}{\tm}{\lemN{\btyp_0}}) : \evtwo$.

\noindent{8.} {\bf Excluded middle.}
  It suffices to take $\lemC{\typ} \eqdef \lemP{\typ}$.
  Then by \rlem{lem_and_noncontr},
  $\tctx \vdash \lemC{\typ} : (\typ \lor \neg\typ)\PP$.
\end{proof}
\medskip

Finally, this embedding may be understood as providing
a {\bf computational interpretation} for classical
logic. In fact, besides the introduction and elimination rules
that have been proved above, implication may be defined as an
abbreviation, $(\typ \IMP \typtwo) \eqdef (\neg\typ \lor \typtwo)$,
and witnesses for its introduction rule $\lamc{\var:\typ}{\tm}$ and
its elimination rule $\appc{\tm}{\tmtwo}$ may be defined as follows.
If $\tctx,\var:\typ\PP \vdash \tm : \typtwo\PP$
then $\tctx \vdash \lamc{(\var:\typ)}{\tm} : (\typ \IMP \typtwo)\PP$
where:
\[
{\small
\begin{array}{rcl}
  \lamc{\var}{\tm} & \eqdef &
    \claslamp{(\vartwo:(\typ\IMP\typtwo)\NN)}{
      \inip[2]{
        \tm\sub{\var}{\mathbf{X}_{\vartwo}}
      }
    }
\\
\mathbf{X}_\vartwo & \eqdef &
  \claslamp{(\varthree:\typ\NN)}{
    \clasapp{
      (\negen{(
        \clasapn{
          \mathbf{X'}_{\vartwo,\varthree}
        }{
          \claslamp{(\under:(\neg\typ)\NN)}{
            \negip{
              \varthree
            }
          }
        }
      )})
    }{
      \varthree
    }
  }
\\
\mathbf{X'}_{\vartwo,\varthree} & \eqdef &
  \projip[1]{
    \clasapn{
      \vartwo
    }{
      \claslamp{(\under:(\typ\IMP\typtwo)\NN)}{
        \inip[1]{
          \claslamp{(\under:(\neg\typ)\NN)}{
            \negip{
              \varthree
            }
          }
        }
      }
    }
  }
\end{array}
}
\]
If $\tctx \vdash \tm : (\typ \IMP \typtwo)\PP$
and $\tctx \vdash \tmtwo : \typ\PP$,
then $\tctx \vdash \appc{\tm}{\tmtwo} : \typtwo\PP$,
where:
\[
{\small
\begin{array}{r@{}c@{}l}
  \appc{\tm}{\tmtwo}
  & \eqdef &
  \claslamptable{(\var:\typtwo\NN)}{
    \caseptablex{
      (\clasapp{
        \tm
      }{
        \claslamn{(\under:(\typ\imp\typtwo)\PP)}{
          \pairn{
            (\claslamn{(\under:(\neg\typ)\PP)}{
              \negin{
                \tmtwo
              }
            })
          }{
            \var
          }
        }
      })
    }{
      (\vartwo:(\neg\typ)\PP)
    }{
      \abs{\typtwo\pp}{\tmtwo}{
        \negen{
          (\clasapp{
            \vartwo
          }{
            \claslamn{(\under:(\neg\typ)\PP)}{
              \negin{
                \var
              }
            }
          })
        }
      }
    }{
      (\varthree:\typtwo\PP)
    }{
      \clasapp{\varthree}{\var}
    }
  } 
\end{array}
}
\]

\begin{lemma}
\llem{classical_embedding_computation}
The following hold in $\lambdaCeta$ (with $\ruleEta$ reduction):
\[
\begin{array}{rcl}
  \projic{\pairc{\tm_1}{\tm_2}}
  & \to^* & \tm_i
\\
  \casec{\inic{\tm}}{\var}{\tmtwo_1}{\var}{\tmtwo_2}
  & \to^* &
  \tmtwo_i\sub{\var}{\tm}
\\
  \appc{(\lamc{\var}{\tm})}{\tmtwo}
  & \to^* &
  \tm\sub{\var}{\tmtwo}
\\
  \casec{\lemC{\typ}}{\var}{\tmtwo_1}{\var}{\tmtwo_2}
  & \to^* &
  \claslamp{\vartwo}{
    (\clasapp{
      \tmtwo_2\sub{\var}{\tmtwo^*_1}
    }{
      \vartwo
    })
  }
\end{array}
\]
\hfill where $
  \tmtwo^*_1 :=
    \claslamp{\under}{
      \negip{
        (\claslamn{\var}{
          \abs{}{
            \tmtwo_1
          }{ 
            \vartwo
          }
        })
      }
    }$.
\end{lemma}
\begin{proof}
By calculation.
The last rule describes the behaviour
of the law of excluded middle.
\SeeAppendix{See~\rsec{classical_simulation} in the appendix.}
\end{proof}

\section{Conclusion}
\lsec{prk_conclusion}

This work explores a logical system $\PRK$,
formulated in natural deduction style~(\rdef{system_prk}),
based on a, to the best of our knowledge, new realizability
interpretation for classical logic.
The key idea is that a classical proof of a proposition
can be understood as a transformation from a classical refutation
to a strong proof of the proposition.
We summarize our contributions:
system $\PRK$ has been shown to be {\bf sound}~(\rprop{kripke_soundness})
and {\bf complete}~(\rthm{kripke_completeness}) with respect to a Kripke
semantics.
A calculus $\lambdaC$~(\rdef{lambdaC_type_system}, \rdef{the_lambdaC_calculus}) based on $\PRK$ via the propositions-as-types correspondence,
has been defined.
The calculus enjoys good properties,
the most relevant ones being {\bf confluence}~(\rprop{lambdaC_confluent}),
{\bf subject reduction}~(\rprop{subject_reduction}),
{\bf strong normalization}~(\rthm{lambdaC_canonical}),
and {\bf canonicity}~(\rthm{canonicity}).
Finally, we have shown that $\PRK$ is
a {\bf conservative extension}~(\rprop{prk_conservative}) of classical logic,
and classical logic may be {\bf embedded}~(\rthm{prk_embedding}) in $\PRK$.
This provides a {\bf computational interpretation}~(\rlem{classical_embedding_computation}) for classical logic.
\smallskip

{\bf Future Work.}
It is a natural question whether $\lambdaC$
can be extended to second-order logic.
In fact, formulating such a system is straightforward
by extending the realizability
interpretation described in the introduction with equations:
\[
{\small
  \begin{array}{rcl@{\hspace{.5cm}}rcl}
    (\forall\btyp.\typ)\pp & \approx & \forall\btyp.(\typ\PP)
  &
    (\forall\btyp.\typ)\nn & \approx & \exists\btyp.(\typ\NN)
  \\
    (\exists\btyp.\typ)\pp & \approx & \exists\btyp.(\typ\PP)
  &
    (\exists\btyp.\typ)\nn & \approx & \forall\btyp.(\typ\NN)
  \end{array}
}
\]
For instance, introduction and elimination rules
for positive universal quantification in second-order $\lambdaC$ would be:
\[
{\small
  \indrule{\Iallp}{
    \tctx \vdash \tm : \typ\PP
    \HS
    \btyp \not\in \fv{\tctx}
  }{
    \tctx \vdash \lamtp{\btyp}{\tm} : (\all{\btyp}{\typ})\pp
  }
  \indrule{\Eallp}{
    \tctx \vdash \tm : (\all{\btyp}{\typtwo})\pp
  }{
    \tctx \vdash \apptp{\tm}{\typ} : \typtwo\sub{\btyp}{\typ}\PP
  }
}
\]
From the logical point of view,
the system turns out to be a conservative extension of second-order classical
logic, and from the computational point of view it still enjoys confluence
and subject reduction.
However, the techniques described in this paper do not suffice to prove
strong normalization.
A different normalization proof, possibly based on Tait--Girard's technique
of {\em reducibility candidates}, should be explored.

We have not addressed decision problems,
such as determining the validity of a formula in $\PRK$, corresponding to the
type inhabitation problem for $\lambdaC$.
Unfortunately, $\lambdaC$ does not enjoy the {\em subformula property}.
For a counter\-example consider
$\strongabs{\btyp\pp}{(\strongabs{\btyptwo\pp}{\var}{\vartwo})}{(\strongabs{\btyptwo\nn}{\var}{\vartwo})}$,
which is a normal term of type $\btyp\pp$
under the context $\var:\btyp\pp,\vartwo:\btyp\nn$
such that the unrelated type $\btyptwo\pp$ appears in the derivation.

We have not stated explicitly a computational rule for negation, \ie
for $\negapc{(\neglamc{\var}{\tm})}{\tmtwo}$ in \rlem{classical_embedding_computation}
\SeeAppendix{but see~\rsec{appendix:simulation_of_negation} in the appendix}.
Intriguingly, it does not reduce to $\tm\sub{\var}{\tmtwo}$ in general,
\ie the inference schemes for classical negation that we have
constructed are not consistent with a definition of negation
as $\neg\typ \equiv (\typ \to \bot)$ in intuitionistic logic.
We believe this to be not just an artifact of a faulty construction,
but due to a deeper reason.
\smallskip

{\bf Related Work.}
That classical logic may be embedded in intuitionistic logic has been
known as early as Glivenko's proof of his theorem in the late 1920s.
For a long time, however, the generalized belief seemed to be
that classical proofs had no computational content.
In the late 1980s, Griffin~\cite{griffin1989formulae} remarked that
the type of Felleisen's $\mathcal{C}$ operator
(similar to \texttt{call/cc})
corresponds to Peirce's law $(((\typ \to \typtwo) \to \typ) \to \typ)$.
This sparked research on calculi for classical logic.
Many of these works are based on
classical axioms that behave as {\em control operators}, \ie
operators that can manipulate their computational context.
The literature is abundant on this topic---we limit ourselves
to pointing out some influential works.

Parigot~\cite{lambdamu-parigot} proposes a calculus $\lambda\mu$
based on cut elimination in natural deduction with multiple
conclusions for second-order classical logic.
Its control operator $\mu$ is related with the rule we call contraposition.
The study of $\lambda\mu$ is mature:
topics such as separability~\cite{david2001lambdamu,saurin2005separation,DBLP:conf/csl/Saurin08},
abstract machines~\cite{DBLP:journals/mscs/Groote98},
call-by-need~\cite{pedrot2016classical},
intersection types~\cite{DBLP:journals/lmcs/KesnerV19},
and encodings into linear logic~\cite{DBLP:journals/tcs/Laurent03,DBLP:conf/csl/KesnerBV20} have been developed.

Barbanera and Berardi~\cite{symmetric-Barbanera-berardi} propose
a symmetric $\lambda$-calculus based on a system of natural deduction
including a ``symmetric application'' operator $(\tm \bigstar \tmtwo)$,
closely related to our witness of absurdity ($\strongabs{}{\tm}{\tmtwo}$).
The system of~\cite{symmetric-Barbanera-berardi} is sound and complete
with respect to second-order classical logic,
and it is strongly normalizing, but not confluent.

Curien and Herbelin~\cite{Curien00theduality} derive a
calculus $\bar{\lambda}\mu\tilde{\mu}$
from Gentzen's classical sequent calculus.
This exposes the symmetry between
a program yielding an output and a continuation consuming an input.
The interaction between a program and a continuation,
written $\langle{\tm\mid\tmtwo}\rangle$ is also reminiscent
to our witness of absurdity ($\strongabs{}{\tm}{\tmtwo}$).
Many variants of this system have been studied;
for example, recently, Miquey~\cite{DBLP:journals/toplas/Miquey19} has
extended $\bar{\lambda}\mu\tilde{\mu}$ to incorporate dependent types.

Classical calculi such as
$\lambda\mu$ and $\bar{\lambda}\mu\tilde{\mu}$
are typically translated into the $\lambda$-calculus by means of
continuation-passing style (CPS) translations,
whereas our translation from $\lambdaC$ to the extended System~F
is simpler. Part of the complexity
appears to be factored into the proof that classical inference
schemes hold in $\lambdaC$~(\ie \rlem{classical_embedding_computation}).

The works of Andreoli~\cite{andreoli1992logic}
and Girard~\cite{girard1993unity} in linear logic
introduced the notions of {\em focusing} and {\em polarity},
which allow to formulate linear, intuitionistic,
and classical logic as fragments of a single system (Unified~Logic).
Our notions of positive and negative formulae,
which express affirmation and denial,
should not be confused with the subtler notions of positive and
negative formulae in the sense of polarity.

Krivine~\cite{krivine2009} defines a realizability interpretation
for classical logic using an abstract machine $\lambda_c$
that extends Krivine's abstract machine with further instructions.
This approach is based on the idea that adding logical axioms
corresponds to adding instructions to the machine,
and it has been adapted to provide computational meaning
to reasoning principles such as the
axiom of dependent choice~\cite{DBLP:journals/tcs/Krivine03,herbelin2012constructive}.

Ilik, Lee, and Herbelin~\cite{DBLP:journals/apal/IlikLH10} study a
Kripke semantics for classical logic.
Note that our work in~\rsec{prk_kripke_semantics}
provides a different Kripke semantics for $\PRK$,
and hence for classical logic.
The semantics given in~\cite{DBLP:journals/apal/IlikLH10} and our own
have some similarities, but the relation between them is not obvious.
For example, in~\cite{DBLP:journals/apal/IlikLH10} a Kripke model
involves a relation of ``exploding'' world, which has no counterpart in
our system.

\ifCLASSOPTIONcompsoc
  \section*{Acknowledgments}
\else
  \section*{Acknowledgment}
\fi
This work was partially supported by
project grant ECOS~Sud~A17C01.
The authors would like to thank Eduardo\- Bo\-ne\-lli
and the anonymous reviewers for feedback on an early draft.

\bibliographystyle{IEEEtran}
\bibliography{biblio}

\newpage

\begin{alphasection}

\section{Technical Appendix}

\subsection{Proof of the Projection Lemma~(\rlem{projection_lemma})}
\lsec{appendix:projection_lemma}

\begin{lemma}[Projection]
If $\tctx,\ev \vdash \evtwo$
then $\tctx,\trunc{\ev} \vdash \trunc{\evtwo}$.
\end{lemma}
\begin{proof}
We call $\ev$ the {\em target assumption}.
The proof proceeds by induction on the derivation of $\tctx,\ev \vdash \evtwo$.
We only study the cases with positive signs, the negative cases are symmetric.

\Case{$\Ax$}:
  let $\tctx,\evtwo \vdash \evtwo$.
  There are two cases, depending on whether the target assumption is
  in $\tctx$ or not.
  \begin{enumerate}
  \item
    {\em If the target assumption is in $\tctx$, \ie $\tctx = \tctx',\ev$.}
    Note that we have $\tctx',\trunc{\ev},\evtwo \vdash \evtwo$
    by the $\Ax$ rule.
    By truncating the conclusion (\rlem{projection_of_conclusions})
    we conclude that $\tctx',\trunc{\ev},\evtwo \vdash \trunc{\evtwo}$,
    as required.
  \item
    {\em If the target assumption is $\evtwo$.}
    Then we have that $\tctx,\trunc{\evtwo} \vdash \trunc{\evtwo}$
    by the $\Ax$ rule.
  \end{enumerate}

\Case{$\Abs$}:
  let $\tctx,\ev \vdash \evtwo$ be derived
  from $\tctx,\ev \vdash \evthree$ and $\tctx,\ev \vdash \evthree\OP$
  for some strong proposition $\evthree$.
  By \ih we have that
  $\tctx,\trunc{\ev} \vdash \trunc{\evthree}$ and $\tctx,\trunc{\ev} \vdash \trunc{\evthree\OP}$
  so by the generalized absurdity rule ($\Abs'$) we have that
  $\tctx,\trunc{\ev} \vdash \trunc{\evtwo}$.

\Case{\Iandp}:
  let $\tctx,\ev \vdash (\typ\land\typtwo)\pp$
  be derived from $\tctx,\ev \vdash \typ\PP$ and $\tctx,\ev \vdash \typ\PP$.
  By \ih, $\tctx,\trunc{\ev} \vdash \typ\PP$ and $\tctx,\trunc{\ev} \vdash \typ\PP$.
  By the $\Iandp$ rule,
  $\tctx,\trunc{\ev} \vdash (\typ\land\typtwo)\pp$.
  Projecting the conclusion (\rlem{projection_of_conclusions}),
  $\tctx,\trunc{\ev} \vdash (\typ\land\typtwo)\PP$
  as required.

\Case{$\Eandp$}:
  let $\tctx,\ev \vdash \typ_i\PP$ be derived from
  $\tctx,\ev \vdash (\typ_1\land\typ_2)\pp$.
  Then the proof is of the form:
  \[
    \indruleN{\ClStr}{
      \indruleN{\Eandp}{
        \indruleN{\Ecp}{
          \derivdots{\derivation}
          \HS
          \derivdots{\derivationtwo}
        }{
          \tctx,\trunc{\ev},\typ_i\NN \vdash (\typ_1\land\typ_2)\pp
        }
      }{
        \tctx,\trunc{\ev},\typ_i\NN \vdash \typ_i\PP
      }
    }{
      \tctx,\trunc{\ev} \vdash \typ_i\PP
    }
  \]
  where:
  \[
    \begin{array}{rcl}
    \derivation & \eqdef &
      \indruleNParen{\Weakening}{
        \indruleN{}{\text{\ih}}{\tctx,\trunc{\ev} \vdash (\typ_1\land\typ_2)\PP}
      }{
        \tctx,\trunc{\ev},\typ_i\NN \vdash (\typ_1\land\typ_2)\PP
      }
    \\
    \derivationtwo & \eqdef &
      \indruleNParen{\PrCon}{
        \indruleN{\Iandn}{
          \indruleN{\Ax}{}{\tctx,\trunc{\ev},\typ_i\NN \vdash \typ_i\NN}
        }{
          \tctx,\trunc{\ev},\typ_i\NN \vdash (\typ_1\land\typ_2)\nn
        }
      }{
        \tctx,\trunc{\ev},\typ_i\NN \vdash (\typ_1\land\typ_2)\NN
      }
    \end{array}
  \]

\Case{$\Iorp$}:
  let $\tctx,\ev \vdash (\typ_1 \lor \typ_2)\pp$
  be derived from $\tctx,\ev \vdash \typ_i\PP$.
  By \ih, $\tctx,\trunc{\ev} \vdash \typ_i\PP$.
  By the $\Iorp$ rule, $\tctx,\trunc{\ev} \vdash (\typ_1 \lor \typ_2)\pp$.
  Projecting the conclusion~(\rlem{projection_of_conclusions}),
  $\tctx,\trunc{\ev} \vdash (\typ_1 \lor \typ_2)\PP$.

\Case{$\Eorp$}:
  let $\tctx,\ev \vdash \evtwo$
  be derived from
  $\tctx,\ev \vdash (\typ_1\lor\typ_2)\pp$
  and
  $\tctx,\ev,\typ_i\PP \vdash \evtwo$ for each $i \in \set{1,2}$.
  By \ih,
  $\tctx,\trunc{\ev} \vdash (\typ_1\lor\typ_2)\PP$
  and
  $\tctx,\trunc{\ev},\typ_i\PP \vdash \trunc{\evtwo}$
  for each $i \in \set{1,2}$. Then the proof is of the form:
  \[
    \indruleN{\ClStr}{
      \indruleN{\Eorp}{
        \indruleN{\Ecp}{
          \derivdots{\derivationthree}
          \HS\HS
          \derivdots{\derivationtwo}
          \HS
        }{
          \tctx,\trunc{\ev},\trunc{\evtwo\OP} \vdash (\typ_1\lor\typtwo_2)\pp
        }
        \HS
        \derivdots{\derivation_1} \HS \derivdots{\derivation_2}
        \HS
      }{
        \tctx,\trunc{\ev},\trunc{\evtwo\OP} \vdash \trunc{\evtwo}
      }
    }{
      \tctx,\trunc{\ev} \vdash \trunc{\evtwo}
    }
  \]
  where:
  \[
    \begin{array}{rcl}
    \derivationthree & \eqdef &
      \indruleNParen{\Weakening}{
        \indruleN{}{\text{\ih}}{\tctx,\trunc{\ev} \vdash (\typ_1\lor\typ_2)\PP}
      }{
        \tctx,\trunc{\ev},\trunc{\evtwo\OP} \vdash (\typ_1\lor\typ_2)\PP
      }
    \\
    \\
    \derivationtwo & \eqdef &
      \indruleNParen{\Iorn}{
        \derivdots{\derivationtwo_1} \HS \derivdots{\derivationtwo_2}
      }{
        \tctx,\trunc{\ev},\trunc{\evtwo\OP} \vdash (\typ_1\lor\typ_2)\NN
      }
    \end{array}
  \]
  and for each $i \in \set{1,2}$ the derivations
  $\derivation_i$ and $\derivationtwo_i$ are given by:
  \[
    \begin{array}{rcl}
    \derivation_i & \eqdef &
    \indruleNParen{\Weakening}{
      \indruleN{}{\text{\ih}}{\tctx,\trunc{\ev},\typ_i\PP \vdash \trunc{\evtwo}}
    }{
      \tctx,\trunc{\ev},\trunc{\evtwo\OP},\typ_i\PP \vdash \trunc{\evtwo}
    }
    \\
    \\
    \derivationtwo_i & \eqdef &
    \indruleNParen{\Contrapose}{
      \indruleN{}{\text{\ih}}{\tctx,\trunc{\ev},\typ_i\PP \vdash \trunc{\evtwo}}
    }{
      \tctx,\trunc{\ev},\trunc{\evtwo\OP} \vdash \typ_i\NN
    }
    \end{array}
  \]

\Case{$\Inotp$}:
  let $\tctx,\ev \vdash (\neg\typ)\pp$ be derived from $\tctx,\ev \vdash \typ\NN$.
  By \ih we have that $\tctx,\trunc{\ev} \vdash \typ\NN$.
  By the $\Inotp$ rule,
  $\tctx,\trunc{\ev} \vdash (\neg\typ)\pp$.
  Projecting the conclusion~(\rlem{projection_of_conclusions}),
  $\tctx,\trunc{\ev} \vdash (\neg\typ)\PP$.

\Case{$\Enotp$}:
  let $\tctx,\ev \vdash \typ\NN$ be derived from $\tctx,\ev \vdash (\neg\typ)\pp$.
  Then the proof is of the form:
  \[
  \hspace{-.5cm}
    \indruleN{\ClStr}{
      \indruleN{\Enotp}{
        \indruleN{\Ecp}{
          \derivdots{\derivation}
          \HS
          \derivdots{\derivationtwo}
        }{
          \tctx,\trunc{\ev},\typ\PP \vdash (\neg\typ)\pp
        }
      }{
        \tctx,\trunc{\ev},\typ\PP \vdash \typ\NN
      }
    }{
      \tctx,\trunc{\ev} \vdash \typ\NN
    }
  \]
  where:
  \[
    \begin{array}{rcl}
      \derivation & \eqdef &
        \indruleNParen{\Weakening}{
          \indruleN{}{\text{\ih}}{\tctx,\trunc{\ev} \vdash (\neg\typ)\PP}
        }{
          \tctx,\trunc{\ev},\typ\PP \vdash (\neg\typ)\PP
        }
    \\
    \\
      \derivationtwo & \eqdef &
        \indruleNParen{\Icn}{
          \indruleN{\Inotn}{
            \indruleN{\Ax}{
            }{
              \tctx,\trunc{\ev},\typ\PP,(\neg\typ)\PP \vdash \typ\PP
            }
          }{
            \tctx,\trunc{\ev},\typ\PP,(\neg\typ)\PP \vdash (\neg\typ)\nn
          }
        }{
          \tctx,\trunc{\ev},\typ\PP \vdash (\neg\typ)\NN
        }
    \end{array}
  \]

\Case{$\Icp$}:
  let $\tctx,\ev \vdash \typ\PP$
  be derived from $\tctx,\ev,\typ\NN \vdash \typ\pp$.
  By \ih, $\tctx,\trunc{\ev},\typ\NN \vdash \typ\PP$,
  so by \rlem{classical_strengthening} we have that
  $\tctx,\trunc{\ev} \vdash \typ\PP$.

\Case{$\Ecp$}:
  let $\tctx,\ev \vdash \typ\pp$
  be derived from $\tctx,\ev \vdash \typ\PP$
  and $\tctx,\ev \vdash \typ\NN$.
  Then, in particular, by \ih on the first premise,
  we have $\tctx,\trunc{\ev} \vdash \typ\PP$,
  as required.
\end{proof}

\subsection{Proof of Properties of Forcing~(\rlem{properties_forcing})}
\lsec{appendix:properties_forcing}

\begin{lemma}[Monotonicity of forcing]
\llem{appendix:monotonicity_forcing}
If $\kripforce{\wl}{\ev}$
and $\wl \wleq \wl'$
then $\kripforce{\wl'}{\ev}$.
\end{lemma}
\begin{proof}
By induction on the measure $\#(\ev)$.
We only check the positive propositions;
the negative cases are dual---\eg the proof for $(\typ\land\typtwo)\nn$
is symmetric to the proof for $(\typ\lor\typtwo)\pp$:
\begin{enumerate}
\item
  {\em Propositional variable, $\ev = \btyp\pp$.}
  Let $\kripforce{\wl}{\btyp\pp}$,
  that is $\btyp \in \wlpos{\wl}$.
  Then by the monotonicity property we have that $\btyp \in \wlpos{\wl'}$,
  so $\kripforce{\wl'}{\btyp\pp}$.
\item
  {\em Conjunction, $\ev = (\typ\land\typtwo)\pp$.}
  Let $\kripforce{\wl}{(\typ\land\typtwo)\pp}$,
  that is
  $\kripforce{\wl}{\typ\PP}$ and $\kripforce{\wl}{\typtwo\PP}$.
  Then by \ih
  $\kripforce{\wl'}{\typ\PP}$ and $\kripforce{\wl'}{\typtwo\PP}$
  so $\kripforce{\wl'}{(\typ\land\typtwo)\pp}$.
\item
  {\em Disjunction, $\ev = (\typ\lor\typtwo)\pp$.}
  Let $\kripforce{\wl}{(\typ\lor\typtwo)\pp}$,
  that is
  $\kripforce{\wl}{\typ\PP}$ or $\kripforce{\wl}{\typtwo\PP}$.
  We consider the two possibilities.
  On one hand, if $\kripforce{\wl}{\typ\PP}$ 
  then by \ih $\kripforce{\wl'}{\typ\PP}$
  so $\kripforce{\wl'}{(\typ\lor\typtwo)\pp}$.
  On the other hand, if $\kripforce{\wl}{\typtwo\PP}$ 
  then by \ih $\kripforce{\wl'}{\typtwo\PP}$
  so $\kripforce{\wl'}{(\typ\lor\typtwo)\pp}$.
\item
  {\em Negation, $\ev = (\neg\typ)\pp$.}
  Let $\kripforce{\wl}{(\neg\typ)\pp}$,
  that is $\kripforce{\wl}{\typ\NN}$.
  Then by \ih $\kripforce{\wl'}{\typ\NN}$
  so $\kripforce{\wl'}{(\neg\typ)\pp}$.
\item
  {\em Classical proposition, $\ev = \typ\PP$.}
  Let $\kripforce{\wl}{\typ\PP}$,
  that is, for every $\wl'' \wgeq \wl$
  we have that $\kripnotforce{\wl''}{\typ\nn}$.
  Our goal is to prove that $\kripforce{\wl'}{\typ\PP}$,
  so let $\wl'' \wgeq \wl'$
  and let us check that $\kripnotforce{\wl''}{\typ\nn}$.
  Indeed, given that $\wl'' \wgeq \wl' \wgeq \wl$
  we have that $\kripnotforce{\wl''}{\typ\nn}$.
\end{enumerate}
\end{proof}

\begin{lemma}[Stabilization of forcing]
\llem{appendix:stabilization_forcing}
For every world $\wl$ and every proposition $\ev$,
there is a world $\wl' \wgeq \wl$
such that
either $\kripforce{\wl'}{\ev}$
or $\kripforce{\wl'}{\ev\OP}$,
but not both.
\end{lemma}
\begin{proof}
By induction on the measure $\#(\ev)$.
We only check the positive propositions;
the negative cases are dual.
\begin{enumerate}
\item
  {\em Propositional variable, $\ev = \btyp\pp$ and $\ev\OP = \btyp\nn$.}
  By the stabilization property,
  there exists $\wl' \wgeq \wl$
  such that $\btyp \in \wlpos{\wl'} \symdiff \wlneg{\wl'}$,
  \ie $\btyp \in \wlpos{\wl'}$ or $\btyp \in \wlneg{\wl'}$ but not both,
  so we consider two cases:
  \begin{enumerate}
  \item
    If $\btyp \in \wlpos{\wl'} \setminus \wlneg{\wl'}$
    then $\kripforce{\wl'}{\btyp\pp}$ and $\kripnotforce{\wl'}{\btyp\nn}$.
  \item
    If $\btyp \in \wlneg{\wl'} \setminus \wlpos{\wl'}$
    then $\kripforce{\wl'}{\btyp\nn}$ and $\kripnotforce{\wl'}{\btyp\pp}$.
  \end{enumerate}
\item
  {\em Conjunction,
       $\ev = (\typ\land\typtwo)\pp$ and $\ev\OP = (\typ\land\typtwo)\nn$.}
  By \ih there is a world $\wl_1 \wgeq \wl$
  such that either $\kripforce{\wl_1}{\typ\PP}$
  or $\kripforce{\wl_1}{\typ\NN}$ but not both,
  so we consider two subcases:
  \begin{enumerate}
  \item
    If $\kripforce{\wl_1}{\typ\PP}$ and $\kripnotforce{\wl_1}{\typ\NN}$,
    then by \ih there is a world $\wl_2 \wgeq \wl_1$
    such that either $\kripforce{\wl_2}{\typtwo\PP}$
    or $\kripforce{\wl_2}{\typtwo\NN}$ but not both,
    so we consider two further subcases:
    \begin{enumerate}
    \item
      If $\kripforce{\wl_2}{\typtwo\PP}$ and $\kripnotforce{\wl_2}{\typtwo\NN}$,
      then we take $\wl' := \wl_2$.
      By monotonicity~(\rlem{appendix:monotonicity_forcing})
      we have that $\kripforce{\wl_2}{\typ\PP}$
      so indeed $\kripforce{\wl_2}{(\typ\land\typtwo)\pp}$.
      We are left to show that $\kripnotforce{\wl_2}{(\typ\land\typtwo)\nn}$.
      We already know that $\kripnotforce{\wl_2}{\typtwo\NN}$,
      so to conclude it suffices to show that $\kripnotforce{\wl_2}{\typ\NN}$.
      Indeed, suppose that $\kripforce{\wl_2}{\typ\NN}$ holds.
      By \ih there exists $\wl_3 \wgeq \wl_2$ such that
      either $\kripforce{\wl_3}{\typ\PP}$ or $\kripforce{\wl_3}{\typ\NN}$
      but {\em not both}.
      However, by monotonicity~(\rlem{appendix:monotonicity_forcing})
      ---given that both
      $\kripforce{\wl_2}{\typ\PP}$ and $\kripforce{\wl_2}{\typ\NN}$ hold---
      we know that
      both $\kripforce{\wl_3}{\typ\PP}$ and $\kripforce{\wl_3}{\typ\NN}$ 
      hold,
      a contradiction.
    \item
      If $\kripforce{\wl_2}{\typtwo\NN}$ and $\kripnotforce{\wl_2}{\typtwo\PP}$,
      then we take $\wl' := \wl_2$,
      and we have that $\kripforce{\wl_2}{(\typ\land\typtwo)\nn}$
      and $\kripnotforce{\wl_2}{(\typ\land\typtwo)\pp}$.
    \end{enumerate}
  \item
    If $\kripforce{\wl_1}{\typ\NN}$ and $\kripnotforce{\wl_1}{\typ\PP}$,
    then we take $\wl' := \wl_1$,
    and we have that $\kripforce{\wl_1}{(\typ\land\typtwo)\nn}$
    and $\kripnotforce{\wl_1}{(\typ\land\typtwo)\pp}$.
  \end{enumerate}
\item
  {\em Disjunction, $\ev = (\typ\lor\typtwo)\pp$ and $\ev\OP = (\typ\lor\typtwo)\nn$.}
  By \ih there is a world $\wl_1 \wgeq \wl$ such that
  either $\kripforce{\wl_1}{\typ\PP}$ or $\kripforce{\wl_1}{\typ\NN}$
  but not both, so we consider two subcases:
  \begin{enumerate}
  \item
    If $\kripforce{\wl_1}{\typ\PP}$ and $\kripnotforce{\wl_1}{\typ\NN}$,
    then we take $\wl' := \wl_1$,
    and we have that $\kripforce{\wl_1}{(\typ\lor\typtwo)\pp}$
    and $\kripnotforce{\wl_1}{(\typ\lor\typtwo)\nn}$.
  \item
    If $\kripforce{\wl_1}{\typ\NN}$ and $\kripnotforce{\wl_1}{\typ\PP}$,
    then by \ih there is a world $\wl_2 \wgeq \wl_1$ such that
    either $\kripforce{\wl_2}{\typtwo\PP}$ or $\kripforce{\wl_2}{\typtwo\NN}$
    but not both, so we consider two further subcases:
    \begin{enumerate}
    \item
      If $\kripforce{\wl_2}{\typtwo\PP}$ and $\kripnotforce{\wl_2}{\typtwo\NN}$,
      then we take $\wl' := \wl_2$, and we have that
      $\kripforce{\wl_2}{(\typ\lor\typtwo)\pp}$
      and
      $\kripnotforce{\wl_2}{(\typ\lor\typtwo)\nn}$.
    \item
      If $\kripforce{\wl_2}{\typtwo\NN}$ and $\kripnotforce{\wl_2}{\typtwo\PP}$,
      then we take $\wl' := \wl_2$.
      By monotonicity~(\rlem{appendix:monotonicity_forcing})
      we have that $\kripforce{\wl_2}{\typ\NN}$
      so indeed $\kripforce{\wl_2}{(\typ\lor\typtwo)\nn}$.
      We are left to show that $\kripnotforce{\wl_2}{(\typ\lor\typtwo)\pp}$.
      We already know that $\kripnotforce{\wl_2}{\typtwo\PP}$,
      so we are left to show that $\kripnotforce{\wl_2}{\typ\PP}$.
      Indeed, suppose that $\kripforce{\wl_2}{\typ\PP}$ holds.
      By \ih there exists $\wl_3 \wgeq \wl_2$ such that
      either $\kripforce{\wl_3}{\typ\PP}$ or $\kripforce{\wl_3}{\typ\NN}$ 
      holds but {\em not both}.
      However, by monotonicity~(\rlem{appendix:monotonicity_forcing})
      ---given that both
      $\kripforce{\wl_2}{\typ\PP}$ and $\kripforce{\wl_2}{\typ\NN}$ hold---
      we know that both
      $\kripforce{\wl_3}{\typ\PP}$ and $\kripforce{\wl_3}{\typ\NN}$ hold,
      a contradiction.
    \end{enumerate}
  \end{enumerate}
\item
  {\em Negation, $\ev = (\neg\typ)\pp$ and $\ev\OP = (\neg\typ)\nn$.}
  By \ih there is a world $\wl' \wgeq \wl$
  such that either $\kripforce{\wl'}{\typ\PP}$ or $\kripforce{\wl'}{\typ\NN}$
  hold but not both, so we consider two cases:
  \begin{enumerate}
  \item
    If $\kripforce{\wl'}{\typ\PP}$ and $\kripnotforce{\wl'}{\typ\NN}$,
    then $\kripforce{\wl'}{(\neg\typ)\nn}$
    and $\kripnotforce{\wl'}{(\neg\typ)\pp}$.
  \item
    If $\kripforce{\wl'}{\typ\NN}$ and $\kripnotforce{\wl'}{\typ\PP}$,
    then $\kripforce{\wl'}{(\neg\typ)\pp}$
    and $\kripnotforce{\wl'}{(\neg\typ)\nn}$.
  \end{enumerate}
\item
  {\em Classical proposition, $\ev = \typ\PP$ and $\ev\OP = \typ\NN$.}
  By \ih there is a world $\wl' \wgeq \wl$ such that
  either $\kripforce{\wl'}{\typ\pp}$ or $\kripforce{\wl'}{\typ\nn}$
  but not both.
  We consider two subcases:
  \begin{enumerate}
  \item
    If $\kripforce{\wl'}{\typ\pp}$ and $\kripnotforce{\wl'}{\typ\nn}$,
    then we claim that
    $\kripforce{\wl'}{\typ\PP}$ and $\kripnotforce{\wl'}{\typ\NN}$.
    Indeed, let us prove each condition:
    \begin{enumerate}
    \item
      In order to show that $\kripforce{\wl'}{\typ\PP}$,
      it suffices to check that given
      $\wl'' \wgeq \wl'$ we have that $\kripnotforce{\wl''}{\typ\nn}$.
      Indeed, suppose that $\kripforce{\wl''}{\typ\nn}$.
      Then by \ih there exists $\wl''' \wgeq \wl''$
      such that either $\kripforce{\wl'''}{\typ\pp}$
      or $\kripforce{\wl'''}{\typ\nn}$ but {\em not both}.
      However, by monotonicity~(\rlem{appendix:monotonicity_forcing})
      ---given that both $\kripforce{\wl'}{\typ\pp}$
         and $\kripforce{\wl''}{\typ\nn}$ hold,
         and $\wl' \wleq \wl'' \wleq \wl'''$---
      we know that both
      $\kripforce{\wl'''}{\typ\pp}$
      and
      $\kripforce{\wl'''}{\typ\nn}$
      hold, a contradiction.
    \item
      In order to show that $\kripnotforce{\wl'}{\typ\NN}$,
      it suffices to note that $\kripforce{\wl'}{\typ\pp}$,
      which contradicts the definition of $\kripforce{\wl'}{\typ\NN}$,
      given that accessibility is reflexive, \ie $\wl' \wleq \wl'$.
    \end{enumerate}
  \item
    If $\kripforce{\wl'}{\typ\nn}$ and $\kripnotforce{\wl'}{\typ\pp}$,
    then we claim that
    $\kripforce{\wl'}{\typ\NN}$ and $\kripnotforce{\wl'}{\typ\PP}$.
    Indeed, let us prove each condition:
    \begin{enumerate}
    \item
      In order to show that $\kripforce{\wl'}{\typ\NN}$,
      it suffices to check that given $\wl'' \wgeq \wl'$
      we have that $\kripnotforce{\wl''}{\typ\pp}$.
      Indeed, suppose that $\kripforce{\wl''}{\typ\pp}$.
      Then by \ih there exists $\wl''' \wgeq \wl''$
      such that either $\kripforce{\wl'''}{\typ\pp}$ and
      $\kripforce{\wl'''}{\typ\nn}$ but {\em not both}.
      However, by monotonicity~(\rlem{appendix:monotonicity_forcing})
      ---given that both
         $\kripforce{\wl''}{\typ\pp}$ and $\kripforce{\wl'}{\typ\nn}$ hold,
         and $\wl' \wleq \wl'' \wleq \wl'''$---
      we know that both $\kripforce{\wl'''}{\typ\pp}$
      and $\kripforce{\wl'''}{\typ\nn}$ hold, a contradiction.
    \item
      In order to show that $\kripnotforce{\wl'}{\typ\PP}$
      it suffices to note that $\kripforce{\wl'}{\typ\nn}$,
      which contradicts the definition of $\kripforce{\wl'}{\typ\PP}$,
      given that accessibility is reflexive, \ie $\wl' \wleq \wl'$.
    \end{enumerate}
  \end{enumerate}
\end{enumerate}
\end{proof}

\begin{lemma}[Non-contradiction of forcing]
\llem{appendix:non_contradiction_forcing}
If $\kripforce{\wl}{\ev}$ then $\kripnotforce{\wl}{\ev\OP}$.
\end{lemma}
\begin{proof}
Suppose that both $\kripforce{\wl}{\ev}$ and $\kripforce{\wl}{\ev\OP}$ hold.
By stabilization~(\rlem{appendix:stabilization_forcing})
there is a world $\wl' \wgeq \wl$
such that
either $\kripforce{\wl'}{\ev}$ or $\kripforce{\wl'}{\ev\OP}$
but {\em not both}.
However, by monotonicity~(\rlem{appendix:monotonicity_forcing})
we know that both
$\kripforce{\wl'}{\ev}$ and $\kripforce{\wl'}{\ev\OP}$ must hold,
a contradiction.
\end{proof}

\subsection{Proof of Soundness of $\PRK$ with respect to the Kripke semantics}
\lsec{appendix:kripke_soundness}

\begin{lemma}[Rule of classical forcing]
\llem{appendix:rule_of_classical_forcing}
\quad
\begin{enumerate}
\item $(\kripforce{\wl}{\typ\PP})$
      if and only if,
      for all $\wl' \wgeq \wl$,
      $(\kripforce{\wl'}{\typ\NN})$ implies
      $(\kripforce{\wl'}{\typ\pp})$.
\item $(\kripforce{\wl}{\typ\NN})$
      if and only if,
      for all $\wl' \wgeq \wl$,
      $(\kripforce{\wl'}{\typ\PP})$
      implies
      $(\kripforce{\wl'}{\typ\nn})$.
\end{enumerate}
\end{lemma}
\begin{proof}
We only prove the first item.
The second one is symmetric, flipping all the signs.
\begin{itemize}
\item[$(\Rightarrow)$]
  Suppose that $\kripforce{\wl}{\typ\PP}$,
  let $\wl' \wgeq \wl$,
  and let us show that the implication
  $(\kripforce{\wl'}{\typ\NN}) \implies (\kripforce{\wl'}{\typ\pp})$
  holds.
  In fact, the implication holds vacuously,
  given that
  $\kripforce{\wl'}{\typ\PP}$ by monotonicity~(\rlem{monotonicity_forcing}),
  and therefore
  $\kripnotforce{\wl'}{\typ\NN}$ by non-contradiction~(\rlem{non_contradiction_forcing}).
\item[$(\Leftarrow)$]
  Suppose that for every $\wl' \wgeq \wl$
  the implication
  $(\kripforce{\wl'}{\typ\NN}) \implies (\kripforce{\wl'}{\typ\pp})$
  holds.
  Let us show that $\kripforce{\wl}{\typ\PP}$ holds,
  \ie that for every $\wl' \wgeq \wl$
  we have that $\kripnotforce{\wl'}{\typ\nn}$.
  Let $\wl'$ be a world such that $\wl' \wgeq \wl$
  and, by contradiction, suppose that $\kripforce{\wl'}{\typ\nn}$.
  Then by non-contradiction~(\rlem{non_contradiction_forcing}) we
  have that $\kripnotforce{\wl'}{\typ\pp}$.
  Hence, to obtain a contradiction, using the implication of the hypothesis,
  it suffices to show that $\kripforce{\wl'}{\typ\NN}$,
  that is, that for every $\wl'' \wgeq \wl'$ we have that
  $\kripnotforce{\wl''}{\typ\pp}$.
  Indeed, let $\wl'' \wgeq \wl'$.
  By monotonicity~(\rlem{monotonicity_forcing}) $\kripforce{\wl''}{\typ\nn}$,
  so by non-contradiction~(\rlem{non_contradiction_forcing}) $\kripnotforce{\wl''}{\typ\pp}$,
  as required.
\end{itemize}
\end{proof}

\begin{proposition}[Soundness]
If $\tctx \vdash \ev$ is provable in $\PRK$, then $\kripentails{\tctx}{\ev}$.
\end{proposition}
\begin{proof}
By induction on the derivation of $\tctx \vdash \ev$.
The axiom rule, and the introduction and elimination rules for
conjunction, disjunction, and negation are straightforward
using the definition of Kripke model.
The interesting cases are the following rules:
\begin{enumerate}
\item
  \indrulename{\Abs}:
  let $\tctx \vdash \evtwo$
  be derived from $\tctx \vdash \ev$ and $\tctx \vdash \ev\OP$
  for some strong proposition $\ev$.
  Suppose that $\kripforcefull{\wl}{\tctx}$ holds in an arbitrary world $\wl$
  under an arbitrary Kripke model $\kripmod$,
  and let us show that $\kripforce{\wl}{\evtwo}$.
  Note that by \ih we have that
  $\kripforce{\wl}{\ev}$ and $\kripforce{\wl}{\ev\OP}$.
  But this is impossible by non-contradiction~(\rlem{non_contradiction_forcing}).
  Hence $\kripforce{\wl}{\evtwo}$.
\item
  \indrulename{\Icp}:
  let $\tctx \vdash \typ\PP$
  be derived from $\tctx,\typ\NN \vdash \typ\pp$.
  Suppose that $\kripforcefull{\wl}{\tctx}$ holds
  in an arbitrary world $\wl$
  under an arbitrary Kripke model $\kripmod$,
  and let us show that $\kripforce{\wl}{\typ\PP}$.
  We claim that for every $\wl' \wgeq \wl$ the implication
  $(\kripforce{\wl'}{\typ\NN}) \implies (\kripforce{\wl'}{\typ\pp})$
  holds.
  Indeed, suppose that $\kripforce{\wl'}{\typ\NN}$.
  Moreover, by monotonicity~(\rlem{monotonicity_forcing}),
  we have that $\kripforcefull{\wl'}{\tctx}$.
  So $\kripforcefull{\wl'}{\tctx,\typ\NN}$ holds.
  Hence by \ih we have that $\kripforce{\wl'}{\typ\pp}$.
  Given that the implication
  $(\kripforce{\wl'}{\typ\NN}) \implies (\kripforce{\wl'}{\typ\pp})$
  holds for all $\wl' \wgeq \wl$,
  using the rule of classical forcing (\rlem{rule_of_classical_forcing})
  we conclude that $\kripforce{\wl}{\typ\PP}$,
  as required.
\item
  \indrulename{\Icn}:
  similar to the \indrulename{\Icp} case.
\item
  \indrulename{\Ecp}, \indrulename{\Ecn}:
  similar to the \indrulename{\Abs} case.
\end{enumerate}
\end{proof}

\subsection{Auxiliary lemmas to prove Completeness of $\PRK$ with respect to the Kripke semantics}
\lsec{appendix:kripke_completeness}

In the following proof we use an encoding of falsity
with the pure proposition $\Bot \eqdef (\btyp_0 \land \neg\btyp_0)$
for some fixed propositional variable $\btyp_0$.
Remark that $\tctx \vdash \Bot\NN$ is provable, being an instance
of the law of non-contradiction~(\rexample{lem_and_noncontr}).

\begin{lemma}[Consistent extension]
\llem{appendix:consistent_extension}
Let $\tctx$ be a consistent set, and let $\ev$ be a proposition.
Then $\tctx \cup \set{\ev}$ and $\tctx \cup \set{\ev\OP}$
are not both inconsistent.
\end{lemma}
\begin{proof}
Suppose that $\tctx \cup \set{\ev}$
and $\tctx \cup \set{\ev\OP}$ are both inconsistent.
In particular we have that
$\tctx,\ev \vdash \Bot\PP$
and
$\tctx,\ev\OP \vdash \Bot\PP$.
By the projection lemma~(\rlem{projection_lemma})
we have that $\tctx,\trunc{\ev} \vdash \Bot\PP$
and $\tctx,\trunc{\ev\OP} \vdash \Bot\PP$.
Moreover, by contraposition~(\rlem{admissible_rules_logic})
we have that
$\tctx,\Bot\NN \vdash \trunc{\ev\OP}$ and
$\tctx,\Bot\NN \vdash \trunc{\ev}$.
Since $\Bot\NN$ is provable~(\rexample{lem_and_noncontr}),
applying the cut rule~(\rlem{admissible_rules_logic})
we have that $\tctx \vdash \trunc{\ev\OP}$ and $\tctx \vdash \trunc{\ev}$.
The generalized absurdity rule
allows us to derive $\tctx \vdash \evtwo$
for any $\evtwo$ from these two sequents,
so $\tctx$ is inconsistent.
This contradicts the hypothesis that $\tctx$ is consistent.
\end{proof}

\begin{lemma}[Saturation]
\llem{appendix:kripke_saturation}
Let $\tctx$ be a consistent set of propositions,
and let $\evtwo$ be a proposition such that $\tctx \nvdash \evtwo$.
Then there exists a prime theory $\tctx' \supseteq \tctx$ such that
$\tctx' \nvdash \evtwo$.
\end{lemma}
\begin{proof}
Consider an enumeration of all propositions $(\ev_1,\ev_2,\hdots)$.
We build a sequence of sets
$\tctx = \tctx_0 \subseteq \tctx_1 \subseteq \tctx_2 \subseteq \hdots$,
with the invariant that $\tctx_n \nvdash \evtwo$ for all $n \geq 0$,
according to the following construction.

In the $n$-th step, suppose that $\tctx_1,\hdots,\tctx_n$ have already been
constructed, and consider the first proposition $\ev$ in the enumeration
such that $\tctx_n \vdash \ev$ but the disjunctive property fails for $\ev$,
that is,
either $\ev$ is of the form $(\typ \lor \typtwo)\pp$ with $\typ\PP,\typtwo\PP \notin \tctx_n$
or $\ev$ is of the form $(\typ \land \typtwo)\nn$ with $\typ\NN,\typtwo\NN \notin \tctx_n$.
There are two subcases:
\begin{enumerate}
\item
  If $\ev = (\typ \lor \typtwo)\pp$ with $\typ\PP,\typtwo\PP \notin \tctx_n$,
  note that $\tctx_n,\typ\PP \vdash \evtwo$ and $\tctx_n,\typtwo\PP \vdash \evtwo$
  cannot both hold simultaneously.
  Indeed, if both $\tctx_n,\typ\PP \vdash \evtwo$ and $\tctx_n,\typtwo\PP \vdash \evtwo$ hold,
  given that also $\tctx_n \vdash (\typ\lor\typtwo)\pp$,
  applying $\Eorp$ we would have $\tctx_n \vdash \evtwo$, contradicting the hypothesis.
  Hence we may define $\tctx_{n+1}$ as follows:
  \[
    \tctx_{n+1} \eqdef
      \begin{cases}
        \tctx_n \cup \set{\typ\PP}    & \text{if $\tctx_n,\typ\PP \nvdash \evtwo$} \\
        \tctx_n \cup \set{\typtwo\PP} & \text{otherwise} \\
      \end{cases}
  \]
  Note that,
  in the second case, $\tctx_n,\typtwo\PP \nvdash \evtwo$ holds.
\item
  If $\ev = (\typ \land \typtwo)\nn$ with $\typ\NN,\typtwo\NN \notin \tctx_n$,
  the construction is similar,
  defining $\tctx_{n+1}$ as
  either $\tctx_n \cup \set{\typ\NN}$
  or $\tctx_n \cup \set{\typtwo\NN}$.
\end{enumerate}
Now we define $\tctx_\omega$ and $\tctx'$ as follows:
\[
  \begin{array}{rcl}
  \tctx_\omega & \eqdef & \bigcup_{n \in \Nat} \tctx_n \\
  \tctx' & \eqdef & \tctx_\omega
                    \cup \set{\typ^{\pm} \ST \tctx_\omega \vdash \typ^{\pm}}
  \end{array}
\]
Note that $\tctx \subseteq \tctx_\omega \subseteq \tctx'$.
Moreover, we claim that $\tctx'$ is a prime theory:

{\bf Closure by deduction.}
  Let $\tctx' \vdash \ev$, and let us show that $\ev \in \tctx'$.
  Since all assumptions in $\tctx'$ of the form $\typ^{\pm}$ are provable
  from $\tctx_\omega$, this means that $\tctx_\omega \vdash \ev$
  by the cut rule~(\rlem{admissible_rules}).
  We consider four subcases,
  depending on whether $\ev$ is a strong/classical proof/refutation.
  We only study the positive cases; the negative cases are symmetric: 
  \begin{enumerate}
  \item {\em Strong proof, \ie $\ev = \typ\pp$.}
    Then $\tctx_\omega \vdash \typ\pp$
    so $\typ\pp \in \tctx'$ by definition of $\tctx'$.
  \item {\em Classical proof, \ie $\ev = \typ\PP$.}
    Then $\tctx_\omega \vdash \typ\PP$
    so in particular $\tctx_\omega \vdash (\typ \lor \typ)\pp$
    applying the $\Iorp[1]$ rule.
    Then there is an $n_0$ such that
    $\tctx_n \vdash (\typ \lor \typ)\pp$ for all $n \geq n_0$.
    Then it cannot be the case that $\typ\PP \notin \tctx_n$ for all $n \geq n_0$,
    because the proposition $(\typ \lor \typ)\pp$
    must be eventually treated by
    the construction of $(\tctx_n)_{n \in \Nat}$ above.
    This means that there is an $n \geq n_0$ 
    such that $\typ\PP \in \tctx_n$,
    and therefore $\typ\PP \in \tctx_\omega \subseteq \tctx'$,
    as required.
  \end{enumerate}

{\bf Consistency.}
  It suffices to note that $\tctx' \nvdash \evtwo$.
  Indeed, suppose that $\tctx' \vdash \evtwo$.
  Then $\tctx_\omega \vdash \evtwo$
  by the cut rule~(\rlem{admissible_rules}),
  so there exists an $n_0$ such that $\tctx_n \vdash \evtwo$ for all $n \geq n_0$.
  This contradicts the invariant of
  the construction of $(\tctx_n)_{n \in \Nat}$ above.

{\bf Disjunctive property.}
  We consider only the positive case. The negative case is symmetric.
  Suppose that $\tctx' \vdash (\typ \lor \typtwo)\pp$.
  Then $\tctx_\omega \vdash (\typ \lor \typtwo)\pp$
  by the cut rule~(\rlem{admissible_rules}),
  so there exists an $n_0$ such that $\tctx_n \vdash (\typ \lor \typtwo)\pp$ for all $n \geq n_0$.
  Then it cannot be the case that $\typ\PP,\typtwo\PP \notin \tctx_n$ for all $n \geq n_0$,
  because the proposition $(\typ \lor \typtwo)\pp$
  must be eventually treated by
  the construction of $(\tctx_n)_{n \in \Nat}$ above.
  This means that there is an $n \geq n_0$ 
  such that either $\typ\PP \in \tctx_n$ or $\typtwo\PP \in \tctx_n$,
  and therefore we have that
  either $\typ\PP \in \tctx_\omega \subseteq \tctx'$,
  or $\typtwo\PP \in \tctx_\omega \subseteq \tctx'$, as required.

Finally, note that $\tctx' \nvdash \evtwo$,
as has already been shown in the proof of consistency above.
\end{proof}

\begin{definition}[Canonical model]
The {\em canonical model} is the structure
$\kripmod_0 = (\wlset_0,\subseteq,\wlpos{},\wlneg{})$:
\begin{enumerate}
\item $\wlset_0$ is the set of all prime theories,
      \ie $\wlset_0 \eqdef \set{\tctx \ST \text{$\tctx$ is prime}}$.
\item $\subseteq$ is the set-theoretic inclusion between prime theories.
\item $\wlpos{\tctx} = \set{\btyp \ST \btyp\pp \in \tctx}$
      and
      $\wlneg{\tctx} = \set{\btyp \ST \btyp\nn \in \tctx}$.
\end{enumerate}
\end{definition}

\begin{lemma}
\llem{appendix:canonical_model_is_kripke}
The canonical model is a Kripke model.
\end{lemma}
\begin{proof}
Let us check the two required properties.
{\bf Monotonicity}
is immediate, since if $\tctx \subseteq \tctx'$ 
then $\btyp^{\pm} \in \tctx$ implies $\btyp^{\pm} \in \tctx'$.
For {\bf stabilization},
let $\tctx$ be a prime theory and let $\btyp$ be a propositional variable.
First note that $\tctx \cup \set{\btyp\pp}$
and $\tctx \cup \set{\btyp\nn}$ cannot both be inconsistent,
by the consistent extension lemma~(\rlem{appendix:consistent_extension}).
We consider two subcases,
depending on whether $\tctx \cup \set{\btyp\pp}$ is consistent:
\begin{enumerate}
\item {\em If $\tctx \cup \set{\btyp\pp}$ is consistent.}
  Then $\tctx,\btyp\pp \nvdash \btyp\nn$
  because $\tctx,\btyp\pp \vdash \btyp\nn$ would make
  the set $\tctx\cup\set{\btyp\pp}$ inconsistent.
  Then by saturation~(\rlem{appendix:kripke_saturation})
  there is a prime theory $\tctx' \supseteq \tctx \cup \set{\btyp\pp}$
  such that $\tctx' \nvdash \btyp\nn$.
  Hence we have that $\tctx' \supseteq \tctx$
  with $\btyp \in \wlpos{\tctx'} \setminus \wlneg{\tctx'}$.
\item {\em Otherwise, so $\tctx \cup \set{\btyp\nn}$ is consistent.}
  Similarly as in the previous case,
  we have that $\tctx,\btyp\nn \nvdash \btyp\pp$,
  so by saturation~(\rlem{appendix:kripke_saturation})
  there is a prime theory $\tctx' \supseteq \tctx \cup \set{\btyp\nn}$
  such that $\tctx' \nvdash \btyp\pp$,
  and this implies that $\btyp \in \wlneg{\tctx'} \setminus \wlpos{\tctx'}$.
\end{enumerate}
\end{proof}

\begin{lemma}[Main Semantic Lemma]
\llem{appendix:kripke_main_semantic_lemma}
Let $\tctx$ be a prime theory.
Then
$\kripforce[\kripmod_0]{\tctx}{\ev}$ holds in the canonical model
if and only if $\ev \in \tctx$.
\end{lemma}
\begin{proof}
We proceed by induction on the measure $\#(\ev)$.
We only study the positive cases, the negative cases are symmetric.
\smallskip

{\bf Propositional variable, $\ev = \btyp\pp$.}
  \[
    \kripforce[\kripmod_0]{\tctx}{\btyp\pp}
    \iff
    \btyp \in \wlpos{\tctx}
    \iff
    \btyp\pp \in \tctx
  \]

{\bf Strong conjunction, $\ev = (\typ \land \typtwo)\pp$.}
  \[
  \begin{array}{rcll}
  &&
    \kripforce[\kripmod_0]{\tctx}{(\typ \land \typtwo)\pp}
  \\
  & \iff &
    \kripforce[\kripmod_0]{\tctx}{\typ\PP} \text{ and }
    \kripforce[\kripmod_0]{\tctx}{\typtwo\PP}
  \\
  & \iff &
    \typ\PP \in \tctx \text{ and } \typtwo\PP \in \tctx
    & \text{by \ih}
  \\
  & \iff &
    (\typ \land \typtwo)\pp \in \tctx
  \end{array}
  \]
  The last equivalence uses the
  fact that $\tctx$ is closed by deduction,
  using rule $\Iandp$ for the ``only if'' direction
  and rules $\Eandp[1],\Eandp[2]$ for the ``if'' direction.
  \smallskip

{\bf Strong disjunction, $\ev = (\typ \lor \typtwo)\pp$.}
  \[
  \begin{array}{rcll}
  &&
    \kripforce[\kripmod_0]{\tctx}{(\typ \lor \typtwo)\pp}
  \\
  & \iff &
    \kripforce[\kripmod_0]{\tctx}{\typ\PP}
    \text{ or }
    \kripforce[\kripmod_0]{\tctx}{\typtwo\PP}
  \\
    & \iff &
    \typ\PP \in \tctx
    \text{ or }
    \typtwo\PP \in \tctx
    & \text{by \ih}
  \\
    & \iff &
    (\typ \lor \typtwo)\pp \in \tctx
  \end{array}
  \]
  The last equivalence uses the fact that
  $\tctx$ is a prime theory,
  using rules $\Iorp[1]$ and $\Iorp[2]$
  for the ``only if'' direction, and the fact that $\tctx$ is disjunctive
  for the ``if'' direction.
  \smallskip

{\bf Strong negation, $\ev = (\neg\typ)\pp$.}
  \[
  \begin{array}{rcll}
    \kripforce[\kripmod_0]{\tctx}{(\neg\typ)\pp}
  & \iff &
    \kripforce[\kripmod_0]{\tctx}{\typ\NN}
  \\
  & \iff &
    \typ\NN \in \tctx
    & \text{by \ih}
  \\
  & \iff &
    (\neg\typ)\pp \in \tctx
  \end{array}
  \]
  The last equivalence uses the fact that
  $\tctx$ is closed by deduction,
  using rule $\Inotp$ for the ``only if'' direction
  and rule $\Enotp$ for the ``if'' direction.
  \smallskip

{\bf Classical proposition, $\ev = \typ\PP$.}
  \[
  \begin{array}{rcll}
    \kripforce[\kripmod_0]{\tctx}{\typ\PP}
  & \iff &
    \forall \tctx' \supseteq \tctx,\ %
    \kripnotforce[\kripmod_0]{\tctx'}{\typ\nn}
  \\
  & \iff &
    \forall \tctx' \supseteq \tctx,\ %
    \typ\nn \notin \tctx'
    & \text{by \ih}
  \\
  & \iff &
    \typ\PP \in \tctx
  \end{array}
  \]
  Note that $\tctx'$ does not vary over arbitrary sets of propositions,
  but only over prime theories.
  To justify the last equivalence, we prove each implication separately:
  \begin{itemize}
  \item[$(\Rightarrow)$]
    We show the contrapositive.
    Let $\typ\PP \notin \tctx$
    and let us show that there is a prime theory $\tctx' \supseteq \tctx$
    such that $\typ\nn \in \tctx'$.
    First we claim that $\tctx\cup\set{\typ\nn}$ is consistent.
    \begin{itemize}
    \item[] {\em Proof of the claim.}
      Suppose by contradiction that $\tctx\cup\set{\typ\nn}$ is inconsistent.
      Then in particular $\tctx,\typ\nn \vdash \Bot\PP$.
      (Recall that we encode falsity
      as $\Bot \eqdef (\btyp_0 \land \neg\btyp_0)$).
      By the projection lemma~(\rlem{projection_lemma})
      we have that $\tctx,\typ\NN \vdash \Bot\PP$.
      By contraposition~(\rlem{admissible_rules_logic})
      $\tctx,\Bot\NN \vdash \typ\PP$.
      Since $\Bot\NN$ is provable~(\rexample{lem_and_noncontr}),
      by the cut rule~(\rlem{admissible_rules_logic})
      we have that $\tctx \vdash \typ\PP$.
      But $\tctx$ is closed by deduction, so $\typ\PP \in \tctx$.
      This contradicts the fact that $\typ\PP \notin \tctx$
      and concludes the proof of the claim.
    \end{itemize}
    Now since $\tctx\cup\set{\typ\nn}$ is consistent,
    by saturation~(\rlem{appendix:kripke_saturation}),
    we may extend it
    to a prime theory $\tctx' \supseteq \tctx\cup\set{\typ\nn}$.
    This concludes this case.
  \item[$(\Leftarrow)$]
    Suppose that $\typ\PP \in \tctx$, and let $\tctx' \supseteq \tctx$
    such that $\typ\nn \in \tctx'$.
    Then since $\tctx'$ is closed by deduction, using the $\Icp$ rule
    we have that $\typ\NN \in \tctx'$.
    Since $\tctx'$ contains both $\typ\PP$ and $\typ\NN$,
    using the generalized absurdity rule we may
    derive an arbitrary proposition from $\tctx'$,
    which means that $\tctx'$ is inconsistent,
    contradicting the fact that $\tctx'$ is a prime theory.
  \end{itemize}
\end{proof}


\subsection{Proof of Subject~Reduction of $\lambdaC$~(\rprop{subject_reduction})}
\lsec{appendix:subject_reduction}

\begin{proposition}[Subject reduction]
\lprop{appendix:subject_reduction}
If $\tctx \vdash \tm : \ev$
and $\tm \toa{} \tmtwo$, then $\tctx \vdash \tmtwo : \ev$.
\end{proposition}
\begin{proof}
Since reduction is closed under arbitrary contexts,
the term on the left hand side is of the form
$\gctxof{\tm_0}$ and it reduces to $\gctxof{\tm_1}$
contracting the redex $\tm_0$.
We proceed by induction on the context $\gctx$ under which the rewriting
step takes place.
The interesting case is when the context is empty.
All other cases are easy by resorting to the \ih.
We proceed by case analysis on each of the reduction rules.
Note that each rule actually stands for two rules, depending on the
instantiations of the signs.
We write only one of these cases; if the signs are flipped the proof is
symmetric.
We use the admissible typing rules
$\Cut$ and $\Abs'$~(\rlem{admissible_rules}).
\medskip

\noindent$\bullet$ $\ruleProj$:
let $i \in \set{1,2}$. We have:
    \[
    {\small
    \indrule{\Eandp}{
      \indrule{\Iandp}{
        \indrule{}{
          \typingDerivation_1
        }{
          \tctx \vdash \tm_1 : \typ_1\PP
        }
        \HS
        \indrule{}{
          \typingDerivation_2
        }{
          \tctx \vdash \tm_2 : \typ_2\PP
        }
      }{
        \tctx \vdash \pairp{\tm_1}{\tm_2} : (\typ_1 \land \typ_2)\pp
      }
    }{
        \tctx \vdash \tm_i : \typ_i\PP
    }
    }
    \]
Then:
    \[
    {\small
      \prooftree 
      \typingDerivation_i
      \justifies \tctx \vdash \tm_i : \typ\PP
      \thickness=0.05em
      \endprooftree
    }
    \]

\noindent$\bullet$ $\ruleCase$:
    let $i \in \set{1,2}$. We have:
    \[
    {\small
    \indrule{\Eorp}{
      \indrule{\Iorp}{
        \indrule{}{
          \typingDerivation
        }{
          \tctx \vdash \tm : \typ_i\PP
        }
      }{
        \tctx \vdash \inip{\tm} : (\typ_1 \lor \typ_2)\pp
      }
      \HS
      \typingDerivation_1
      \HS
      \typingDerivation_2
    }{
      \tctx \vdash \casep{(\inip{\tm})}{\var}{\tmtwo_1}{\var}{\tmtwo_2} : \ev
    }
    }
    \]
    where for each $j \in \set{1,2}$,
    the derivation $\typingDerivation_j$ is:
    \[
      \indrule{}{
        \typingDerivation'_j
      }{
        \tctx, \var : \typ_j\PP \vdash \tmtwo_j : \ev
      }
    \]
    Then:
    \[
    {\small
        \prooftree 
        \[
          \typingDerivation'_i
          \justifies \tctx, \var : \typ_i\PP \vdash \tmtwo_i : \ev
        \]
        \[
          \typingDerivation
          \justifies \tctx \vdash \tm : \typ_i\PP
        \]
        \justifies \tctx \vdash \tmtwo_i\sub{\var}{\tm} : \ev
        \using{\indrulename{\Cut}}
        \thickness=0.05em
        \endprooftree
    }
    \]

\noindent$\bullet$ $\ruleNeg$:
  We have:
    \[
    {\small
      \prooftree 
      \[
        \[
          \typingDerivation
          \justifies \tctx \vdash \tm : \typ\NN
        \]
        \justifies \tctx \vdash \negip{\tm} : (\lnot\typ)\pp
        \using{\Inotp}
        \thickness=0.05em
      \]
      \justifies \tctx \vdash \negep{(\negip{\tm})} : \typ\NN 
      \thickness=0.05em
      \using{\indrulename{\Enotp}}
      \endprooftree
    }
    \]
    Then:
    \[
    {\small
      \prooftree 
      \typingDerivation
      \justifies \tctx \vdash \tm : \typ\NN
      \thickness=0.05em
      \endprooftree
    }
    \]

\noindent$\bullet$ $\ruleBeta$:
    we have:
    \[
    {\small
      \prooftree 
      \[
        \[
          \typingDerivation
          \justifies \tctx, \var : \typ\NN \vdash \tm : \typ\pp
        \]
        \justifies \tctx \vdash \claslamp{\var}{\tm} : \typ\PP
        \using{\Icp}
        \thickness=0.05em
      \]
      \[
        \typingDerivation'
        \justifies \tctx \vdash \tmtwo : \typ\NN
      \]
      \justifies \tctx \vdash \clasapp{(\claslamp{\var}{\tm})}{\tmtwo} : \typ\pp 
      \thickness=0.05em
      \using{\indrulename{\Ecp}}
      \endprooftree
    }
    \]
    Then:
    \[
    {\small
      \prooftree 
        \[
          \typingDerivation
          \justifies \tctx, \var : \typ\NN \vdash \tm : \typ\pp
        \]
        \[
          \typingDerivation'
          \justifies \tctx \vdash \tmtwo : \typ\NN
        \]
        \justifies \tctx \vdash \tm\sub{\var}{\tmtwo} : \typ\pp
        \using{\indrulename{\Cut}}
        \thickness=0.05em
        \endprooftree
    }
    \]
\noindent$\bullet$ $\ruleAbsPairInj$:
    we have:
    \[
    {\footnotesize
      \prooftree 
      \[
        \[
          \typingDerivation_1
          \justifies \tctx \vdash \tm_1 : \typ_1\PP
        \]
        \[
          \typingDerivation_2
          \justifies \tctx \vdash \tm_2 : \typ_2\PP
        \]
        \justifies \tctx \vdash \pairp{\tm_1}{\tm_2} : (\typ_1 \land \typ_2)\pp
        \using{\Iandp}
        \thickness=0.05em
      \]
      \[
        \[
          \typingDerivation'
          \justifies \tctx \vdash \tmtwo : \typ_i\NN
        \]
        \justifies \tctx \vdash \inin{\tmtwo} : (\typ_1 \land \typ_2)\nn
        \using{\indrulename{\Iandn}}
      \]
      \justifies \tctx \vdash \strongabs{\ev}{\pairp{\tm_1}{\tm_2}}{\inin{\tmtwo}} : \ev
      \thickness=0.05em
      \using{\indrulename{\Abs}}
      \endprooftree
    }
    \]
    Then:
    \[
    {\small
      \prooftree 
        \[
          \typingDerivation_i
          \justifies \tctx \vdash \tm_i : \typ_i\PP
        \]
        \[
          \typingDerivation'
          \justifies \tctx \vdash \tmtwo : \typ_i\NN
        \]
        \justifies \tctx \vdash \abs{\ev}{\tm_i}{\tmtwo} : \ev
        \using{\indrulename{\Abs'}}
        \thickness=0.05em
        \endprooftree
    }
    \]

\noindent$\bullet$ $\ruleAbsInjPair$:
    similar to the previous case.

\noindent$\bullet$ $\ruleAbsNeg$:
    \[
    {\small
      \prooftree 
      \[
        \[
          \typingDerivation
          \justifies \tctx \vdash \tm : \typ\NN 
        \]
        \justifies \tctx \vdash \negip{\tm} : (\lnot\typ)\pp
        \using{\indrulename{\Inotp}}
      \]
      \[
        \[
          \typingDerivation'
          \justifies \tctx \vdash \tmtwo : \typ\PP 
        \]
        \justifies \tctx \vdash \negin{\tmtwo} : (\lnot\typ)\nn
        \using{\indrulename{\Inotn}}
      \]
      \justifies \tctx \vdash \strongabs{\ev}{(\negip{\tm})}{(\negin{\tmtwo})} : \ev
      \thickness=0.05em
      \using{\indrulename{\Abs}}
      \endprooftree
    }
    \]
    Then:
    \[
    {\small
      \prooftree 
      \[
        \typingDerivation
        \justifies \tctx \vdash \tm : \typ\NN 
      \]
      \[
        \typingDerivation'
        \justifies \tctx \vdash \tmtwo : \typ\PP 
      \]
      \justifies \tctx \vdash \abs{\ev}{\tm}{\tmtwo} : \ev
      \thickness=0.05em
      \using{\indrulename{\Abs'}}
      \endprooftree
    }
    \]
\end{proof}

\subsection{Proof of the Positivity Condition for~\rcoro{systemF_SN_posneg}}
\lsec{appendix:positivity_condition}

\begin{definition}
Recall that the set of type constraints $\typeConstraintsPosNeg$
is given by all equations of the following form, for all
types $\typ,\typtwo$ of System~F:
\[
  \Pos{\typ}{\typtwo}
  \equiv
  (\Neg{\typ}{\typtwo} \to \typ)
  \HS
  \Neg{\typ}{\typtwo}
  \equiv
  (\Pos{\typ}{\typtwo} \to \typtwo)
\]
\end{definition}

\begin{proposition}
The set of type constraints $\typeConstraintsPosNeg$
verifies Mendler's positivity condition
(stated in the body of the paper,
and also in the appendix in \rdef{positivity}).
\end{proposition}
\begin{proof}
Define the {\em complexity} of a type as follows:
\[
  \begin{array}{rcl@{\hspace{1cm}}rcl}
    \compl{\btyp}               & \eqdef & 1 \text{\HS if $\btyp \in \mathbf{V}$}
  \\
    \compl{\Pos{\typ}{\typtwo}}
    = \compl{\Neg{\typ}{\typtwo}}
    = \compl{\typ \to \typtwo}
    & \eqdef & 1 + \compl{\typ} + \compl{\typtwo}
  \\
    \compl{\forall\btyp.\typ}   & \eqdef & 1 + \compl{\typ}
  \end{array}
\]
Recall that $\posvars{\typ}$ (resp. $\negvars{\typ}$)
stand for the set of type variables occurring
positively (resp. negatively) in a given type $\typ$.
Moreover, the set of type variables occurring
{\em weakly positively}
(resp. {\em weakly negatively}) 
in $\typ$
are written $\wposvars{\typ}$ (resp. $\wnegvars{\typ}$)
and defined as follows:
\[
  \begin{array}{rcl}
    \wposvars{\btyp} & \eqdef & \set{\btyp}
                                \text{\HS if $\btyp \in \mathbf{V}$}
  \\
    \wposvars{\Pos{\typ}{\typtwo}} & \eqdef &
                                \set{\Pos{\typ}{\typtwo}}
                                \cup \wposvars{\typ} \cup \wnegvars{\typtwo}
  \\
    \wposvars{\Neg{\typ}{\typtwo}} & \eqdef &
                                \set{\Neg{\typ}{\typtwo}}
                                \cup \wnegvars{\typ} \cup \wposvars{\typtwo}
  \\
    \wposvars{\typ \to \typtwo} & \eqdef & \wnegvars{\typ} \cup \wposvars{\typtwo}
  \\
    \wposvars{\forall\btyp.\typ} & \eqdef & \wposvars{\typ} \setminus \set{\btyp}
  \end{array}
\]
\[
  \begin{array}{rcl}
  \\
    \wnegvars{\btyp} & \eqdef & \emptyset
                                \text{\HS if $\btyp \in \mathbf{V}$}
  \\
    \wnegvars{\Pos{\typ}{\typtwo}} & \eqdef &
                                \wnegvars{\typ} \cup \wposvars{\typtwo}
  \\
    \wnegvars{\Neg{\typ}{\typtwo}} & \eqdef &
                                \wposvars{\typ} \cup \wnegvars{\typtwo}
  \\
    \wnegvars{\typ \to \typtwo} & \eqdef & \wposvars{\typ} \cup \wnegvars{\typtwo}
  \\
    \wnegvars{\forall\btyp.\typ} & \eqdef & \wnegvars{\typ} \setminus \set{\btyp}
  \end{array}
\]
It is easy to check that $\posvars{\typ} \subseteq \wposvars{\typ}$
and $\negvars{\typ} \subseteq \wnegvars{\typ}$
by simultaneous induction on $\typ$.
It is also easy to check that
if $\btyp \in \wposvars{\typ} \cup \wnegvars{\typ}$
then $\compl{\btyp} \leq \compl{\typ}$,
by induction on $\typ$.
Moreover, let $X, Y$ be types.
A type $\typ$ is said to be {\em $(X,Y)$-positive}
if $\Pos{X}{Y} \in \wposvars{\typ}$
or $\Neg{X}{Y} \in \wnegvars{\typ}$.
Symmetrically,
a type $\typ$ is said to be {\em $(X,Y)$-negative}
if $\Pos{X}{Y} \in \wnegvars{\typ}$
or $\Neg{X}{Y} \in \wposvars{\typ}$.
It is straightforward to prove the following
{\bf invariant} for the equivalence
$\typ \equiv \typtwo$ between types induced by the recursive type constraints,
by induction on the derivation of $\typ \equiv \typtwo$.
\begin{enumerate}
\item
  If $\typ \equiv \typtwo$, then
  $\typ$ is $(X,Y)$-positive if and only if $\typtwo$ is $(X,Y)$-positive.
\item
  If $\typ \equiv \typtwo$, then
  $\typ$ is $(X,Y)$-negative if and only if $\typtwo$ is $(X,Y)$-negative.
\end{enumerate}
To prove Mendler's positivity condition, we must check that given
any type variable $\btyp$
of the form $\Pos{\typ}{\typtwo}$ or
of the form $\Neg{\typ}{\typtwo}$,
then whenever $\btyp \equiv \typthree$
we have that $\btyp$ does not occur negatively in $\typthree$.
We consider two cases, depending on whether
$\btyp = \Pos{\typ}{\typtwo}$
or
$\btyp = \Neg{\typ}{\typtwo}$:
\begin{enumerate}
\item
  Let $\Pos{\typ}{\typtwo} \equiv \typthree$
  and suppose that $\Pos{\typ}{\typtwo} \in \negvars{\typthree}$.
  Then we have that $\Pos{\typ}{\typtwo} \in \wnegvars{\typthree}$,
  so $\typthree$ is $(\typ,\typtwo)$-negative.
  By the invariant, $\Pos{\typ}{\typtwo}$ is also $(\typ,\typtwo)$-negative,
  so either $\Pos{\typ}{\typtwo} \in \wnegvars{\Pos{\typ}{\typtwo}}$
  or $\Neg{\typ}{\typtwo} \in \wposvars{\Pos{\typ}{\typtwo}}$.
  Both conditions are impossible, indeed:
  \begin{enumerate}
  \item
    Suppose that $\Pos{\typ}{\typtwo} \in \wnegvars{\Pos{\typ}{\typtwo}}$.
    Then, given that $\Pos{\typ}{\typtwo}$ does not occur weakly negatively
    at the root of $\Pos{\typ}{\typtwo}$, so it
    must occur either inside $\typ$ or inside $\typtwo$,
    so $\compl{\Pos{\typ}{\typtwo}} < \compl{\Pos{\typ}{\typtwo}}$,
    which is a contradiction.
  \item
    Suppose that $\Neg{\typ}{\typtwo} \in \wnegvars{\Pos{\typ}{\typtwo}}$.
    Then, again, $\Neg{\typ}{\typtwo}$
    must occur either inside $\typ$ or inside $\typtwo$,
    so $\compl{\Neg{\typ}{\typtwo}} < \compl{\Pos{\typ}{\typtwo}}$,
    which is a contradiction.
  \end{enumerate}
\item
  If $\Neg{\typ}{\typtwo} \equiv \typthree$ then, symmetrically as above,
  we have that $\Neg{\typ}{\typtwo} \notin \negvars{\typthree}$.
\end{enumerate}
\end{proof}

\subsection{Proof of the Simulation Lemma for the Translation from $\PRK$
            to the Extended System~F}

\begin{lemma}
\llem{appendix:semF_simulation}
If $\tm \toa{} \tmtwo$ in $\lambdaC$
then $\semF{\tm} \toa{}^+ \semF{\tmtwo}$
in System~F extended with $\typeConstraintsPosNeg$.
\end{lemma}
\begin{proof}
By case analysis on the rewriting rule used to derive
the step $\tm \toa{} \tmtwo$.
Note that showing contextual closure is immediate, so we only
study the cases in which the rewriting rule is applied at the root:

\noindent$\bullet$ $\ruleProj$:
  $
    \semF{\projipn{\pairpn{\tm_1}{\tm_2}}}
  =
    \projiF{\pairF{\semF{\tm_1}}{\semF{\tm_2}}}
  \toa{}
    \semF{\tm_i}
  $.
\medskip

\noindent$\bullet$ $\ruleCase$:
  \[
    \begin{array}{rcll}
      &&
      \semF{\casepn{\inipn{\tm}}{(\var:\ev)}{\tmtwo_1}{(\var:\evtwo)}{\tmtwo_2}}
    \\
      & = &
      \caseF{\iniF{\semF{\tm}}}{(\var:\semF{\ev})}{\semF{\tmtwo_1}}{(\var:\semF{\evtwo})}{\semF{\tmtwo_2}}
    \\
      & \to &
      \semF{\tmtwo_i}\sub{\var}{\semF{\tm}}
    \\
      & = &
      \semF{\tmtwo_i\sub{\var}{\semF{\tm}}}
      \\
      && \HS\text{by \rlem{semF_properties}}
    \end{array}
  \]
\medskip

\noindent$\bullet$ $\ruleNeg$:
  $
    \semF{\negepn{\negipn{\tm}}}
  =
    (\lam{\var^\tunit}{\semF{\tm}})\,\trivF
  \toa{}
    \semF{\tm}\sub{\var}{\trivF}
  =
    \semF{\tm}
  $
  by \rlem{semF_properties}, since $\var \not\in \fv{\tm}$
  by definition of $\semF{\negipn{\tm}}$.
\medskip

\noindent$\bullet$ $\ruleBeta$:
  $
    \semF{\clasappn{(\claslampn{(\var:\ev)}{\tm})}{\tmtwo}}
    =
    (\lam{\var^{\semF{\ev}}}{\semF{\tm}})\,\semF{\tmtwo}
    \to
    \semF{\tm}\sub{\var}{\semF{\tmtwo}}
    =
    \semF{\tm\sub{\var}{\tmtwo}}
  $
  by \rlem{semF_properties}.
\medskip

\noindent$\bullet$ $\ruleAbsPairInj$:
  we consider two subcases, depending on the signs:
  \begin{enumerate}
  \item
    Let $\vdash \tm_1 : \typ_1\PP$,
        $\vdash \tm_2 : \typ_2\PP$,
    and $\vdash \tmtwo : \typ_i\NN$ for some $i \in \set{1,2}$.
    Then:
    \[
      \begin{array}{ll}
      &
        \semF{\strongabs{\ev}{\pairp{\tm_1}{\tm_2}}{\inin{\tmtwo}}}
      \\
      = &
        \funabsF{(\typ_1 \land \typ_2)\pp}{\ev}
          \,\pairF{\semF{\tm_1}}{\semF{\tm_2}}
          \,\iniF{\semF{\tmtwo}}
      \\
      \toa{}^+ &
          \caseFtable{\iniF{\semF{\tmtwo}}}{
            (\varthree:\semF{\typ_1\NN})
          }{
            \funabsF{\typ_1\PP}{\ev}\,\projiF[1]{ \pairF{\semF{\tm_1}}{\semF{\tm_2}} }\,\varthree
          }{
            (\varthree:\semF{\typ_2\NN})
          }{
            \funabsF{\typ_2\PP}{\ev}\,\projiF[2]{ \pairF{\semF{\tm_1}}{\semF{\tm_2}} }\,\varthree
          }
      \\
      & \HS\text{by definition of $\funabsF{(\typ_1 \land \typ_2)\pp}{\ev}$}
      \\
      \toa{} &
        \funabsF{\typ_i\PP}{\ev}\,\projiF{ \pairF{\semF{\tm_1}}{\semF{\tm_2}} }\,\semF{\tmtwo}
      \\
      \toa{} &
        \funabsF{\typ_i\PP}{\ev}\,\semF{\tm_i}\,\semF{\tmtwo}
      \\
      \toa{}^+ &
        \funabsF{\typ_i\pp}{\ev}\,(\semF{\tm_i}\,\semF{\tmtwo})
                                  (\semF{\tmtwo}\,\semF{\tm_i})
        \\&\HS\text{by definition of $\funabsF{\typ_i\PP}{\ev}$}
      \\
      = &
        \semF{\strongabs{\ev}{(\clasapp{\tm_i}{\tmtwo})}{(\clasapn{\tm_i}{\tmtwo})}}
      \\
      = &
        \semF{\abs{\ev}{\tm_i}{\tmtwo}}
      \end{array}
    \]
  \item
    Let $\tctx \vdash \tm_1 : \typ_1\NN$,
        $\tctx \vdash \tm_2 : \typ_2\NN$,
    and $\tctx \vdash \tmtwo : \typ_i\PP$ for some $i \in \set{1,2}$.
    Then, symmetrically as for the previous case,
    $
      \semF{\strongabs{\ev}{\pairn{\tm_1}{\tm_2}}{\inip{\tmtwo}}}
      \toa{}^+
      \semF{\abs{\ev}{\tm_i}{\tmtwo}}
    $.
  \end{enumerate}
\medskip

\noindent$\bullet$ $\ruleAbsInjPair$:
  symmetric to the previous case.
\medskip

\noindent$\bullet$ $\ruleAbsNeg$:
  we consider two subcases, depending on the signs:
  \begin{enumerate}
  \item
    Let $\tctx \vdash \tm : \typ\NN$ and $\tctx \vdash \tmtwo : \typ\PP$.
    Then:
    \[
      \begin{array}{rcll}
      &&
        \semF{\strongabs{\ev}{(\negip{\tm})}{(\negin{\tmtwo})}}
      \\
      & = &
        \funabsF{(\neg\typ)\pp}{\ev}
        (\lam{\var^{\tunit}}{\semF{\tm}})
        (\lam{\vartwo^{\tunit}}{\semF{\tmtwo}})
        \\&& \HS\text{where $\var \not\in \fv{\tm}$,
                      $\vartwo \not\in \fv{\tmtwo}$}
      \\
      & \toa{}^+ &
        \funabsF{\typ\NN}{\ev}
        ((\lam{\var^{\tunit}}{\semF{\tm}})\,\trivF)
        ((\lam{\vartwo^{\tunit}}{\semF{\tmtwo}})\,\trivF)
        \\&& \HS\text{by definition of $\funabsF{(\neg\typ)\pp}{\ev}$}
      \\
      & \toa{}^+ &
        \funabsF{\typ\NN}{\ev}\,\semF{\tm}\,\semF{\tmtwo}
      \\
      & \toa{}^+ &
        \funabsF{\typ\nn}{\ev}\,(\semF{\tm}\,\semF{\tmtwo})\,(\semF{\tmtwo}\,\semF{\tm})
        \\&& \HS\text{by definition of $\funabsF{\typ\NN}{\ev}$}
      \\
      & = &
        \semF{\strongabs{\ev}{(\clasapn{\tm}{\tmtwo})}{(\clasapp{\tmtwo}{\tm})}}
      \\
      & = &
        \semF{\abs{\ev}{\tm}{\tmtwo}}
      \end{array}
    \]
  \item
    Let $\tctx \vdash \tm : \typ\PP$ and $\tctx \vdash \tmtwo : \typ\NN$.
    Then, symmetrically as for
    the previous case:
    $
      \semF{\strongabs{\ev}{(\negin{\tm})}{(\negip{\tmtwo})}}
      \toa{}^+
      \semF{\abs{\ev}{\tm}{\tmtwo}}
    $.
  \end{enumerate}
\end{proof}

\subsection{Proof of Characterization of Normal Forms~(\rprop{characterization_of_normal_terms})}
\lsec{appendix:characterization_of_normal_forms}

\begin{proposition}
\lprop{appendix:characterization_of_normal_terms}
A term is normal if and only if it does not reduce in $\lambdaC$.
\end{proposition}
\begin{proof}
$(\Rightarrow)$
Let $\tm$ be a normal term,
and let us check that it is a $\toa{}$-normal form.
We proceed by induction on the derivation that $\tm$ is a normal term.

The cases corresponding to introduction rules are straightforward by \ih.
For example, if $\tm = \pairpn{\nf_1}{\nf_2}$,
then by \ih $\nf_1$ and $\nf_2$ have no $\toa{}$-redexes.
Moreover, there are no rules involving a pair $\pairpn{-}{-}$ at the
root, so $\pairpn{\nf_1}{\nf_2}$ is in $\toa{}$-normal form.

The cases corresponding to elimination rules and the absurdity rule
are also straightforward by \ih,
observing that there cannot be a redex at the root.
For example, if $\tm = \projipn{\neu}$, then by \ih
$\neu$ has no $\toa{}$-redexes.
Moreover, the only rule involving a projection $\projipn{-}$ at the
root is $\ruleProj$,
which would require that $\neu = \pairpn{\tm_1}{\tm_2}$.
But this is impossible ---as can be checked by exhaustive case analysis on
$\neu$---, so $\tm$ is in $\toa{}$-normal form. 
\medskip

\noindent $(\Leftarrow)$
  Let $\tm$ be a $\toa{}$-normal form,
  let us check that it is a normal term.
  We proceed by induction on the structure of the term $\tm$:
  \begin{enumerate}
  \item {\em Variable, $\var$:} it is a neutral term.
  \item {\em Absurdity, $\strongabs{\ev}{\tm}{\tmtwo}$:}
    by \ih, $\tm$ and $\tmtwo$ are normal terms.
    If either $\tm$ or $\tmtwo$ is a neutral term, we are done.
    We are left to analyze the case in which they are not neutral terms,
    \ie both $\tm$ and $\tmtwo$ are built using introduction rules.
    Note that the types of $\tm$ and $\tmtwo$ are $\evtwo$ and $\evtwo\OP$
    respectively, for some strong type $\evtwo$.
    We proceed by case analysis on the form of the proposition $\evtwo$.
    There are four cases:
    \begin{enumerate}
    \item {\em Proof/refutation of a propositional variable, $\evtwo = \btyp^\pm$.}
      This case is impossible, since $\tm$ only may be of one of the following forms:
      $\pairpn{\nf}{\nf}$, $\inipn{\nf}$, $\claslampn{\var : \ev}{\nf}$, or
      $\negipn{\nf}$, none of which are of type $\btyp^\pm$.
    \item {\em Proof of a conjunction, $\evtwo = (\typ\land\typtwo)\pp$
               or refutation of a disjunction $\evtwo = (\typ\lor\typtwo)\nn$.}
      Then $\tm$ is of the form $\pairpn{\tm_1}{\tm_2}$
      and $\tmtwo$ is of the form $\ininp{\tmtwo'}$ for some $i \in \set{1,2}$,
      so the rule $\ruleAbsPairInj$ may be applied at the root, contradicting
      the hypothesis that the term is $\to{}$-normal.
    \item {\em Disjunction, $\evtwo = (\typ\land\typtwo)^\pm$.}
      Then $\tm$ is of the form $\inipn{\tmtwo'}$ for some $i \in \set{1,2}$
      and $\tmtwo$ is of the form $\pairnp{\tm_1}{\tm_2}$,
      so the rule $\ruleAbsInjPair$ may be applied at the root, contradicting
      the hypothesis that the term is $\to{}$-normal.
    \item {\em Negation, $\evtwo = (\neg\typ)^\pm$.}
      Then $\tm$ is of the form $\negipn{\tm'}$
      and $\tmtwo$ is of the form $\neginp{\tmtwo'}$,
      so the rule $\ruleAbsNeg$ may be applied at the root, contradicting
      the hypothesis that the term is $\to{}$-normal.
    \end{enumerate}
  \item {\em Pair, $\pairpn{\tm}{\tmtwo}$:}
    by \ih, $\tm$ and $\tmtwo$ are normal terms, so $\pairpn{\tm}{\tmtwo}$
    is also a normal term.
  \item {\em Projection, $\projipn{\tm}$:}
    by \ih, $\tm$ is a normal term.
    It suffices to show that $\tm$ is neutral.
    Indeed, if $\tm$ is a normal but not neutral term,
    then since the type of $\tm$ may be
    either of the form $(\typ\land\typtwo)\pp$ or of the form $(\typ\lor\typtwo)\nn$,
    we have that $\tm$ is of the form $\pairpn{\tmtwo}{\tmthree}$.
    Then the rule $\ruleProj$ may be applied at the root, contradicting
    the hypothesis that the term is $\to{}$-normal.
  \item {\em Injection, $\inipn{\tm}$:}
    by \ih, $\tm$ is a normal term, so $\inipn{\tm}$ is also normal.
  \item {\em Case, $\casepn{\tm}{\var}{\tmtwo}{\var}{\tmthree}$:}
    by \ih $\tm$, $\tmtwo$ and $\tmthree$ are normal terms.
    It suffices to show that $\tm$ is neutral.
    Indeed, if $\tm$ is a normal but not neutral term,
    then since the type of $\tm$ may be
    either of the form $(\typ\lor\typtwo)\pp$ or of the form $(\typ\land\typtwo)\nn$,
    we have that $\tm$ is of the form $\inipn{\tm'}$ for some $i \in \set{1,2}$.
    Then the rule $\ruleCase$ may be applied at the root, contradicting
    the hypothesis that the term is $\to{}$-normal.
  \item {\em Negation introduction, $\negipn{\tm}$:}
    by \ih, $\tm$ is a normal term. Then $\negipn{\tm}$ is also normal.
  \item {\em Negation elimination, $\negepn{\tm}$:}
    by \ih, $\tm$ is a normal term.
    It suffices to show that $\tm$ is neutral.
    Indeed, if $\tm$ is a normal but not neutral term,
    then since the type of $\tm$ is of the form $(\neg\typ)^\pm$,
    then $\tm$ is of the form $\negipn{\tm'}$.
    Then the rule $\ruleNeg$ may be applied at the root, contradicting
    the hypothesis that the term is $\to{}$-normal.
  \item {\em Classical introduction, $\claslampn{\var : \ev}{\tm}$:}
    by \ih, $\tm$ is a normal term, so $\claslampn{\var : \ev}{\tm}$
    is also normal.
  \item {\em Classical elimination, $\clasappn{\tm}{\tmtwo}$:}
    by \ih, $\tm$ and $\tmtwo$ are normal terms.
    It suffices to show that $\tm$ is neutral.
    Indeed, if $\tm$ is a normal but not neutral term,
    then since the type of $\tm$ may be
    either of the form $\typ\PP$ or of the form $\typ\NN$,
    we have that $\tm$ is of the form $\claslampn{\var}{\tm'}$.
    Then the rule $\ruleBeta$ may be applied at the root, contradicting
    the hypothesis that the term is $\to{}$-normal.
  \end{enumerate}
\end{proof}

\subsection{Proof of Canonicity~(\rthm{canonicity})}
\lsec{appendix:canonicity}

We give a slightly different statement of
Canonicity, adding the additional
hypothesis that $\tm$ is already a normal form.
This addition comes at no loss of generality,
given that $\lambdaC$ enjoys subject reduction~(\rprop{subject_reduction})
and strong normalization~(\rthm{lambdaC_canonical}).

\begin{theorem}[Canonicity]
\lthm{appendix:canonicity}
\quad
\begin{enumerate}
\item
  Let $\vdash \tm : \ev$ where $\tm$ is a normal form.
  Then $\tm$ is canonical.
\item
  Let $\tctx \vdash \tm : \typ^\pm$
  where $\tctx$ is classical and $\tm$ is a normal form.
  Then either $\tm$ is canonical
  or $\tm$ is of the form $\casectxof{\tm'}$
  where $\casectx$ is a case-context
  and $\tm'$ is an open explosion.
\item
  Let $\tctx \vdash \tm : \typ\PP$
  or $\tctx \vdash \tm : \typ\NN$,
  where $\tctx$ is classical and $\tm$ is a normal form.
  Then either $\tm = \claslampn{\var}{\tm'}$
  or $\tm = \elctxof{\tm'}$,
  where $\elctx$ is an eliminative context
  and $\tm'$ is a variable or an open explosion.
\end{enumerate}
\end{theorem}
\begin{proof}
\quad
\begin{enumerate}
\item
  Let $\vdash \tm : \ev$ where $\tm$ is a normal form.
  Note,
  by induction on the formation rules for neutral terms~(\rdef{normal_terms})
  that a neutral term must have at least one free variable.
  But $\tm$ is typed in the empty typing context, so it must be closed.
  Hence $\tm$ is not a neutral term, so by
  \rprop{characterization_of_normal_terms}, it must be canonical.
\item
  Let $\tctx \vdash \tm : \ev$ where $\tctx$ is classical
  and $\tm$ is a normal form.
  By \rprop{characterization_of_normal_terms} either $\tm$ is canonical
  or it is a neutral term. If $\tm$ is canonical we are done.
  If $\tm$ is a neutral term it suffices to show the following claim,
  namely that if $\tctx \vdash \tm : \typtwo^\pm$ is a derivable
  judgment such that $\tctx$ is classical
  and $\tm$ is a neutral term,
  then $\tm$ is of the form $\tm = \casectxof{\tm'}$,
  where $\casectx$ is a case-context
  and $\tm'$ is an open explosion.
  We proceed by induction on the formation rules
  for neutral terms~(\rdef{normal_terms}):
  \begin{enumerate}
  \item {\em Variable, $\tm = \var$.}
    this case is impossible, given that $\tctx$ is assumed to be classical,
    so $\tctx \vdash \var : \ev$
    where $\ev$ must be of the form $\typthree\PP$ or $\typthree\NN$,
    hence $\ev$ cannot be of the form $\typtwo^\pm$.
  \item {\em Projection, $\projipn{\neu}$:}
    this case is impossible, as $\tctx \vdash \projipn{\neu} : \ev$
    where $\ev$ must be of the form $\typthree\PP$ or $\typthree\NN$,
    hence $\ev$ cannot be of the form $\typtwo^\pm$.
  \item {\em Case, $\casepn{\neu}{\var}{\nf_1}{\var}{\nf_2}$:}
    by inversion of the typing rules we have that
    either $\tctx \vdash \neu : (\typ \lor \typtwo)\pp$
    or $\tctx \vdash \neu : (\typ \land \typtwo)\nn$.
    In both cases we may apply the \ih to conclude that
    $\neu$ is of the form
    $\neu = \casectxof{\tm'}$ where $\casectx$ is a case-context
    and $\tm'$ is an open explosion.
    Therefore $\tm = \casepn{(\casectxof{\tm'})}{\var}{\nf_1}{\var}{\nf_2}$
    where now
    $\casepn{(\casectx)}{\var}{\nf_1}{\var}{\nf_2}$
    is a case-context.
  \item {\em Classical elimination, $\clasappn{\neu}{\nf}$:}
    then $\tm$ is an explosion
    under the empty case-context.
    Moreover, $\neu$ must have at least one free variable
    so $\tm$ is indeed an open explosion.
  \item {\em Negation elimination, $\negepn{\neu}$:}
    this case is impossible, as $\tctx \vdash \negepn{\neu} : \ev$
    where $\ev$ must be of the form $\typthree\PP$ or $\typthree\NN$,
    hence $\ev$ cannot be of the form $\typtwo^\pm$.
  \item {\em Absurdity, $\strongabs{}{\neu}{\nf}$ or $\strongabs{}{\nf}{\neu}$:}
    then $\tm$ is an explosion
    under the empty case-context.
    Moreover, $\neu$ must have at least one free variable
    so $\tm$ is indeed an open explosion.
  \end{enumerate}
\item
  Let $\tctx \vdash \tm : \typ\PP$
  or $\tctx \vdash \tm : \typ\NN$,
  where $\tctx$ is classical and $\tm$ is a normal form.
  By \rprop{characterization_of_normal_terms} either $\tm$ is canonical
  or it is a neutral term.
  If $\tm$ is canonical, then by the constraints on its type it
  must be of the form $\tm = \claslampn{\var}{\tm'}$, so we are done.
  If $\tm$ is neutral, it suffices to show the following claim
  namely that if $\tctx \vdash \tm : \ev$
  is a derivable judgment, with $\ev \in \set{\typtwo\PP,\typtwo\NN}$,
  such that $\tctx$ is classical
  and $\tm$ is a neutral term,
  then $\tm$ is of the form $\tm = \elctxof{\tm'}$,
  where $\elctx$ is an eliminative context
  and $\tm'$ is a variable or an open explosion.
  We proceed by induction on the formation rules
  for neutral terms~(\rdef{normal_terms}):
  \begin{enumerate}
  \item {\em Variable, $\tm = \var$.}
    immediate, as $\tm$ is a variable under the empty eliminative context.
  \item {\em Projection, $\projipn{\neu}$:}
    by inversion of the typing rules, we have that
    either $\tctx \vdash \neu : (\typ \land \typtwo)\pp$
    or $\tctx \vdash \neu : (\typ \lor \typtwo)\nn$.
    In both cases we may apply the second item of this lemma
    to conclude that $\neu$ is of the form $\neu = \casectxof{\tm'}$
    where $\casectx$ is a case-context and $\tm'$ is an open explosion.
    Therefore $\tm = \projipn{\casectxof{\tm'}}$,
    where now $\projipn{\casectx}$ is an eliminative context.
  \item {\em Case, $\casepn{\neu}{\var}{\nf_1}{\var}{\nf_2}$:}
    by inversion of the typing rules, we have that
    either $\tctx \vdash \neu : (\typ \lor \typtwo)\pp$
    or $\tctx \vdash \neu : (\typ \land \typtwo)\nn$.
    In both cases we may apply the second item of
    this lemma to conclude that
    $\neu$ is of the form $\neu = \casectxof{\tm'}$
    where $\casectx$ is an eliminative context
    and $\tm'$ is an open explosion.
    Therefore
    $\tm = \casepn{(\casectxof{\tm'})}{\var}{\nf_1}{\var}{\nf_2}$,
    where now $\casepn{(\casectx)}{\var}{\nf_1}{\var}{\nf_2}$
    is an eliminative context.
  \item {\em Classical elimination, $\clasappn{\neu}{\nf}$:}
    then $\tm$ is an explosion
    under the empty eliminative context.
    Moreover, $\neu$ must have at least one free variable
    so $\tm$ is indeed an open explosion.
  \item {\em Negation elimination, $\negepn{\neu}$:}
    by inversion of the typing rules, we have that
    $\tctx \vdash \neu : (\neg\typ)^\pm$.
    By the second item of this lemma,
    $\neu$ is of the form $\neu = \casectxof{\tm'}$
    where $\casectx$ is a case-context and $\tm'$ is an open explosion.
    Therefore $\tm = \negepn{\casectxof{\tm'}}$,
    where now $\negepn{\casectx}$ is an eliminative context.
  \item {\em Absurdity, $\strongabs{}{\neu}{\nf}$ or $\strongabs{}{\nf}{\neu}$:}
    then $\tm$ is an explosion
    under the empty eliminative context.
    Moreover, $\neu$ must have at least one free variable
    so $\tm$ is indeed an open explosion.
  \end{enumerate}
\end{enumerate}
\end{proof}

\subsection{Proof that $\lambdaCeta$ is Strongly Normalizing and Confluent~(\rthm{lambdaCeta_canonical})}
\lsec{appendix:extensionality}

\begin{lemma}[Local confluence]
\llem{lambdaCeta_local_confluence}
The $\lambdaCeta$-calculus has the weak Church--Rosser property.
\end{lemma}
\begin{proof}
Let $\tm_0 \toa{} \tm_1$ and $\tm_0 \toa{} \tm_2$,
and let us show that the diagram can be closed,
\ie that there is a term $\tm_3$ such that
$\tm_1 \rto \tm_3$ and $\tm_2 \rto \tm_3$.
The proof is by induction on $\tm_0$ and by case analysis on the
relative positions of the steps
$\tm_0 \toa{} \tm_1$ and $\tm_0 \toa{} \tm_2$.
Most cases are straightforward by resorting to the \ih.
We study only the interesting cases,
when the patterns of the redexes overlap.
There are two such cases:
\begin{enumerate}
\item $\ruleBeta$/$\ruleEta$:
  Let $\var \notin \fv{\tm}$.
  The overlap involves a step
  $\clasappn{(\claslampn{\var}{\clasappn{\tm}{\var}})}{\tmtwo}
   \toa{\ruleBeta}
   \clasappn{\tm}{\tmtwo}$
  and a step
  $\clasappn{(\claslampn{\var}{\clasappn{\tm}{\var}})}{\tmtwo}
   \toa{\ruleEta}
   \clasappn{\tm}{\tmtwo}$,
  so the diagram is trivially closed in zero rewriting steps.
\item $\ruleEta$/$\ruleBeta$:
  Let $\var \notin \fv{\tm}$.
  The overlap involves a step
  $\claslampn{\var}{\clasappn{(\claslampn{\vartwo}{\tm})}{\var}}
   \toa{\ruleEta} \claslampn{\vartwo}{\tm}$
  and a step
  $\claslampn{\var}{\clasappn{(\claslampn{\vartwo}{\tm})}{\var}}
   \toa{\ruleBeta}
   \claslampn{\var}{\tm\sub{\vartwo}{\var}}$.
  Note that the targets of the steps are $\alpha$-equivalent,
  so the diagram is trivially closed in zero rewriting steps.
\end{enumerate}
\end{proof}

\begin{lemma}[Properties of reduction in $\lambdaCeta$]
\llem{lambdaCeta_properties}
\quad
\begin{enumerate}
\item
  {\em Reduction does not create free variables.}
  If $\tm \to \tm'$ then $\fv{\tm} \supseteq \fv{\tm'}$.
\item
  {\em Substitution (I).}
  Let $\tctx,\var:\typ \vdash \tm : \typtwo$
  and $\tctx \vdash \tmtwo : \typ$.
  If $\tm \to \tm'$
  then $\tm\sub{\var}{\tmtwo} \to \tm'\sub{\var}{\tmtwo}$.
\item
  {\em Substitution (II).}
  Let $\tctx,\var:\typ \vdash \tm : \typtwo$
  and $\tctx \vdash \tmtwo : \typ$.
  If $\tmtwo \to \tmtwo'$
  then $\tm\sub{\var}{\tmtwo} \rto \tm\sub{\var}{\tmtwo'}$.
\item
  {\em Substitution (III).}
  Let $\tctx,\var:\typ \vdash \tm : \typtwo$
  and $\tctx \vdash \tmtwo : \typ$.
  If $\tm \rto \tm'$ and $\tmtwo \rto \tmtwo'$
  then $\tm\sub{\var}{\tmtwo} \rto \tm'\sub{\var}{\tmtwo'}$.
\end{enumerate}
\end{lemma}
\begin{proof}
Items 1., 2., and 3. are by induction on $\tm$.
Item 4. is by induction on the sum of the lengths of the sequences
$\tm \rto \tm'$ and $\tmtwo \rto \tmtwo'$,
resorting to the two previous items.
\end{proof}

\begin{lemma}[Postponement of $\ruleEta$ steps]
\llem{lambdaCeta_postponement}
Let $\tm \toa{\ruleEta} \tmtwo \toa{\ruleAnon} \tmthree$
where $\ruleAnon$ is a rewriting rule other than $\ruleEta$.
Then there exists a term $\tmtwo'$ such that
$\tm \ptoa{\ruleAnon} \tmtwo' \rtoa{\ruleEta} \tmthree$.
\end{lemma}
\begin{proof}
By induction on $\tm$.
If the $\ruleEta$ step and the $\ruleAnon$ step
are not reduction steps at the root, it is immediate to conclude,
resorting to the \ih when appropriate.

If the $\ruleEta$ step is at the root,
then the first step is of the form
$\tm = \claslampn{\var}{(\clasappn{\tmtwo}{\var})}
 \toa{\ruleEta} \tmtwo$,
where $\var \notin \fv{\tmtwo}$.
Taking $\tmtwo' := \claslampn{\var}{(\clasappn{\tmthree}{\var})}$
we have that
$\tm = \claslampn{\var}{(\clasappn{\tmtwo}{\var})}
 \toa{\ruleAnon} \claslampn{\var}{(\clasappn{\tmthree}{\var})}
 \toa{\ruleEta} \tmthree$,
so we are done.
For the last reduction step, we use
the fact that reduction does not create
free variables~(\rlem{lambdaCeta_properties}).

Otherwise, we have that the $\ruleEta$ step is {\em not} at the root
and the $\ruleAnon$ step is at the root.
Then we proceed by case analysis, depending on the kind of rule applied.
We only study the positive cases (the negative cases are symmetric):
\begin{enumerate}
\item $\ruleProj$:
  then we have that
  $\tm \toa{\ruleEta} \tmtwo = \projip{\pairp{\tmtwo_1}{\tmtwo_2}}
       \toa{\ruleProj} \tmtwo_i$.
  Recall that the $\ruleEta$ step is not at the root of $\tm$.
  Moreover, it cannot be the case that $\tm = \projip{\tm'}$
  and the $\ruleEta$ step is at the root of $\tm'$,
  because the type of $\tm'$ must be of the form $(\typ \land \typtwo)\pp$
  but the $\ruleEta$ rule can only be applied on a term constructed with a
  $\claslampn{-}{-}$, whose type is classical.
  This means that $\tm$ must be of the form $\projip{\pairp{\tm_1}{\tm_2}}$
  and that the $\ruleEta$ step is
  either internal to $\tm_1$
  or internal to $\tm_2$,
  which
  implies that $\tm_1 \rtoa{\ruleEta} \tmtwo_1$
  and $\tm_2 \rtoa{\ruleEta} \tmtwo_2$.
  Taking $\tmtwo' := \tm_i$
  we have that
  $\tm = \projip{\pairp{\tm_1}{\tm_2}}
    \toa{\ruleProj} \tm_i
    \rtoa{\ruleEta} \tmtwo_i$,
  as required.
\item $\ruleCase$:
  then we have that
  $\tm
   \toa{\ruleEta} \tmtwo
   = \casep{\inip{\tmtwo_0}}{\vartwo}{\tmtwo_1}{\vartwo}{\tmtwo_2}
   \toa{\ruleCase} \tmtwo_i\sub{\vartwo}{\tmtwo_0}$.
  Recall that the $\ruleEta$ step is not at the root of $\tm$.
  Moreover, it cannot be the case that
  $\tm = \casep{\tm'}{\vartwo}{\tmtwo_1}{\vartwo}{\tmtwo_2}$
  and the $\ruleEta$ step is at the root of $\tm'$,
  because the type of $\tm'$ must be of the form $(\typ \lor \typtwo)\pp$,
  but the $\ruleEta$ rule can only be applied on a term constructed with a
  $\claslampn{-}{-}$, whose type is classical.
  This means that $\tm$ must be of the form
  $\casep{\inip{\tm_0}}{\vartwo}{\tm_1}{\vartwo}{\tm_2}$
  and that the $\ruleEta$-step is
  either internal to $\tm_0$, or internal to $\tm_1$, or internal to $\tm_2$,
  which implies that
  $\tm_0 \rtoa{\ruleEta} \tmtwo_0$
  and $\tm_1 \rtoa{\ruleEta} \tmtwo_1$
  and $\tm_2 \rtoa{\ruleEta} \tmtwo_2$.
  Taking $\tmtwo' := \tm_i\sub{\vartwo}{\tm_0}$
  we have that
  $\tm = \casep{\inip{\tm_0}}{\vartwo}{\tm_1}{\vartwo}{\tm_2}
   \toa{\ruleCase} \tm_i\sub{\vartwo}{\tm_0}
   \rtoa{\ruleEta} \tmtwo_i\sub{\vartwo}{\tmtwo_0}$
  resorting to~\rlem{lambdaCeta_properties} for the last step.
\item $\ruleNeg$:
  then we have that
  $\tm \toa{\ruleEta} \negep{(\negip{\tmtwo_1})} \toa{\ruleNeg} \tmtwo_1$.
  Recall that the $\ruleEta$-reduction step is not at the root of $\tm$.
  Moreover, it cannot be the case that $\tm = \negep{\tm'}$
  and the $\ruleEta$-reduction step is at the root of $\tm'$,
  because the type of $\tm'$ must be of the form $(\neg\typ)\pp$
  but the $\ruleEta$ rule can only be applied on a term constructed with a
  $\claslampn{-}{-}$, whose type is classical.
  This means that $\tm$ must be of the form
  $\negep{(\negip{\tm_1})}$
  and that the $\ruleEta$ step is internal to $\tm_1$,
  \ie $\tm_1 \toa{\ruleEta} \tmtwo_1$.
  Then taking $\tmtwo' := \tm_1$
  we have that
  $\tm = \negep{(\negip{\tm_1})} \toa{\ruleNeg} \tm_1 \toa{\ruleEta} \tmtwo_1$
  as required.
\item $\ruleBeta$:
  then we have that
  $\tm
   \toa{\ruleEta} \clasapp{(\claslamp{\vartwo}{\tmtwo_1})}{\tmtwo_2}
   \toa{\ruleBeta} \tmtwo_1\sub{\vartwo}{\tmtwo_2}$.
  Recall that the $\ruleEta$ step is not at the root of $\tm$.
  There are three cases, depending on the position of the $\ruleEta$-step:
  \begin{enumerate}
  \item
    {\em Immediately to the left of the application.}
    That is, $\tm = \clasapp{\tm'}{\tmtwo_2}$
    and the $\ruleEta$ step is at the root of $\tm'$,
    \ie $\tm' \toa{\ruleEta} \claslamp{\vartwo}{\tmtwo_1}$
    is a reduction step at the root.
    Then
    $\tm' = \claslamp{\var}{(\clasapp{(\claslamp{\vartwo}{\tmtwo_1})}{\var})}$.
    Hence taking $\tmtwo' := \tmtwo_1\sub{\vartwo}{\tmtwo_2}$
    we have that
    \[
      \begin{array}{cl}
      &
      \tm = \clasapp{(\claslamp{\var}{(\clasapp{(\claslamp{\vartwo}{\tmtwo_1})}{\var})})}{\tmtwo_2}
      \\
      \toa{\ruleBeta} &
      \clasapp{(\claslamp{\vartwo}{\tmtwo_1})}{\tmtwo_2}
      \\
      \toa{\ruleBeta} &
      \tmtwo_1\sub{\vartwo}{\tmtwo_2}
      \end{array}
    \]
    using two $\ruleBeta$ steps and no $\ruleEta$ steps.
  \item
    {\em Inside the abstraction.}
    That is, $\tm = \clasapp{(\claslamp{\vartwo}{\tm_1})}{\tmtwo_2}$
    with $\tm_1 \toa{\ruleEta} \tmtwo_1$.
    Then taking $\tmtwo' := \tm_1\sub{\vartwo}{\tmtwo_2}$
    we have that
    $\tm = \clasapp{(\claslamp{\vartwo}{\tm_1})}{\tmtwo_2}
     \toa{\ruleBeta} \tm_1\sub{\vartwo}{\tmtwo_2}
     \toa{\ruleEta} \tmtwo_1\sub{\vartwo}{\tmtwo_2}$
    resorting to~\rlem{lambdaCeta_properties} for the last step.
  \item
    {\em To the right of the application.}
    That is, $\tm = \clasapp{(\claslamp{\vartwo}{\tmtwo_1})}{\tm_2}$
    with $\tm_2 \toa{\ruleEta} \tmtwo_2$.
    Then taking $\tmtwo' := \tmtwo_1\sub{\vartwo}{\tm_2}$
    we have that
    $\tm = \clasapp{(\claslamp{\vartwo}{\tmtwo_1})}{\tm_2}
     \toa{\ruleBeta} \tmtwo_1\sub{\vartwo}{\tm_2}
     \rtoa{\ruleEta} \tmtwo_1\sub{\vartwo}{\tmtwo_2}$
    resorting to~\rlem{lambdaCeta_properties} for the last step.
  \end{enumerate}
\item $\ruleAbsPairInj$:
  then we have that
  $\tm \toa{\ruleEta} \strongabs{}{\pairp{\tmtwo_1}{\tmtwo_2}}{\inin{\tmtwo_3}}
       \toa{\ruleAbsPairInj} \abs{}{\tmtwo_i}{\tmtwo_3}$.
  Recall that the $\ruleEta$ step is not at the root of $\tm$.
  Moreover, it cannot be the case that
  $\tm = \strongabs{}{\tm'}{\inin{\tmtwo_3}}$
  and the $\ruleEta$ step is at the root of $\tm'$,
  because the type of $\tm'$ must be of the form $(\typ \land \typtwo)\pp$,
  but the $\ruleEta$ rule can only be applied on a term constructed with a
  $\claslampn{-}{-}$, whose type is classical.
  For similar reasons, 
  it cannot be the case that $\tm = \strongabs{}{\pairp{\tmtwo_1}{\tmtwo_2}}{\tm'}$
  with the $\ruleEta$ step is at the root of $\tm'$,
  because then the type of $\tm'$ must be of the form $(\typ \land \typtwo)\nn$.
  This means that $\tm$ must be of the form
  $\strongabs{}{\pairp{\tm_1}{\tm_2}}{\inin{\tm_3}}$
  and that the $\ruleEta$ step is
  either internal to $\tm_1$, or internal to $\tm_2$, or internal to $\tm_3$.
  This implies that
  $\tm_1 \rtoa{\ruleEta} \tmtwo_1$
  and $\tm_2 \rtoa{\ruleEta} \tmtwo_2$
  and $\tm_3 \rtoa{\ruleEta} \tmtwo_3$.
  Taking $\tmtwo' := \abs{}{\tm_i}{\tm_3}$
  we have that
    $\tm = \strongabs{}{\pairp{\tm_1}{\tm_2}}{\inin{\tm_3}}
     \toa{\ruleAbsPairInj} \abs{}{\tm_i}{\tm_3}
     = \strongabs{}{(\clasapp{\tm_i}{\tm_3})}{(\clasapn{\tm_3}{\tm_i})}
     \rtoa{\ruleEta}
     \strongabs{}{(\clasapp{\tmtwo_i}{\tmtwo_3})}{(\clasapn{\tmtwo_3}{\tmtwo_i})}
     = \abs{}{\tmtwo_i}{\tmtwo_3}$.
\item $\ruleAbsInjPair$:
  Symmetric to the previous case.
\item $\ruleAbsNeg$:
  then we have that
  $\tm
   \toa{\ruleEta} \strongabs{}{(\negip{\tmtwo_1})}{(\negin{\tmtwo_2})}
   \toa{\ruleAbsNeg} \abs{}{\tmtwo_1}{\tmtwo_2}$.
  Recall that the $\ruleEta$ step is not at the root of $\tm$.
  Moreover, it cannot be the case that
  $\tm = \strongabs{}{\tm'}{(\negin{\tmtwo_2})}$
  and the $\ruleEta$ step is at the root of $\tm'$,
  because the type of $\tm'$ must be of the form $(\neg\typ)\pp$,
  but the $\ruleEta$ rule can only be applied on a term constructed with a
  $\claslampn{-}{-}$, whose type is classical.
  For similar reasons,
  it cannot be the case that
  $\tm = \strongabs{}{\negip{\tmtwo_1}}{\tm'}$
  with the $\ruleEta$ step is at the root of $\tm'$,
  because then the type of $\tm'$ must be of the form $(\neg\typ)\nn$.
  This means that $\tm$ must be of the form
  $\strongabs{}{(\negip{\tm_1})}{(\negin{\tm_2})}$ 
  and that the $\ruleEta$ step is either internal to $\tm_1$
  or internal to $\tm_2$.
  This implies that $\tm_1 \rtoa{\ruleEta} \tmtwo_1$
  and $\tm_2 \rtoa{\ruleEta} \tmtwo_2$.
  Taking $\tmtwo' := \abs{}{\tm_1}{\tm_2}$
  we have that
  $\tm = \strongabs{}{(\negip{\tm_1})}{(\negin{\tm_2})}
   \toa{\ruleAbsNeg} \abs{}{\tm_1}{\tm_2}
   = \strongabs{}{(\clasapp{\tm_1}{\tm_2})}{(\clasapp{\tm_2}{\tm_1})}
   \rtoa{\ruleEta}
   \strongabs{}{(\clasapp{\tmtwo_1}{\tmtwo_2})}{(\clasapp{\tmtwo_2}{\tmtwo_1})}
   = \abs{}{\tmtwo_1}{\tmtwo_2}$.
\end{enumerate}
\end{proof}

\begin{theorem}
\lthm{appendix:lambdaCeta_canonical}
The $\lambdaCeta$-calculus is strongly normalizing and confluent.
\end{theorem}
\begin{proof}
Strong normalization follows from
postponement of the $\ruleEta$ rule~(\rlem{lambdaCeta_postponement})
and strong normalization
of the calculus without $\ruleEta$~(\rthm{lambdaC_canonical})
by the usual rewriting techniques.

More precisely, let us write $\toa{\neg\ruleEta}$ for reduction
not using $\ruleEta$, that is,
$\toa{\neg\ruleEta} \eqdef (\toa{\,} \setminus \toa{\ruleEta})$.
Suppose there is an infinite reduction sequence
$\tm_1 \to \tm_2 \to \tm_3 \hdots$ in $\lambdaCeta$.
Let $\tm_1 \rtoa{\neg\ruleEta} \tm_i$
be the longest prefix of the sequence whose
steps are not $\ruleEta$ steps.
This prefix cannot be infinite given that $\lambdaC$ is strongly normalizing.
Let $\tm_i \rtoa{\ruleEta} \tm_{i+n}$ be the longest sequence of $\ruleEta$
steps starting on $\tm_i$. This sequence cannot be infinite given that
an $\ruleEta$ step decreases the size of the term.
Now there must be a step $\tm_{i+n} \toa{\neg\ruleEta} \tm_{i+n+1}$.
Applying the postponement lemma~(\rlem{lambdaCeta_postponement})
$n$ times, we obtain an infinite sequence of the form
$\tm_1 \rtoa{\neg\ruleEta} \tm_i \toa{\neg\ruleEta} \tm'_{i+1} \hdots$.
By repeatedly applying this argument, we may build an infinite sequence
of $\toa{\neg\ruleEta}$ steps, contradicting the fact that
$\lambdaC$ is strongly normalizing.

Confluence of $\lambdaCeta$ follows from the fact that it is
strongly normalizing and locally confluent~(\rlem{lambdaCeta_local_confluence}),
resorting to Newman's Lemma~\cite[Theorem~1.2.1]{terese}.
\end{proof}

\subsection{Computation Rules for the Embedding of Classical Logic into $\PRK$}
\lsec{classical_simulation}

The statements of all of the following lemmas are in $\lambdaCeta$
(with $\ruleEta$ reduction).
\medskip

\subsubsection{Simulation of conjunction}

\begin{definition}[Conjunction introduction]
Let $\tctx \vdash \tm : \typ\PP$ and $\tctx \vdash \tmtwo : \typtwo\PP$.
Then $\tctx \vdash \pairc{\tm}{\tmtwo} : (\typ \land \typtwo)\PP$
where:
\[
  \pairc{\tm}{\tmtwo} \eqdef
  \claslamp{(\under:(\typ\land\typtwo)\NN)}{
    \pairp{\tm}{\tmtwo}
  }
\]
\end{definition}

\begin{definition}[Conjunction elimination]
Let $\tctx \vdash \tm : (\typ_1 \land \typ_2)\PP$.
Then $\tctx \vdash \projic{\tm} : \typ_i\PP$ where:
\[
  \projic{\tm} \eqdef
  \claslamp{(\var:\typ_i\NN)}{
    \clasapp{
      \projip{
        \clasapp{
          \tm
        }{
          \claslamn{(\under:(\typ_1 \land \typ_2)\PP)}{\inin{\var}}
        }
      }
    }{
      \var
    }
  }
\]
\end{definition}

\begin{lemma}
  $\projic{\pairc{\tm_1}{\tm_2}} \rto \tm_i$
\end{lemma}
\begin{proof}
  \[
  {\small
  \begin{array}{rcll}
  &&
    \projic{\pairc{\tm_1}{\tm_2}}
  \\
  & = &
    \claslamp{\var:\typ_i\NN}{
      \clasapp{
        \projip{
          \clasapp{
            (\claslamp{\under}{\pairp{\tm_1}{\tm_2}})
          }{
            \claslamn{\under}{\inin{\var}}
          }
        }
      }{
        \var
      }
    }
  \\
  & \toa{\ruleBeta} &
    \claslamp{\var:\typ_i\NN}{
      \clasapp{
        \projip{\pairp{\tm_1}{\tm_2}}
      }{
        \var
      }
    }
  \\
  & \toa{\ruleProj} &
    \claslamp{\var:\typ_i\NN}{
      \clasapp{
        \tm_i
      }{
        \var
      }
    }
  \\
  & \toa{\ruleEta} &
    \tm_i
  \end{array}
  }
  \]
\end{proof}

\subsubsection{Simulation of disjunction}

\begin{definition}[Disjunction introduction]
Let $\tctx \vdash \tm : \typ_i\PP$.
Then $\tctx \vdash \inic{\tm} : (\typ_1 \lor \typ_2)\PP$ where:
\[
  \inic{\tm} \eqdef
  \claslamp{(\under:(\typ_1\lor\typ_2)\NN)}{
    \inip{\tm}
  }
\]
\end{definition}

\begin{definition}[Disjunction elimination]
Let
$\tctx \vdash \tm : (\typ \lor \typtwo)\PP$
and $\tctx, \var : \typ\PP \vdash \tmtwo : \typthree\PP$
and $\tctx, \var : \typtwo\PP \vdash \tmthree : \typthree\PP$.
Then $\tctx \vdash
      \casec{\tm}{(\var:\typ\PP)}{\tmtwo}{(\var:\typtwo\PP)}{\tmthree}
      : \typthree\PP$,
where:
\[
  \claslamp{(\vartwo:\typthree\NN)}{
    \caseptablex{
      (\clasapp{
        \tm
      }{
        \claslamn{(\under:(\typ\lor\typtwo)\PP)}{
          \pairn{
            \contrapose{\var}{\vartwo}{
              \tmtwo
            }
          }{
            \contrapose{\var}{\vartwo}{
              \tmthree
            }
          }
        }
      })
    }{
      (\var : \typ\PP)
    }{
      \clasapp{
        \tmtwo
      }{
        \vartwo
      }
    }{
      (\var : \typtwo\PP)
    }{
      \clasapp{
        \tmthree
      }{
        \vartwo
      }
    }
  }
\]
\end{definition}

\begin{lemma}
  $\casec{\inic{\tm_i}}{\var}{\tmtwo_1}{\var}{\tmtwo_2}
   \rto \tmtwo_i\sub{\var}{\tm}$
\end{lemma}
\begin{proof}
  \[
  {\small
  \begin{array}{rcll}
  &&
    \casec{\inic{\tm_i}}{\var}{\tmtwo_1}{\var}{\tmtwo_2}
  \\
  & = &
    \claslamptable{(\vartwo:\typthree\NN)}{
      \caseptablex{
        \clasapptable{
          \claslamp{(\under:(\typ_1\lor\typ_2)\NN)}{
            \inip{\tm}
          }
        }{
          \claslamn{(\under:(\typ\lor\typtwo)\PP)}{
            \pairn{
              \contrapose{\var}{\vartwo}{
                \tmtwo_1
              }
            }{
              \contrapose{\var}{\vartwo}{
                \tmtwo_2
              }
            }
          }
        }
      }{
        (\var : \typ\PP)
      }{
        \clasapp{
          \tmtwo_1
        }{
          \vartwo
        }
      }{
        (\var : \typtwo\PP)
      }{
        \clasapp{
          \tmtwo_2
        }{
          \vartwo
        }
      }
    }
  \\
  & \toa{\ruleBeta} &
    \claslamp{(\vartwo:\typthree\NN)}{
      \casep{
        \inip{\tm}
      }{
        (\var : \typ\PP)
      }{
        \clasapp{
          \tmtwo_1
        }{
          \vartwo
        }
      }{
        (\var : \typtwo\PP)
      }{
        \clasapp{
          \tmtwo_2
        }{
          \vartwo
        }
      }
    }
  \\
  & \toa{\ruleCase} &
    \claslamp{(\vartwo:\typthree\NN)}{
      \clasapp{\tmtwo_i\sub{\var}{\tm}}{\vartwo}
    }
  \\
  & \toa{\ruleEta} &
    \tmtwo_i\sub{\var}{\tm}
  \end{array}
  }
  \]
\end{proof}

\subsubsection{Simulation of negation}
\lsec{appendix:simulation_of_negation}

\begin{definition}[Negation introduction]
By \rlem{lem_and_noncontr} we have that
$\tctx \vdash \lemN{\btyp_0} : (\btyp_0 \land \neg\btyp_0)\NN$,
that is $\tctx \vdash \lemN{\btyp_0} : \Bot\NN$.
Moreover, suppose that $\tctx, \var:\typ\PP \vdash \tm : \Bot\PP$.
Then $\tctx \vdash \neglamc{(\var:\typ\PP)}{\tm} : (\neg\typ)\PP$,
where:
\[
  \neglamc{(\var:\typ\PP)}{\tm} \eqdef
  \claslamp{(\under:(\neg\typ)\NN)}{
    \negip{
      \claslamn{(\var:\typ\PP)}{
        (\abs{
          \typ\nn
        }{
          \tm
        }{
          \lemN{\btyp_0}
        })
      }
    }
  }
\]
\end{definition}

\begin{definition}[Negation elimination]
Let $\tctx \vdash \tm : (\neg\typ)\PP$
and $\tctx \vdash \tmtwo : \typ\PP$.
Then $\tctx \vdash \negapc{\tm}{\tmtwo} : \Bot\PP$,
where:
\[
  \negapc{\tm}{\tmtwo} \eqdef
  \abs{
    \Bot\PP
  }{
    \tm
  }{
    \claslamn{(\under:(\neg\typ)\PP)}{
      \negin{
        \tmtwo
      }
    }
  }
\]
\end{definition}

\begin{lemma}
  $\negapc{(\neglamc{\var}{\tm})}{\tmtwo} \rto 
    \strongabs{}{
      (\abs{}{
        \tm\sub{\var}{\tmtwo}
      }{
        \lemN{\btyp_0}
      })
    }{
      (\clasapn{
        \tmtwo
      }{
        (\claslamn{\var}{
          (\abs{}{
            \tm
          }{
            \lemN{\btyp_0}
          })
        })
      })
    }$
\end{lemma}
\begin{proof}
\[
{\small
\begin{array}{rl}
  & \negapc{(\neglamc{\var}{\tm})}{\tmtwo}
\\
\\
= & \abs{
      \Bot\PP
    }{
      (\claslamp{\under}{
        \negip{
          \claslamn{\var}{
            (\abs{}{
              \tm
            }{
              \lemN{\btyp_0}
            })
          }
        }
      })
    }{
      \claslamn{\under}{
        \negin{
          \tmtwo
        }
      }
    }
\\
\\
= & \strongabstable{}{
      (\clasapp{
        (\claslamp{\under}{
          \negip{
            \claslamn{\var}{
              (\abs{}{
                \tm
              }{
                \lemN{\btyp_0}
              })
            }
          }
        })
      }{
        \claslamn{\under}{
          \negin{
            \tmtwo
          }
        }
      })
    }{
      (\clasapn{
        \claslamn{\under}{
          \negin{
            \tmtwo
          }
        }
      }{
        (\claslamp{\under}{
          \negip{
            \claslamn{\var}{
              (\abs{}{
                \tm
              }{
                \lemN{\btyp_0}
              })
            }
          }
        })
      })
    }
\\
\\
\toa{\ruleBeta}(2) &
  \strongabs{}{
      (\negip{
        \claslamn{\var}{
          (\abs{}{
            \tm
          }{
            \lemN{\btyp_0}
          })
        }
      })
    }{
      (\negin{\tmtwo})
    }
\\
\\
\toa{\ruleAbsNeg} &
  \strongabs{}{
    (\clasapp{
      (\claslamn{\var}{
        (\abs{}{
          \tm
        }{
          \lemN{\btyp_0}
        })
      })
    }{
      \tmtwo
    })
  }{
    (\clasapn{
      \tmtwo
    }{
      (\claslamn{\var}{
        (\abs{}{
          \tm
        }{
          \lemN{\btyp_0}
        })
      })
    })
  }
\\
\\
\toa{\ruleBeta} &
  \strongabs{}{
    (\abs{}{
      \tm\sub{\var}{\tmtwo}
    }{
      \lemN{\btyp_0}
    })
  }{
    (\clasapn{
      \tmtwo
    }{
      (\claslamn{\var}{
        (\abs{}{
          \tm
        }{
          \lemN{\btyp_0}
        })
      })
    })
  }
\end{array}
}
\]
\end{proof}

\subsubsection{Simulation of implication}

Define implication $\typ \IMP \typtwo$ as an abbreviation of
$\neg\typ \lor \typtwo$.
\medskip

\begin{definition}[Implication introduction]
If $\tctx,\var:\typ\PP \vdash \tm : \typtwo\PP$
then $\tctx \vdash \lamc{(\var:\typ)}{\tm} : (\typ \IMP \typtwo)\PP$
where:
\[
{\small
\begin{array}{rcl}
  \lamc{\var}{\tm} & \eqdef &
    \claslamp{(\vartwo:(\typ\IMP\typtwo)\NN)}{
      \inip[2]{
        \tm\sub{\var}{\mathbf{X}_{\vartwo}}
      }
    }
\\
\mathbf{X}_\vartwo & \eqdef &
  \claslamp{(\varthree:\typ\NN)}{
    \clasapp{
      (\negen{(
        \clasapn{
          \mathbf{X'}_{\vartwo,\varthree}
        }{
          \claslamp{(\under:(\neg\typ)\NN)}{
            \negip{
              \varthree
            }
          }
        }
      )})
    }{
      \varthree
    }
  }
\\
\mathbf{X'}_{\vartwo,\varthree} & \eqdef &
  \projip[1]{
    \clasapn{
      \vartwo
    }{
      \claslamp{(\under:(\typ\IMP\typtwo)\NN)}{
        \inip[1]{
          \claslamp{(\under:(\neg\typ)\NN)}{
            \negip{
              \varthree
            }
          }
        }
      }
    }
  }
\end{array}
}
\]
\end{definition}

\medskip
\begin{definition}[Implication elimination]
If $\tctx \vdash \tm : (\typ \IMP \typtwo)\PP$
and $\tctx \vdash \tmtwo : \typ\PP$,
then $\tctx \vdash \appc{\tm}{\tmtwo} : \typtwo\PP$,
where:
\[
{\small
\begin{array}{r@{}c@{}l}
  \appc{\tm}{\tmtwo}
  & \eqdef &
  \claslamptable{(\var:\typtwo\NN)}{
    \caseptablex{
      (\clasapp{
        \tm
      }{
        \claslamn{(\under:(\typ\imp\typtwo)\PP)}{
          \pairn{
            (\claslamn{(\under:(\neg\typ)\PP)}{
              \negin{
                \tmtwo
              }
            })
          }{
            \var
          }
        }
      })
    }{
      (\vartwo:(\neg\typ)\PP)
    }{
      \abs{\typtwo\pp}{\tmtwo}{
        \negen{
          (\clasapp{
            \vartwo
          }{
            \claslamn{(\under:(\neg\typ)\PP)}{
              \negin{
                \var
              }
            }
          })
        }
      }
    }{
      (\varthree:\typtwo\PP)
    }{
      \clasapp{\varthree}{\var}
    }
  } 
\end{array}
}
\]
\end{definition}
\medskip

The following lemma is the computational rule for implication:

\begin{lemma}
$\appc{(\lamc{\var}{\tm})}{\tmtwo} \rto \tm\sub{\var}{\tmtwo}$
\end{lemma}
\begin{proof}

First let
$\tmthree =
      \claslamn{\under}{
        \pairn{
          (\claslamn{\under}{
            \negin{
              \tmtwo
            }
          })
        }{
          \var'
        }
      }$
and note that:
\[
{\small
  \begin{array}{r@{\ }c@{\ }l}
  &&
    \mathbf{X'}_{\tmthree,\varthree}
  \\
  & = &
    \projip[1]{
      \clasapn{
        \tmthree
      }{
        \claslamp{(\under:(\typ\IMP\typtwo)\NN)}{
          \inip[1]{
            \claslamp{(\under:(\neg\typ)\NN)}{
              \negip{
                \varthree
              }
            }
          }
        }
      }
    }
  \\
  & \toa{\ruleBeta} &
    \projip[1]{
      \pairn{
        (\claslamn{\under}{
          \negin{
            \tmtwo
          }
        })
      }{
        \var'
      }
    }
  \\
  & \toa{\ruleProj} &
    \claslamn{\under}{
      \negin{
        \tmtwo
      }
    }
  \end{array}
}
\]
Hence:
\[
{\small
  \begin{array}{r@{\ }c@{\ }l}
  &&
    \mathbf{X}_{\tmthree}
  \\
  & = &
    \claslamp{(\varthree:\typ\NN)}{
      \clasapp{
        (\negen{(
          \clasapn{
            \mathbf{X'}_{\tmthree,\varthree}
          }{
            \claslamp{(\under:(\neg\typ)\NN)}{
              \negip{
                \varthree
              }
            }
          }
        )})
      }{
        \varthree
      }
    }
  \\
  & \rto &
    \claslamp{(\varthree:\typ\NN)}{
      \clasapp{
        (\negen{(
          \clasapn{
            (\claslamn{\under}{
              \negin{
                \tmtwo
              }
            })
          }{
            \claslamp{(\under:(\neg\typ)\NN)}{
              \negip{
                \varthree
              }
            }
          }
        )})
      }{
        \varthree
      }
    }
  \\
  & \toa{\ruleBeta} &
    \claslamp{(\varthree:\typ\NN)}{
      \clasapp{
        (\negen{(
              \negin{
                \tmtwo
              }
        )})
      }{
        \varthree
      }
    }
  \\
  & \toa{\ruleNeg} &
    \claslamp{(\varthree:\typ\NN)}{
      \clasapp{
        \tmtwo
      }{
        \varthree
      }
    }
  \\
  & \toa{\ruleEta} &
    \tmtwo
  \end{array}
}
\]
Hence:
\[
{\small
  \begin{array}{r@{\ }c@{\ }l}
  & &
  \appc{(\lamc{\var}{\tm})}{\tmtwo}
  \\
  & = &
  \claslamptable{\var'}{
    \caseptablex{
      \clasapptable{
        (\claslamp{\vartwo}{
            \inip[2]{
              \tm\sub{\var}{\mathbf{X}_{\vartwo}}
            }
          })
      }{
        \claslamn{\under}{
          \pairn{
            (\claslamn{\under}{
              \negin{
                \tmtwo
              }
            })
          }{
            \var'
          }
        }
      }
    }{
      \vartwo'
    }{
      \abs{\typtwo\pp}{\tmtwo}{
        \negen{
          (\clasapp{
            \vartwo'
          }{
            \claslamn{\under}{
              \negin{
                \var'
              }
            }
          })
        }
      }
    }{
      \varthree'
    }{
      \clasapp{\varthree'}{\var'}
    }
  } 
  \\
  & \toa{\ruleBeta} &
  \claslamptable{\var'}{
    \caseptablex{
      \inip[2]{
        \tm\sub{\var}{\mathbf{X}_{\tmthree}}
      }
    }{
      \vartwo'
    }{
      \abs{\typtwo\pp}{\tmtwo}{
        \negen{
          (\clasapp{
            \vartwo'
          }{
            \claslamn{\under}{
              \negin{
                \var'
              }
            }
          })
        }
      }
    }{
      \varthree'
    }{
      \clasapp{\varthree'}{\var'}
    }
  } 
  \\
  & \toa{\ruleCase} &
  \claslamp{\var'}{
    \clasapp{
      \tm\sub{\var}{\mathbf{X}_{\tmthree}}
    }{\var'}
  } 
  \\
  & \rto &
  \claslamp{\var'}{
    \clasapp{
      \tm\sub{\var}{\tmtwo}
    }{\var'}
  } 
  \\
  & \toa{\ruleEta} &
    \tm\sub{\var}{\tmtwo}
  \end{array}
}
\]
\end{proof}

\subsubsection{Computational content of the law of excluded middle}

\begin{lemma}
\quad
\[
{\small
  \begin{array}{ll}
  &
  \casec{\lemC{\typ}}{\var}{\tmtwo_1}{\var}{\tmtwo_2}
  \\
  \rto
  &
  \claslamp{\vartwo}{
    \clasapp{
      \tmtwo_2\sub{\var}{
        \claslamp{\under}{
          \negip{
            (\claslamn{\var}{
              \abs{}{
                \tmtwo_1
              }{ 
                \vartwo
              }
            })
          }
        }
      }
    }{
      \vartwo
    }
  }
  \end{array}
}
\]
\end{lemma}
\begin{proof}
Recall that $\lemC{\typ} = \lemP{\typ}$, where:
\[
  {\footnotesize
    \begin{array}{r@{\,}c@{\,}l}
    \lemP{\typ} & \eqdef &
      \claslamp{(\var:(\typ\lor\neg\typ)\NN)}{
        \inip[2]{
          \claslamp{(\vartwo:\neg\typ\NN)}{
            \negip{
              \projin[1]{
                \clasapn{
                  \var
                }{
                  \lemPinner{\vartwo}{\typ}
                }
              }
            }
          }
        }
      }
    \\
    \lemPinner{\vartwo}{\typ} & \eqdef &
      \claslamp{(\under:(\typ\lor\neg\typ)\NN)}{
        \inip[1]{
          \claslamp{(\varthree:\typ\NN)}{
            (\abs{
              \typ\pp
            }{
              \vartwo
            }{
              \claslamp{(\under:\neg\typ\NN)}{
                \negip{
                  \varthree
                }
              }
            })
          }
        }
      }
    \end{array}
  }
\]
Let $\tmthree = \claslamn{\under}{
                  \pairn{
                    \contrapose{\var'}{\vartwo'}{
                      \tmtwo_1
                    }
                  }{
                    \contrapose{\var'}{\vartwo'}{
                      \tmtwo_2
                    }
                  }
                }$.
Then:
  \[
  {\small
  \begin{array}{r@{\ }c@{\ }l}
  &&
    \casec{\lemC{\typ}}{\var'}{\tmtwo_1}{\var'}{\tmtwo_2}
  \\
  & = &
    \claslamptable{\vartwo'}{
      \caseptablex{
        \clasapptable{
          \claslamp{\var}{
            \inip[2]{
              \claslamp{\vartwo}{
                \negip{
                  \projin[1]{
                    \clasapn{
                      \var
                    }{
                      \lemPinner{\vartwo}{\typ}
                    }
                  }
                }
              }
            }
          }
        }{
          \tmthree
        }
      }{
        \var'
      }{
        \clasapp{
          \tmtwo_1
        }{
          \vartwo'
        }
      }{
        \var'
      }{
        \clasapp{
          \tmtwo_2
        }{
          \vartwo'
        }
      }
    }
  \\
  & \toa{\ruleBeta} &
    \claslamptable{\vartwo'}{
      \caseptablex{
         \inip[2]{
           \claslamp{\vartwo}{
             \negip{
               \projin[1]{
                 \clasapn{
                   \tmthree
                 }{
                   \lemPinner{\vartwo}{\typ}
                 }
               }
             }
           }
         }
      }{
        \var'
      }{
        \clasapp{
          \tmtwo_1
        }{
          \vartwo'
        }
      }{
        \var'
      }{
        \clasapp{
          \tmtwo_2
        }{
          \vartwo'
        }
      }
    }
  \\
  & \toa{\ruleCase} &
    \claslamp{\vartwo'}{
      \clasapp{
        \tmtwo_2\sub{\var'}{
          \claslamp{\vartwo}{
            \negip{
              \projin[1]{
                \clasapn{
                  \tmthree
                }{
                  \lemPinner{\vartwo}{\typ}
                }
              }
            }
          }
        }
      }{
        \vartwo'
      }
    }
  \\
  & \toa{\ruleBeta} &
    \claslamp{\vartwo'}{
      \clasapp{
        \tmtwo_2\sub{\var'}{
          \claslamp{\under}{
            \negip{
              \projin[1]{
                \pairn{
                  \contrapose{\var'}{\vartwo'}{
                    \tmtwo_1
                  }
                }{
                  \contrapose{\var'}{\vartwo'}{
                    \tmtwo_2
                  }
                }
              }
            }
          }
        }
      }{
        \vartwo'
      }
    }
  \\
  & \toa{\ruleProj} &
    \claslamp{\vartwo'}{
      \clasapp{
        \tmtwo_2\sub{\var'}{
          \claslamp{\under}{
            \negip{
              \contrapose{\var'}{\vartwo'}{
                \tmtwo_1
              }
            }
          }
        }
      }{
        \vartwo'
      }
    }
  \\
  & = &
    \claslamp{\vartwo'}{
      \clasapp{
        \tmtwo_2\sub{\var'}{
          \claslamp{\under}{
            \negip{
              (\claslamn{\var'}{
                \abs{}{
                  \tmtwo_1
                }{ 
                  \vartwo'
                }
              })
            }
          }
        }
      }{
        \vartwo'
      }
    }
  \\
  \end{array}
  }
  \]
\end{proof}

\newpage
\section{Formal Systems}

\subsection{System~F Extended with Recursive Type Constraints}
\lsec{appendix:system_f}

\begin{definition}[System~F$\extwith{\typeConstraints}$]
The set of {\em types} is given by:
\[
  \typ,\typtwo,\hdots ::= \btyp \mid \typ \to \typtwo \mid \forall\btyp.\typ
\]
The set of {\em terms} is given by:
\[
  \tm,\tmtwo,\hdots
      ::= \var
          \mid \lam{\var^\typ}{\tm}  \mid \tm\,\tmtwo
          \mid \lam{\btyp}{\tm} \mid \tm\,\typ
\]
we omit type annotations over variables when clear from the context.
A {\em type constraint} is an equation of the form $\btyp \equiv \typ$.
Each set $\typeConstraints$ of type constraints
induces a notion of equivalence between types,
written $\typ \equiv \typtwo$ and defined as the
congruence generated by $\typeConstraints$.
More precisely:
\[
{\small
  \indrule{constr}{
    (\typ \equiv \typtwo) \in \typeConstraints
  }{
    \typ \equiv \typtwo
  }
  \indrule{refl}{}{
    \typ \equiv \typ
  }
  \indrule{sym}{
    \typ \equiv \typtwo
  }{
    \typtwo \equiv \typ
  }
}
\]
\[
{\small
  \indrule{trans}{
    \typ \equiv \typtwo
    \HS
    \typtwo \equiv \typthree
  }{
    \typ \equiv \typthree
  }
  \indrule{cong}{
    \typ \equiv \typtwo
  }{
    \typthree\sub{\btyp}{\typ} \equiv \typthree\sub{\btyp}{\typtwo}
  }
}
\]
We suppose that $\typeConstraints$ is fixed.
Typing judgments are of the form $\tctx \vdash \tm : \typ$.
\[
{\small
  \indrule{\rulename{Ax}}{}{
    \tctx,\var:\typ \vdash \var : \typ
  }
  \indrule{Conv}{
    \tctx \vdash \tm : \typ
    \HS
    \typ \equiv \typtwo
  }{
    \tctx \vdash \tm : \typtwo
  }
}
\]
\[
{\small
  \indrule{\rulename{I$\rightarrow$}}{
    \tctx,\var:\typ \vdash \tm : \typtwo
  }{
    \tctx \vdash \lam{\var^\typ}{\tm} : \typ \to \typtwo
  }
  \indrule{E$\rightarrow$}{
    \tctx \vdash \tm : \typ \to \typtwo
    \HS
    \tctx \vdash \tmtwo : \typ
  }{
    \tctx \vdash \tm\,\tmtwo : \typtwo
  }
}
\]
\[
{\small
  \indrule{I$\forall$}{
    \tctx \vdash \tm : \typ
    \HS
    \btyp \notin \fv{\tctx}
  }{
    \tctx \vdash \lam{\btyp}{\tm} : \forall\btyp.\typ
  }
  \indrule{E$\forall$}{
    \tctx \vdash \tm : \forall\btyp.\typ
  }{
    \tctx \vdash \tm\,\typtwo : \typ\sub{\btyp}{\typtwo}
  }
}
\]
Reduction is defined as the closure by arbitrary contexts of the following
rewriting rules:
\[
  \begin{array}{rcl}
    (\lam{\var}{\tm})\,\tmtwo & \to & \tm\sub{\var}{\tmtwo} \\
    (\lam{\btyp}{\tm})\,\typ  & \to & \tm\sub{\btyp}{\typ} \\
  \end{array}
\]
\end{definition}

\begin{definition}[Positive/negative occurrences]
The set of type variables occurring positively (resp. negatively) 
in a type $\typ$
are written $\posvars{\typ}$ (resp. $\negvars{\typ}$) and defined by:
\[
{\small
  \begin{array}{r@{\ }c@{\ }l@{\hspace{.5cm}}r@{\ }c@{\ }l}
    \posvars{\btyp} & \eqdef & \set{\btyp}
  &
    \negvars{\btyp} & \eqdef & \emptyset
  \\
    \posvars{\typ \to \typtwo} & \eqdef & \negvars{\typ} \cup \posvars{\typtwo}
  &
    \negvars{\typ \to \typtwo} & \eqdef & \posvars{\typ} \cup \negvars{\typtwo}
  \\
    \posvars{\forall\btyp.\typ} & \eqdef & \posvars{\typ} \setminus \set{\btyp}
  &
    \negvars{\forall\btyp.\typ} & \eqdef & \negvars{\typ} \setminus \set{\btyp}
  \end{array}
}
\]
\end{definition}

\begin{definition}[Positivity condition]
\ldef{positivity}
A set of type constraints $\typeConstraints$
verifies the {\em positivity condition} if
for every type constraint $(\btyp \equiv \typ) \in \typeConstraints$
and every type $\typtwo$ such that $\btyp \equiv \typtwo$
one has that $\btyp \not\in \negvars{\typtwo}$.
\end{definition}

\begin{theorem}[Mendler]
\lthm{appendix:systemF_SN_Mendler}
If $\typeConstraints$ verifies the positivity condition,
then System~F$\extwith{\typeConstraints}$
is strongly normalizing.
\end{theorem}
\begin{proof}
See~\cite[Theorem~13]{mendler1991inductive}.
\end{proof}

{\bf Abbreviations.}
We define the following standard abbreviations for types:
\[
  \begin{array}{rcl}
    \tunit              & \eqdef & \forall\btyp.(\btyp \to \btyp) \\
    \tzero              & \eqdef & \forall\btyp.\btyp \\
    \neg\typ            & \eqdef & \typ \to \tzero \\
    \typ \times \typtwo & \eqdef &
      \forall\btyp.((\typ \to \typtwo \to \btyp) \to \btyp) \\
    \typ + \typtwo & \eqdef &
      \forall\btyp.((\typ \to \btyp) \to (\typtwo \to \btyp) \to \btyp) \\
  \end{array}
\]
And the following terms.
We omit the typing contexts for succintness:
\[
{\small
  \begin{array}{rcllll}
  \trivF
    & \eqdef
    & \lam{\btyp}{\lam{\var^\btyp}{\var}}
    \\
    & : & \tunit
  \\
  \\
  \abortF{\typ}{\tm}
    & \eqdef
    & \tm\,\typ
    \\
    & :
    & \typ
    \\&&\text{if $\tm : \tzero$}
  \\
  \\
  \pairF{\tm}{\tmtwo}
    & \eqdef
    & \lam{\btyp}{\lam{f^{\typ \to \typtwo \to \btyp}}{f\,\tm\,\tmtwo}}
    \\
    & :
    & \typ \times \typtwo
    \\&&\text{if $\tm : \typ$ and $\tmtwo : \typtwo$}
  \\
  \\
  \projiF{\tm}
    & \eqdef
    & \tm\,\typ_i\,(\lam{x_1^{\typ_1}}{\lam{x_2^{\typ_2}}{x_i}})
    \\
    & :
    & \typ_i
    \\&&\text{if $\tm : \typ_1 \times \typ_2$}
  \\
  \\
  \iniF{\tm}
    & \eqdef
    & \lam{\btyp}{
        \lam{f_1^{\typ_1 \to \btyp}}{\lam{f_2^{\typ_2 \to \btyp}}{f_i\,\tm}}
       }
    \\
    & :
    & \typ_1 + \typ_2
    \\&&\text{if $\tm : \typ_i$ and $i \in \set{1,2}$}
  \\
  \\
  \caseF{\tm}{\var:\typ_1}{\tmtwo_1}{\var:\typ_2}{\tmtwo_2}
    & \eqdef
    & \tm\,\typtwo
         \,(\lam{\var^{\typ_1}}{\tmtwo_1})
         \,(\lam{\var^{\typ_2}}{\tmtwo_2})
    \\
    & :
    & \typtwo
    \\
    &&
    \text{if $\tm : \typ_1 + \typ_2$ and
             $\tmtwo_i : \typtwo$}
    \\
    &&
    \text{for each $i \in \set{1,2}$}\hspace{-6cm}
  \end{array}
}
\]

\end{alphasection}

\end{document}